\documentclass[letterpaper, 11pt, reqno]{amsart}
\usepackage[foot]{amsaddr}
\usepackage[T1]{fontenc}

\usepackage{macros}
\usepackage{fullpage}
\let\textHungarianAccent\H
\usepackage{notations}
\AtBeginEnvironment{thebibliography}{\let\H\textHungarianAccent}

\title{Fast Almost-Uniform Sampling of Random \(k\)-SAT Solutions}

\author{Kun He, Zhidan Li, Kuan Yang}
\address[Kun He]{DEKE Lab, Renmin University of China, Beijing, China. \textnormal{E-mail: \url{hekun2023@ruc.edu.cn}}}
\address[Zhidan Li, Kuan Yang]{School of Computer Science, Shanghai Jiao Tong University, Shanghai, China. \textnormal{Emails: \url{yueEnTeRnAL@sjtu.edu.cn}, \url{kuan.yang@sjtu.edu.cn}}.}

\def\DensityCondition{{\alpha \le \frac{2^k}{k^{16}}}}
\newcommand{\defeq}{\triangleq}

\numberwithin{equation}{section}

\begin{document}

\begin{abstract}
We study approximately uniform sampling of satisfying assignments from random $k$-SAT formulas. Fix a constant $B>0$. For every $k\ge 2^{20}$ and density $0<\alpha\le 2^k/k^{16}$, we prove that, with high probability over a random formula on $n$ variables, every accuracy $n^{-B}\le\varepsilon\le 1/2$ admits a sampler whose output distribution is within total variation distance $\varepsilon$ of the uniform distribution on satisfying assignments and whose expected running time is
$$
    \widetilde O_B\!\left(\operatorname{poly}(k,\alpha+1)n^{2+\eta}\right),
    \qquad 0<\eta\le 2^{16-k}.
$$
In particular, the exponent of $n$ is $2+O(2^{-k})$, while the dependence on the width and density is polynomial with a universal degree. This improves the parameter-dependent polynomial exponent of the sampling algorithm of Chen, Lonkar, Wang, Yang, and Yin (STOC 2025) \cite{CLWYY25} in the density regime $2^k/\operatorname{poly}(k)$.

Our algorithm separates a high-degree core from the remaining variables. Every residual formula on the good variables is a deletion subformula of a common bounded-dependency formula. Adapting the rejection--repair framework of Liu, Wang, Yin, Zhang, and
Zhou~\cite{LWYZZ26}, we construct an executable sampler for all such residual formulas using ordinary connected violation sets in place of weighted $2$-tree witnesses, together with a canonical parity pairing, recursively generated proposals, and a truncated source with global abort and restart. The same construction yields local marginal coins, which are used to simulate approximate block heat-bath updates on the core. To analyze the resulting core chain, we integrate out the good variables, represent the core marginal by an absolutely convergent polymer expansion, and combine a component-swap variance comparison with an adaptive exploration bound. Absorbing single-component interactions into the local weights leaves interactions of cluster size at least two; retaining their additional decay gives a spectral gap of at least $(2s n^{\gamma_{\mathrm{mix}}})^{-1}$, where $s\le n$ is the number of core components and $0<\gamma_{\mathrm{mix}}\le2^{-k+12}$. Together with the residual sampler, this yields the stated running time.
\end{abstract}

\maketitle

\newpage

\tableofcontents

\newpage

\section{Introduction} \label{sec:introduction}

Random constraint satisfaction problems provide a fundamental setting in which
to study the relationship between the existence of solutions, their
geometry, and the computational cost of finding them. A central example
is random \(k\)-SAT: a formula \(\Phi = \Phi(k, n, m)\) on \(n\) Boolean variables consists
of \(m=\lfloor\alpha n\rfloor\) independently generated clauses, each
containing \(k\) random literals. The density \(\alpha\) controls the
number of constraints and the interactions among variables. As the
density increases, random formulas undergo a transition from being
satisfiable with high probability to being unsatisfiable with high
probability.

Numerical experiments and statistical-physics analyses of this transition
motivated the \emph{satisfiability conjecture}~\cite{MPZ02,MMZ05}.
It asserts that, for each fixed \(k\ge3\), there is a critical density
\(\alpha_{\mathrm{sat}}(k)\) such that
\[
 \lim_{n\to\infty}
 \Pr{\Phi(k, n,\lfloor\alpha n\rfloor)
                  \text{ is satisfiable}}
 =\begin{cases}
   1, & \alpha<\alpha_{\mathrm{sat}}(k),\\
   0, & \alpha>\alpha_{\mathrm{sat}}(k).
  \end{cases}
\]

Establishing the conjecture and determining the threshold became central challenges in the study of random constraint systems.
Friedgut~\cite{Fri99} established a sharp transition, although this alone does not identify a limiting critical density. Using the second-moment method, Achlioptas and Peres~\cite{AP03} narrowed the satisfiability bounds to an additive gap of order \(k\) around \(2^k\ln2\).
Coja-Oghlan and Panagiotou~\cite{COP16} subsequently obtained the predicted constant-order correction in the large-\(k\) asymptotics.
Ding, Sly, and Sun~\cite{DSS22} then proved the existence and identified the value of the limiting threshold for each sufficiently large fixed \(k\), confirming the predicted asymptotic value
\[
 \alpha_{\mathrm{sat}}(k)
 =2^k\ln 2-\frac{1+\ln 2}{2}+o_k(1).
\]

Knowing where solutions exist leaves a further algorithmic question:
how can they be found efficiently? The threshold proof does not supply a polynomial-time search algorithm throughout the satisfiable regime.
Coja-Oghlan's search algorithm~\cite{Coj10} finds a satisfying assignment at densities up to \((1-o_k(1))2^k\ln k/k\), a \(\Theta(\log k/k)\) factor below the satisfiability threshold but matching the so-called \emph{clustering threshold}, where the emergence of well-separated clusters and frozen variables creates long-range constraints that obstruct local decisions~\cite{AC08}.
Rigorous failure results for belief propagation and basic survey propagation guided decimation~\cite{CO17,Het16}, together with lower bounds for low-degree polynomial algorithms~\cite{BH21}, show that broad families of algorithms encounter barriers at or before a constant multiple of this threshold, motivating the conjecture that the polynomial-time search transition coincides asymptotically with the clustering threshold; this remains open for general polynomial-time algorithms.
This gap also highlights the distinction between existence and algorithmic access, and motivates asking how much of the solution space can be accessed efficiently at high densities.

The statistical-physics picture also describes changes within the satisfiable phase: a dominant connected region gives way to many separated clusters, and near the satisfiability threshold a few clusters carry most of the solution weight~\cite{MMZ05,DSS22,CLWYY25}.
Figure~\ref{fig:phase-transitions} places these geometric transitions alongside the density scales.
\begin{figure}[t]
    \centering
    \includegraphics[scale=0.15]{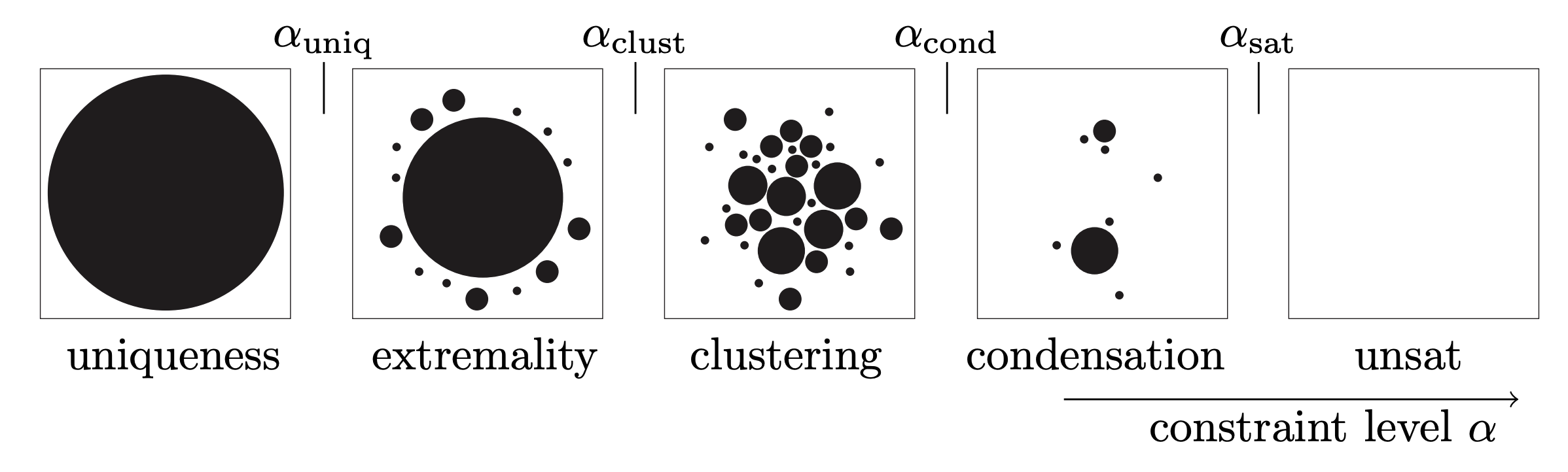}
    \caption{\small The conjectural phase diagram of random \(k\)-SAT, adapted from~\cite{KMRSZ07, DSS22}.
    }
    \label{fig:phase-transitions}
\end{figure}

Beyond finding a single solution, we study \emph{approximately uniform sampling}: given \(\Phi\), generate a satisfying assignment whose distribution is within a prescribed total variation distance \(\eps\) of the uniform law \(\mu_\Phi\).
Sampling asks for more than a certificate of satisfiability. A search algorithm may consistently select a small, atypical part of the solution space; a sampler must represent the relative weights of all its parts. Such access allows us to estimate marginal probabilities and other observables of a typical solution, connecting the algorithmic study of random formulas with their Gibbs distributions and partition functions. Sampling is also closely related to approximate counting, provided that the required conditional instances remain tractable. Thus the sampling problem tests whether the rich probabilistic information carried by a random formula can be made computationally accessible.

Progress in sampling has two dimensions: the range of densities covered and the dependence of the running time on the parameters. An algorithm taking time \((n/\eps)^{\poly{k,\alpha}}\) is polynomial-time for every fixed \(k,\alpha\), but its polynomial degree can increase with both parameters. This distinction is particularly important when the admissible density itself is exponential in \(k\). Our objective is to combine the asymptotically optimal density \(2^k/\poly{k}\) with a running time whose polynomial degree is independent of the width and density.

Early work of Montanari and Shah~\cite{MS07} used correlation decay to approximate the normalized logarithm of a partition function for a permissive, positive-temperature version of random \(k\)-SAT, at densities up to \((1+o_k(1))2\ln k/k\). This is the Gibbs uniqueness threshold and is exponentially below the satisfiability threshold. Galanis, Goldberg, Guo, and Yang~\cite{GGGY21} first extended relative approximation of the number of satisfying assignments to exponentially large densities.
Their adaptation of Moitra's linear-programming method works, for sufficiently large \(k\), when \(\alpha<2^{k/300}/k^3\). A central innovation was to isolate high-degree variables and the variables heavily influenced by them, thereby recovering control of the remaining marginals.
The resulting running time has a parameter-dependent polynomial exponent.
Two concurrent works substantially improved both the density range and the running time. Chen, Galanis, Goldberg, Guo, Herrera-Poyatos, Mani, and Moitra~\cite{CGGGHMM24} developed a spectral-independence analysis robust to high-degree variables and obtained a nearly linear-time sampler up to density \(\alpha\le 2^{0.039k}\), based on a uniform-block Glauber dynamics on a carefully selected subset of variables.
Independently, He, Wu, and Yang~\cite{HWY23} obtained an almost-linear time sampling algorithm reaching the larger density range \(2^{k/3}/\poly{k}\). Their approach replaces the enumeration implicit in the linear-programming method with recursive sampling, supported by a refined analysis of the random formula's structure.
Chen, Lonkar, Wang, Yang, and Yin~\cite{CLWYY25} subsequently reached \(2^k/\poly{k}\) for both approximate counting and sampling, with running time \((n/\eps)^{\poly{k,\alpha}}\). Their analysis of recursive coupling uses denser witness structures in place of the usual \(2\)-tree witnesses, exploiting expansion of random constraint systems. This established that the exponential density condition for sampling can match that of satisfiability and search, while leaving a polynomial factor in \(k\) between the respective density bounds.

These results rely crucially on the random structure of the formula. 
Approximate counting for general \(k\)-CNF formulas encounters hardness at variable degrees of order \(2^{k/2}\), under the usual complexity assumptions~\cite{BGGGS19}. The random model at density \(2^k/\poly{k}\) has average occurrence degree \(k\alpha\), well beyond that exponential bound. The result of~\cite{CLWYY25} therefore demonstrates the advantage of random structure for counting and sampling. The remaining question is whether this advantage can be obtained without a parameter-dependent exponent in the running time:
\begin{quote}
\raggedright\itshape Can one sample approximately uniform solutions of random \(k\)-SAT at density \(2^k/\poly{k}\) in expected time polynomial in \(n,k,\alpha,\eps^{-1}\), with a universal polynomial degree?
\end{quote}

Recent counting and sampling versions of the Lov\'asz local lemma suggest a possible route to a universal polynomial degree. Liu, Wang, Yin, Zhang, and Zhou~\cite{LWYZZ26} achieved a breakthrough on the counting and sampling LLL problem, obtaining an approximate counting algorithm and an exact sampler with expected
near-linear running time under the essentially optimal symmetric condition \(p(D+1)^2=O(1)\). In concurrent work, Jain, Mizgerd, and Pham~\cite{JMP26} establish a general asymmetric counting and sampling criterion on the square of the dependency graph. Under this criterion, they give an FPRAS for the satisfaction probability and an approximate sampler with polynomial expected running time.

However, neither result applies directly as a black box in our regime. For random \(k\)-SAT, even if we ignore high-degree irregularities, a clause of random \(k\)-SAT has violation probability \(p\asymp 2^{-k}\) and typical dependency degree \(D\asymp \alpha k^2\). The symmetric specialization of either squared-dependency condition therefore requires $\alpha \le \frac{2^{k/2}}{\poly{k}}$, exponentially below our target density \(2^k/\poly{k}\). Moreover, the maximum dependency degree of the original random formula is not uniformly bounded as \(n\) grows. We therefore use the random structure of the formula more strongly to overcome these obstacles: high-degree variables are handled by a separate core dynamics, while the residual sampler adapts the rejection--repair mechanism of~\cite{LWYZZ26} from \(2\)-tree witnesses to ordinary connected violation sets. A formula-specific root estimate, valid simultaneously for all residual deletion subformulas, replaces the worst-case \(pD^2\) bound.

\subsection{Our result}

We answer this question affirmatively in the occurrence model with inverse-polynomial accuracy. In particular, our sampler supports error tolerances down to \(n^{-B}\) for any fixed \(B>0\). Here each clause chooses its \(k\) variable occurrences independently and uniformly from \([n]\), and each occurrence receives an independent fair sign. Repeated variables, tautological clauses, and different clause labels with identical contents are allowed.

\begin{theorem} \label{thm:fast-sampler}
    Fix an integer $k \ge 2^{20}$, a real density $\alpha > 0$ with
    \begin{align*}
        \DensityCondition
    \end{align*}
    and a constant $B > 0$.
    For all $n$ above a threshold depending on $k,\alpha,B$, with high probability over $\Phi\sim\Phi(k,n,m)$ with $m=\floor{\alpha n}$, the following holds simultaneously for $n^{-B}\le\eps\le1/2$.
    On input $(\Phi,\eps)$, our randomized algorithm terminates almost surely and outputs a satisfying assignment $X$ in time $T$ such that
    \begin{align*}
        \TV{\Law(X)}{\mu_\Phi} \le \eps, \qquad \E{T} \le \widetilde{O}_B\ab(\poly{k, \alpha + 1} n^{2 + \eta}), \quad 0 < \eta \le 2^{16 - k},
    \end{align*}
    where the polynomial in $k,\alpha+1$ has universal degree, $\widetilde O_B$ hides factors polylogarithmic in $n$ and a constant depending on $B$, and $\eps$ is given under the binary input convention of \Cref{subsec:full-sampler-input}.
\end{theorem}



\begin{remark}
    The cutoff \(2^{20}\), the power \(16\), and the running-time degree are conservative. We have not optimized these constants as much as possible.
\end{remark}

\Cref{thm:fast-sampler} applies throughout $\alpha\le 2^k/k^{16}$. In the density and accuracy ranges covered by both results, the sampling bound of~\cite{CLWYY25} is $(n/\eps)^{\poly{k,\alpha}}$. Here the exponent of $n$ is $2+\eta$, with $0<\eta\le 2^{16-k}$, and the polynomial factor in $k$ and $\alpha+1$ has universal degree.


\subsection{Technique overview}

Our algorithm has two levels: a recursive sampler for the good variables conditional on an assignment to a high-degree core, and a block heat-bath chain targeting the marginal distribution on the core. The main tasks are to make the conditional sampler efficient throughout the high-density regime and to prove that the core chain mixes in polynomial time.

Previous fast sampling algorithms for random $k$-SAT~\cite{CGGGHMM24,HWY23} handle the good variables first and complete the remaining variables conditional on the sampled values.
Their analysis cannot reach density $2^{k}/\poly{k}$. In contrast, we place the bad core in the outer sampling layer, using a block heat-bath chain, while an inner sampler handles the good variables and the residual clauses conditional on the current core assignment.

\subsubsection*{A core-first treatment of high-degree variables.}
Following the separation approach of~\cite{GGGY21,HWY23,CLWYY25}, we collect the high-degree variables and those strongly influenced by them into a bad core. With high probability, every core component contains at most $\log_2 n$ variables and has a nonempty feasible set. Removing the core literals from every non-core clause produces a fixed good formula $\Psi$. For any feasible core assignment $\sigma$, the conditional formula on the good variables is a deletion subformula $\Psi[D_\sigma]$: clauses already satisfied by $\sigma$ disappear, whereas all remaining clauses keep the same good literals. We sample the core marginal by updating one whole component at a time, and use the residual sampler only as an inner routine for implementing these updates and for the final completion. In particular, we never enumerate the joint state space of the core or condition on a previously sampled global assignment to the good variables.

\subsubsection*{A fast rejection--repair sampler for the good variables.}
The residual problem can in principle be approached by adapting two recent algorithms for sampling Lov\'asz local lemma, either the recursive insertion framework of~\cite{JMP26} or the rejection--repair framework of~\cite{LWYZZ26}. We use the latter because its local transport mechanism runs in near-linear time and thus leads to the sharper running-time bound needed here. The bounded-degree analysis of~\cite{LWYZZ26} is based on weighted $2$-tree witnesses and, when applied through its worst-case condition, reaches only the density bound $2^{k/2}/\poly{k}$. We replace these witnesses by ordinary connected violation sets and exploit properties specific to random formulas. Expansion bounds the probability that a small connected family of residual clauses is violated simultaneously; for larger families, averaging over the independent random literal signs controls their total contribution. The resulting root estimate holds on one high-probability event simultaneously for every deletion subformula of $\Psi$, including those selected adaptively by the core chain.

This uniform estimate supports the entire executable rejection--repair construction: a canonical parity pairing, recursively generated connected-set proposals, exact Bernoulli factories, and a truncated source with global abort and restart. It gives an almost-uniform sampler for every residual formula in expected time $\widetilde O(n+m\,\poly{k,\overline D})$. The same proposal mechanism also produces local marginal coins for conditional clause-violation probabilities. These coins allow a candidate assignment to one core component to be tested clause by clause, thereby implementing an approximate heat-bath update without computing extension counts.

\subsubsection*{Rapid mixing of the core dynamics.}
Small core components make individual updates affordable, but do not by themselves imply rapid mixing: after the good variables are integrated out, distinct components remain coupled through their numbers of satisfying extensions. We express this core marginal by an absolutely convergent polymer expansion. Regrouping its cluster terms yields weak, possibly high-order interactions among the core components. For finite truncations of this interaction representation, a two-copy component-swap argument compares the global variance with the sum of the single-component conditional variances. The comparison reduces the dependence between components to connected regions of an auxiliary random hypergraph whose interaction labels can be dominated independently of the sampled core configurations.

We then explore the component containing a fixed block in this dominating hypergraph. The random-formula incidence bound controls the number of clauses touching the explored region, and the core-growth estimate relates its size to the high-degree variables it contains. To obtain the sharper exponent, we first absorb every single-component interaction into its local weight. Each remaining interaction has cluster size at least two, so the weighted polymer-root estimate improves its unweighted root bound from $\rho$ to $\rho_0=\rho e^{-2\tau}$. If $T$ is the total cluster size of the open interactions exposed during exploration, the number of discovered blocks is at most $1+T$. Combining this bound with the exponential moment of $T$ gives
\[
    \operatorname{Var}_{\mu^{\bad}}(f)
    \le 2n^{\gamma_{\mathrm{mix}}}\sum_{i=1}^s
    \mathbb E_{\mu^{\bad}}\!\left[
        \operatorname{Var}\!\left(f\mid X_{-i}\right)
    \right],
\]
where $\gamma_{\mathrm{mix}}=64\rho_0/(1-25/(3k^3))\le2^{-k+12}$. Thus the core heat-bath chain has spectral gap at least $(2s n^{\gamma_{\mathrm{mix}}})^{-1}$ and mixes from any feasible initial state in $O_B(s n^{1+\gamma_{\mathrm{mix}}})$ updates. Local marginal coins simulate these updates with an accumulated error at most $\eps/4$, and a last call to the residual sampler draws the good variables conditional on the mixed core assignment. Accounting for feasible-candidate rejection and for all rejected heat-bath trials adds only $O(2^{-k})$ to the exponent, yielding $\widetilde O_B(\poly{k,1+\alpha}n^{2+O(2^{-k})})$ overall.

\subsection{Organization}
The remainder of the paper is organized as follows. \Cref{sec:preliminary} fixes the notation and reviews the probabilistic and analytic tools used throughout, including polymer models, the Lov\'asz local lemma, and functional inequalities for Markov chains. \Cref{sec:random-sat-structure} establishes the expansion and incidence properties of random formulas, constructs the high-degree core, and proves that feasible core assignments extend to satisfying assignments of the residual formula.
\Cref{sec:core-fast-sampling} proves a polynomial spectral-gap bound for the heat-bath dynamics on the core using a polymer representation.
\Cref{secd:rejection-repair-sampler} constructs the residual sampler and the local marginal coins used in each core update. \Cref{sec:full-sampler-design} combines these routines and proves \Cref{thm:fast-sampler}. \Cref{sec:appendix} supplies the component-swap, exploration, and reciprocal-test proofs.

\subsection*{AI disclosure}
ChatGPT suggested the main approach to proving \Cref{thm:fast-sampler}, combining the rapid convergence of the heat-bath dynamics on the core with a modification of the rejection-repair dynamics in \cite{LWYZZ26} for the good formula to obtain a faster running time.
GPT was also used throughout the preparation of this manuscript for language editing and formatting.
The authors have verified and refined all proofs and take full responsibility for the paper.

\section{Preliminaries} \label{sec:preliminary}

This section fixes notations and provides some necessary preliminaries.

\subsection{Notations}
For a natural number $n$, $[n]$ is the set $\set{1, \ldots, n}$.
Both $\log$ without a subscript and $\ln$ denote the natural logarithm; logarithms to base two are written as $\log_2$.
For an event $\mathcal E$, the indicator of $\mathcal E$ is denoted by $\id{\mathcal E}$.

For an assignment $\sigma$ on $U$, we use $\sigma(v)$ to denote the value of $\sigma$ on $v$.
For a subset $S$ of $U$, denote by $\sigma(S)$ the partial assignment $\sigma|_S$ of $\sigma$ on $S$.
To assign values, we use $S \gets \sigma$ to describe the term that $S$ is assigned the partial assignment $\sigma$.
For the event $A$ which is measurable by assignments, we use $\one_A$ to denote the indicator function of $A$, \IE,
\begin{align*}
    \one_A(\sigma) = \id{A\;\text{holds under the assignment}\;\sigma}.
\end{align*}

Fix a probability distribution $\pi$ on a state space $\Omega$.
For every test function $f : \Omega \to \mathbb{R}$, we define the expectation and the variance of $f$ with respect to $\pi$ as
\begin{align*}
    \E[\pi]{f} \triangleq \sum_{x \in \Omega} \pi(x) f(x), \quad \Var[\pi]{f} \triangleq \E[\pi]{f^2} - \E[\pi]{f}^2.
\end{align*}
The total variation distance of two probability distributions $\mu$ and $\nu$ on a common state space $\Omega$ is defined as
\begin{align} \label{eq:def-TV-distance}
    \TV{\mu}{\nu} \triangleq \frac{1}{2} \sum_{x \in \Omega} \abs{\mu(x) - \nu(x)}. 
\end{align}
A \emph{coupling} of two probability distributions $\mu$ and $\nu$ is a
joint distribution of a pair of random objects $(X,Y)$ such that
$X\sim\mu$ and $Y\sim\nu$. Equivalently, the first and second marginals
of the joint distribution are $\mu$ and $\nu$, respectively. The random
objects $X$ and $Y$ need not be independent; their dependence may be
chosen to facilitate a comparison while preserving both marginal laws.

We also extend \eqref{eq:def-TV-distance} to zero-mean functions.
For every function $f : \Omega \to \mathbb R$ such that $\sum_{x \in \Omega} f(x) = 0$, define $\norm{f}_{\mathrm{TV}} \triangleq \frac{1}{2} \sum_{x \in \Omega} \abs{f(x)}$.

\subsection{Random SAT formula}
A random $k$-SAT instance $\Phi \sim \Phi(k, n, m)$ is a $k$-SAT formula on $n$ variables $\V = \set{v_1, \ldots, v_n}$ with $m$ labeled clauses $\C$.
Each clause $C_a$ of label $a$ selects $k$ literals $\ell_{a, 1}, \ldots, \ell_{a, k}$ independently and uniformly at random from $\set{v_1, \ldots, v_n, \neg v_1, \ldots, \neg v_n}$.
We often refer to $k$ as the width of the formula.
For simplicity, we abuse the use of the clause $C_a$ and its label $a$.
Let $\var(\ell)$ be the variable name of a literal $\ell$.
For a clause label $a$ and a set $S\subseteq\C$ of clause labels, define
\[
    \vbl(a)=\vbl(C_a)\triangleq
    \set{\var(\ell_{a,i})\mid i\in[k]},\qquad
    \vbl(S)\triangleq\bigcup_{a\in S}\vbl(a).
\]
These sets record syntactic variable occurrences, regardless of literal signs, and include the variables occurring in tautological clauses.

Let $\nu$ denote the uniform distribution over all assignments, regardless of whether they satisfy $\Phi$.
For a family $\F$ of constraints, we denote by $\Good(\F)$ the event that all constraints in $\F$ are satisfied.
We use $\mu_{\F}(\cdot)$ to denote the distribution $\nu\left(\cdot \;|\; \Good(\F)\right)$.

For an event $\mathcal{E}$, $\vbl(\mathcal{E})$ is the collection of variables on which it depends. This semantic support may be smaller than the syntactic variable set of a clause defining the event; in particular, a tautological clause has a constant violation event. The structural quantities $\deg(v)$ and $I(U)$, and the graphs $\H_\Phi$ and $\G_\Phi$ below, always use syntactic variable occurrences.
For a variable $v \in \V$, the degree $\deg(v)$ is the number of clause labels containing $v$ and for a variable subset $U \subseteq \V$, the number of its incident clause labels is denoted by $I(U)$, \IE,
\begin{equation}
    I(U) \triangleq \abs{\set{a \in \C \midc \vbl(a) \cap U \neq \emptyset}}.
    \label{eq:incident-clause-count}
\end{equation}

The dependency and primal graphs induced by $\Phi$ are introduced to analyze the structure of a SAT formula.
The dependency graph $\H = \H_{\Phi}$ has a vertex set $V(\H) = \C$.
We add an edge in $\H$ for two distinct clause labels if their representing clauses share at least one common variable. For a vertex $v$ and vertex subset $S$ in a graph $\H$, we denote by $N_{\H}(v)$ and $N_{\H}[v] = N_{\H}(v) \cup {v}$ (respectively $N_{\H}[S]$) the open and closed neighborhoods.
Meanwhile, the primal graph $\G = \G_{\Phi}$ of $\Phi$ has a vertex set $V(\G) = \V$ and joins two distinct variables if they occur in a common clause.
For a variable subset $U$, we say $U$ is connected if and only if $\G[U]$ is connected.

\subsection{Abstract polymer models and the Koteck\'y--Preiss criterion}
An abstract polymer model consists of a finite set $\mathcal P$ of
polymers, a symmetric incompatibility relation $\not\sim$ on
$\mathcal P$, and an activity $w_\gamma\in\mathbb C$ for each
$\gamma\in\mathcal P$. The incompatibility relation is part of the
model: $\gamma\not\sim\gamma'$ means that the polymers $\gamma$ and
$\gamma'$ are incompatible and cannot occur together in a compatible
family. We regard every polymer as incompatible with itself, so
$\gamma\not\sim\gamma$ for every $\gamma\in\mathcal P$. A family
$\mathcal F\subseteq\mathcal P$ is compatible if no two distinct members
are incompatible. Its partition function is
\begin{equation}
 \Xi(w)=
 \sum_{\mathcal F\ {\rm pairwise\ compatible}}
       \prod_{\gamma\in\mathcal F}w_\gamma,
 \label{eq:polymer-partition-prelim}
\end{equation}
where the empty family contributes one.

For an ordered tuple
$\Gamma=\left(\gamma_1,\ldots,\gamma_\ell\right)\in\mathcal P^\ell$, repetitions
allowed, let $I(\Gamma)$ be its incompatibility graph on $[\ell]$:
distinct $i,j$ are adjacent exactly when
$\gamma_i\not\sim\gamma_j$. Define its Ursell coefficient by
\begin{equation}
 \varphi(\Gamma)=
 \sum_{\substack{F\subseteq E\left(I(\Gamma)\right)\\
                 ([\ell],F)\ {\rm connected}}}
       (-1)^{|F|}.
 \label{eq:ursell-coefficient}
\end{equation}
For $\ell=1$, this coefficient is one. If $I(\Gamma)$ is disconnected,
then $\varphi(\Gamma)=0$. Define the formal cluster series
\begin{equation}
 \mathcal L(w)\triangleq
 \sum_{\ell\ge1}\frac1{\ell!}
 \sum_{\left(\gamma_1,\ldots,\gamma_\ell\right)\in\mathcal P^\ell}
 \varphi\left(\gamma_1,\ldots,\gamma_\ell\right)
 \prod_{i=1}^{\ell}w_{\gamma_i}.
 \label{eq:polymer-cluster-expansion-prelim}
\end{equation}
The series in
\eqref{eq:polymer-cluster-expansion-prelim} is infinite because tuples
may contain repeated polymers.

\begin{theorem}[Koteck\'y--Preiss criterion \cite{KP86}]
\label{thm:kotecky-preiss}
Let $u_\gamma\ge0$ satisfy
$|w_\gamma|\le u_\gamma$ for every $\gamma\in\mathcal P$, and let
$A:\mathcal P\to[0,\infty)$ be a penalty function. Suppose that
\begin{equation}
 \forall\,\gamma\in\mathcal P,\qquad
 \sum_{\gamma'\not\sim\gamma}
       u_{\gamma'}e^{A\left(\gamma'\right)}\le A(\gamma).
 \label{eq:kotecky-preiss-condition}
\end{equation}
Then the series $\mathcal L(w)$ in
\eqref{eq:polymer-cluster-expansion-prelim} converges absolutely and
\[
 \mathcal L(w)=\log\Xi(w).
\]
Here $\log\Xi(w)$ denotes the analytic branch satisfying
$\log\Xi(0)=0$.
Moreover, for every fixed polymer $\gamma_0$,
\begin{equation}
 \sum_{\ell\ge1}\frac1{(\ell-1)!}
 \sum_{\gamma_2,\ldots,\gamma_\ell\in\mathcal P}
 \left|\varphi\left(\gamma_0,\gamma_2,\ldots,\gamma_\ell\right)\right|
 \prod_{i=2}^{\ell}u_{\gamma_i}
 \le e^{A\left(\gamma_0\right)}.
 \label{eq:kotecky-preiss-pinned}
\end{equation}
For $\ell=1$, the inner tuple and product in
\eqref{eq:kotecky-preiss-pinned} are empty, and their contribution is
one.
\end{theorem}

For a fixed polymer $\gamma$, the left-hand side of
\eqref{eq:kotecky-preiss-condition} sums over all polymers directly
incompatible with $\gamma$. The factor $e^{A\left(\gamma'\right)}$ accounts for
possible further extensions of a connected cluster from $\gamma'$.
Thus \eqref{eq:kotecky-preiss-condition} bounds the total weighted
contribution of such extensions by $A(\gamma)$. The estimate
\eqref{eq:kotecky-preiss-pinned} is called the pinned form of the
criterion because the polymer $\gamma_0$ is fixed; it bounds the total
absolute weight of all connected clusters containing this prescribed
polymer.

\subsection{\texorpdfstring{\Lovasz}{Lov\'asz} local lemma}
The \Lovasz local lemma, first established in \cite{EL75}, provides a sufficient condition for the satisfiability of a CNF formula.
In this work, we use the following asymmetric form of the local lemma.
\begin{theorem}[{Asymmetric \Lovasz local lemma \cite{HSS11}}] \label{thm:aymmetric-local-lemma}
    Given bad events $B_1, \ldots, B_m$ in a product probability space, set $N(a)\triangleq\set{c\in[m]\setminus\{a\}\mid \vbl(B_a)\cap\vbl(B_c)\neq\emptyset}$. If there exist $\set{x_a \in (0, 1)}_{a = 1}^m$ such that
    \begin{align*}
        \forall a = 1, \ldots, m, \quad \Pr{B_a} \le x_a \prod_{c \in N(a)} \left(1 - x_c\right),
    \end{align*}
    then the probability of avoiding all bad events is at least $\prod_{a = 1}^m \left(1 - x_a\right) > 0$.
    Moreover, for any event $\mathcal E$, it holds that
    \begin{align*}
        \Pr{\mathcal E \mid \bigwedge_{a = 1}^m \overline{B_a}} \le \Pr{\mathcal E} \prod_{a \in N(\mathcal E)} \left(1 - x_a\right)^{-1},
    \end{align*}
    where $N(\mathcal E) \triangleq \set{a \in [m] \mid \vbl(\mathcal E) \cap \vbl\left(B_a\right) \neq \emptyset}$.
\end{theorem}

\subsection{The Efron--Stein inequality}
\label{subsec:efron-stein}
We use the following conditional-variance form of the Efron--Stein
inequality, also known as variance tensorization for product measures;
see \cite[Theorems~4--5]{BBL04}.

\begin{lemma}[Efron--Stein inequality]
\label{lem:efron-stein}
Let $r\ge1$ and let $Z_1,\ldots,Z_r$ be independent random variables,
where $Z_i$ takes values in a finite set $\mathcal S_i$.
Write $Z=\left(Z_1,\ldots,Z_r\right)$ and, for each $i\in[r]$, let
$Z_{-i}=\left(Z_j\right)_{j\in[r]\setminus\{i\}}$ denote all coordinates except $i$.
Then, for every $h:\prod_{i=1}^r\mathcal S_i\to\mathbb R$,
\begin{equation}
 \operatorname{Var}\left(h(Z)\right)
 \le\sum_{i=1}^r\mathbb E\!\left[
       \operatorname{Var}\left(h(Z)\mid Z_{-i}\right)\right].
 \label{eq:efron-stein-conditional-variance}
\end{equation}
\end{lemma}

For each $i$, conditioning on $Z_{-i}$ fixes every coordinate except
$Z_i$. Thus the corresponding summand measures the average variance
arising from the randomness of $Z_i$ alone. The inequality bounds the
total variance by the sum of these single-coordinate contributions.
Independence of the coordinates is required, but the random variables
$Z_1,\ldots,Z_r$ need not have identical distributions.

\subsection{Markov chains and functional inequalities}
A Markov chain on a state space $\Omega$ consists of a transition kernel $P$ on $\Omega$ with its stationary distribution $\pi$.
In this paper, we only consider \emph{ergodic} and \emph{reversible} Markov chains (see details in \cite{LP17}).

The \emph{Dirichlet form} is usually used to study the convergence of Markov chains.
For simplicity, we discuss more details of this relationship in \Cref{sec:core-fast-sampling} and only introduce basic concepts here.
The Dirichlet form associated with $P$ is the following functional defined as: for two test functions $f, g : \Omega \to \mathbb R$,
\begin{align*}
    \DF_P(f, g) \triangleq \frac{1}{2} \sum_{x, y \in \Omega} \pi(x) P(x, y) \left(f(x) - f(y)\right)\left(g(x) - g(y)\right).
\end{align*}
The Rayleigh--Ritz variational characterization of the spectral gap (see \cite[Lemma~13.7 and Remark~13.8]{LP17}) gives\footnote{
    Lemma~13.7 uses nonzero mean-zero functions. For $f$ with positive variance, set $g = f - \mathbb E_{\pi}[f]$.
    Then $\mathbb E_{\pi}[g]=0$, $\mathcal E_{P}(g,g)=\mathcal E_{P}(f,f)$, and $\|g\|_{L^2\left(\pi\right)}^2=\operatorname{Var}_{\pi}(f)$.
    Thus centering yields the equivalent variance formulation displayed here; see Remark~13.8.
}
\[
 \gap(P)\triangleq 1-\lambda_2(P)
 =\inf_{\substack{f:\Omega\to\mathbb R\\
                  \operatorname{Var}_{\pi}(f)>0}}
       \frac{\mathcal E_{P}(f,f)}
            {\operatorname{Var}_{\pi}(f)},
\]
where $\lambda_2(P)$ is the second-largest eigenvalue of $P$.
\section{Structural Properties of Random SAT Formulae} \label{sec:random-sat-structure}

For the design of our sampler, we discuss the structure of random $k$-SAT formulae here.
We follow the routine of high-degree separation in \cite{GGGY21,HWY23,CLWYY25}: given a chosen degree cut-off, we form a \emph{core} of a random $k$-SAT formula by including variables of high degree or `significantly influenced' by high-degree variables.
With high probability, the core part of the original formula is composed of small components and therefore we can enumerate all feasible assignments on each component.

Before we state the structural properties, we clarify necessary parameters in our design and analysis here.
Set
\begin{equation} \label{eq:parameter-choice}
\begin{gathered}
    \lambda = 1 + m/n, \quad \Delta = 6k^7 \lambda, \quad C_{\INC} = 100k^2\lambda, \quad \overline D = \ceil{6k^8\lambda}+1, \\
    n_0 = \e^{k - 3} \lambda(k - 4)^4, \quad q_0 = \left(4n_0\right)^{-1/3}, \quad r = \floor{q_0 n}.
\end{gathered}
\end{equation}
Before presenting detailed theorems and properties of random formulae, we first briefly introduce the meaning of these parameters here:
\begin{itemize}
    \item \(\lambda\) is the clause density plus one. It keeps the density-dependent bounds valid even for sparse formulas.
    \item \(\Delta\) is the high-degree cutoff. Every variable appearing in at least \(\Delta\) clauses is viewed as a `high-degree' variable.
    \item \(C_{\mathrm{INC}}\) controls how many clauses touch a connected variable set.
    \item \(\overline D\) bounds the number of residual clauses in a clause’s closed dependency neighborhood, including the clause itself.
    \item \(n_0\) is the coefficient in the expansion bound. The probability that some set of \(t\) clauses contains fewer than \((k-4)t\) distinct variables is at most \(\bigl[n_0(t/n)^3\bigr]^t\).
    \item \(q_0\) sets the fraction at which this expansion bound decays geometrically, while \(r\) is the integer cutoff.
\end{itemize}

\subsection{Expansion of random SAT formulae}
We apply the following expansion of a random $k$-SAT formula.
\begin{lemma} \label{lem:expansion}
    Suppose that $k \ge 30$ and $n^2 \ge 2n_0$.
    With probability at least $1 - \frac{2n_0}{n^2} - \frac{4}{3} \cdot 4^{-n^{1/3}}$, for every set $S$ of at most $r$ clauses, it holds that
    \begin{align} \label{eq:small-set-expansion}
        \abs{\vbl(S)} \ge (k - 4) \abs{S}.
    \end{align}
\end{lemma}
\begin{proof}
    Fix an integer $t \le r$.
    Assume that there exists a subset $S$ of exactly $t$ clauses such that
    \begin{align*}
        \abs{\vbl(S)} < (k - 4) \abs{S} = (k - 4)t.
    \end{align*}
    Since $(k - 4)t \le (k - 4) q_0 n < n$ by \eqref{eq:parameter-choice}, it holds that all variables in $S$ completely lie in some $U \subseteq [n]$ of size $(k - 4)t$.

    For a fixed subset $S$ of clause labels with $\abs{S} = t$ and a variable subset $U$ of size $(k - 4)t$, we bound the probability that all variables in $S$ lie in $U$ as
    \begin{align*}
        \Pr{\vbl(S) \subseteq U} \le \ab(\frac{(k - 4)t}{n})^{k t}.
    \end{align*}
    By the union bound together with $\binom{n}{k} \le (\e n/k)^k$,
    \begin{align*}
        \Pr{\exists S,\; \abs{S} = t,\; \abs{\vbl(S)} < (k - 4) \abs{S}} &\le \binom{m}{t} \binom{n}{(k - 4)t} \ab(\frac{(k - 4)t}{n})^{kt} \\
        &\le \ab(\frac{\e m}{t})^t \ab(\frac{\e n}{(k - 4)t})^{(k - 4)t} \frac{(k - 4)^{kt}t^{kt}}{n^{kt}} \\
        &\le \ab(\e^{k - 3} \cdot \frac{m}{n}(k - 4)^4 \ab(\frac{t}{n})^3)^t \\
        &\le \ab(n_0 \cdot \frac{t^3}{n^3})^t.
    \end{align*}
    For $t \le n^{1/3}$, $n_0 t^3/n^3 \le n_0/n^2 \le 1/2$.
    Then
    \begin{align*}
        \sum_{t = 1}^{n^{1/3}} \ab(n_0 \cdot \frac{t^3}{n^3})^t \le \frac{n_0/n^2}{1 - n_0/n^2} \le \frac{2n_0}{n^2}.
    \end{align*}
    Otherwise, using $t/n \le q_0$ and $n_0 \cdot q_0^3 = 1/4$,
    \begin{align*}
        \sum_{t = \lfloor n^{1/3}\rfloor + 1}^{\lfloor q_0 n \rfloor} \ab(n_0 \cdot \frac{t^3}{n^3})^t \le \sum_{t > n^{1/3}} 4^{-t} \le \frac{4}{3} \cdot 4^{-n^{1/3}}.
    \end{align*}
    Thus the two cases together conclude the lemma.
\end{proof}

\subsection{Estimate for incident clauses}
We need the estimate for the number of incident clauses of a connected variable subset.
\begin{lemma} \label{lem:incident-clause-estimate}
    For every $k \ge 2$ and $m \ge 0$, with probability at least $1 - n^{-3}$, every connected variable subset $U$ satisfies
    \begin{align*}
        I(U) < C_{\INC} \abs{U} + 4\ln n.
    \end{align*}
\end{lemma}
\begin{proof}
    At first, we consider the case $\abs{U} > 1$.
    For a connected variable subset $U$, pick a spanning tree of $\G[U]$.
    For every edge in this spanning tree, select a clause containing both endpoints of it.
    After discarding repeated clause labels, we obtain a clause certificate $T$ such that
    \begin{align*}
        \abs{T} \le \abs{U} - 1, \quad U \subseteq \vbl(T).
    \end{align*}
    Moreover, $T$ is connected in $\H$.
    
    Next, we count the number of connected clauses.
    Let $N_t$ be the number of connected $t$-clause labels.
    For a collection of $t$ clause labels, if they are connected, there is a spanning tree of size $t$.
    Root the tree at an arbitrary label.
    We reveal such a tree in the order that the parent precedes its children.
    Then conditional on everything already revealed, the probability that a new clause intersects its parent is at most $k \cdot k/n = k^2/n$.
    Consequently, by Cayley's formula,
    \begin{align*}
        \E{N_t} &\le \binom{m}{t} t^{t - 2} \ab(\frac{k^2}{n})^{t - 1} \\
        &\le \frac{\e^t m^t}{t^t} \cdot t^{t - 2} \cdot \frac{k^{2(t - 1)}}{n^{t - 1}} \\
        &\le \frac{n}{k^2 t^2} \ab(\e (m/n) k^2)^t \\
        &\le n \left(\e \lambda k^2\right)^t.
    \end{align*}

    We take the union bound over fixed pairs $(T,U)$, where $T$ is a set of $t$ clause labels and $U$ is a set of $\ell$ variables with $1\le t\le\ell-1$.
    Let $A_{T,U}$ be the event that $T$ is connected in $\H$ and $U\subseteq\vbl(T)$.
    This event depends only on the clauses in $T$.
    Conditional on $A_{T,U}$ (when it has positive probability), all clauses outside $T$ therefore retain their independent, original laws.
    Thus the number $X$ of clauses outside $T$ that intersect the fixed set $U$ satisfies
    \begin{align*}
        X \sim \Bin\ab(m - \abs{T}, p) \quad \text{where}\; p = 1 - \ab(1 - \frac{\abs{U}}{n})^k \le \frac{k\abs{U}}{n}.
    \end{align*}
    Since $\abs{T} \le \abs{U} - 1$, the event $\set{I(U) \ge C_{\INC} \abs{U} + 4\ln n}$ implies
    \begin{align*}
        X \ge \left(C_{\INC} - 1\right) \abs{U} + 4\ln n.
    \end{align*}
    Note that $\left(C_{\INC} - 1\right) \abs{U} \ge \e^2 \cdot \E{X \mid A_{T,U}}$ since $\E{X \mid A_{T,U}} \le m \cdot k\abs{U}/n = (m/n) k \abs{U} \le \lambda k\abs{U}$.
    By the Chernoff bound,
    \begin{align*}
        \Pr{X \ge \left(C_{\INC} - 1\right)\abs{U} + 4\ln n \mid A_{T,U}} \le \e^{-\left(\left(C_{\INC} - 1\right) \abs{U} + 4\ln n\right)}.
    \end{align*}
    Set the event
    \begin{align*}
        \mathcal{E}_{\ell} \triangleq \set{\exists U,\; \abs{U} = \ell, \mathcal G_\Phi[U] \text{ is  connected, and } I(U) \ge C_{\INC} \abs{U} + 4\ln n}.
    \end{align*}
    For each fixed $T$, summing $\Pr{A_{T,U}}$ over all $\ell$-variable sets $U$ gives
    $\E{\id{T\text{ is connected}}\binom{|\vbl(T)|}{\ell}}$.
    Since $|\vbl(T)|\le kt$, the union bound over the fixed pairs gives
    \begin{align*}
        \Pr{\mathcal{E}_{\ell}} &\le \sum_{t = 1}^{\min\;\set{\ell - 1, m}} \E{N_t} \binom{kt}{\ell} \e^{-\left(C_{\INC} - 1\right) \ell - 4\ln n} \\
        &\le \ell n \left(\e \lambda k^2\right)^\ell (\e k)^{\ell} \e^{-\left(C_{\INC} - 1\right) \ell - 4\ln n} \\
        &= n^{-3} \cdot \ell \exp\ab(-\ell\ab(C_{\INC} - 1 - \ln\left(\e^2 \lambda k^3\right))).
    \end{align*}
    By \eqref{eq:parameter-choice}, for $k \ge 2$ and $\lambda \ge 1$,
    \begin{align*}
        C_{\INC} - 1 - \ln\left(\e^2 \lambda k^3\right) \ge 100k^2 \lambda - 3 - 3\ln k - \ln \lambda \ge 2.
    \end{align*}
    Therefore,
    \begin{align*}
        \Pr{\bigcup_{\ell \ge 2} \mathcal{E}_{\ell}} &\le n^{-3} \sum_{\ell \ge 2} \ell \cdot \e^{-2\ell} \\
        &= n^{-3} \cdot \frac{\e^{-4}\left(2 - \e^{-2}\right)}{\left(1 - \e^{-2}\right)^2} \\
        &\le \frac{1}{20} \cdot n^{-3}.
    \end{align*}

    To conclude the lemma, it remains to consider the case $\abs{U} = 1$, \IE, $U = \set{v}$.
    Observe that $I(U)\sim\Bin(m,1-(1-1/n)^k)$ and $\E{I(U)}\le mk/n$.
    By the Chernoff bound,
    \begin{align*}
        \Pr{I(U) \ge C_{\INC} + 4\ln n} \le n^{-4}\e^{-C_{\INC}}.
    \end{align*}
    By the union bound, the failure probability is at most
    \begin{align*}
        \frac{1}{20} \cdot n^{-3} + n \cdot \e^{-C_{\INC}} n^{-4} \le n^{-3}.
    \end{align*}
    Therefore we conclude the lemma.
\end{proof}

\subsection{High-degree separation}
For a subset $S \subseteq \V$, we construct the core induced by $S$ as:
\begin{enumerate}
    \item initialize the core as $S$;
    \item repeatedly select a remaining clause containing more than $8$ distinct variables in the current core, and add all its variables to the core.
\end{enumerate}
We use $\Core(S)$ to denote the final set generated by $S$.

Consider the set collecting all variables of high degree:
\begin{align*}
    \V^{\high} \triangleq \set{v \in \V \mid \deg(v) \ge \Delta}.
\end{align*}
The core of $\Phi$ is defined as $\V^{\bad} = \Core\left(\V^{\high}\right)$ and denote by $\V^{\good} = \V \setminus \V^{\bad}$ the collection of good variables.

A clause $C_a$ of label $a$ is called a \emph{core} clause if its variables are completely contained in $\V^{\bad}$ (equivalently, $\vbl(a) \subseteq \V^{\bad}$).
We set $\C^{\good}$ as the collection of all non-core clause labels, \IE,
\begin{align*}
    \C^{\good} \triangleq \set{a \in \C \mid \vbl(a) \cap \V^{\good} \neq \emptyset}.
\end{align*}
For every $a \in \C^{\good}$, let $C_a^{\good}$ be the clause of $C_a$ after violating all bad variables in it, \IE,
\begin{align*}
    C_a^{\good} \triangleq \bigvee_{i \in [k]\; :\; \var\left(\ell_{a, i}\right) \in \V^{\good}} \ell_{a, i}.
\end{align*}
and the good formula is set as
\begin{align*}
    \Psi \triangleq \bigwedge_{a \in \C^{\good}} C_a^{\good}.
\end{align*}

We first bound the size of $\V^{\high}$.
\begin{lemma} \label{lem:high-degree-tail-bound}
    For $k \ge 30$ and $h = \ceil{q_0 n}$, it holds that
    \begin{align*}
        \Pr{\abs{\V^{\high}} \le r} \ge 1 - \e^{-h}.
    \end{align*}
\end{lemma}
\begin{proof}
    Since $\abs{\V^{\high}}$ is an integer, $\abs{\V^{\high}} > r$ implies $\abs{\V^{\high}} \ge h$.
    Thus it suffices to show that $\Pr{\abs{\V^{\high}} \ge h} \le \e^{-h}$.
    Fix a variable subset $U$ of size $h$.
    Let $X_U$ be the random variable to denote the number of literals whose variable names are in $U$.
    Then $X_U$ is a binomial variable:
    \begin{align*}
        X_U \sim \Bin\ab(km, \frac{h}{n}) \quad \text{thus} \quad \E{X_U} = \frac{m}{n}k h \le \lambda kh.
    \end{align*}
    If $U \subseteq \V^{\high}$, it must hold that $X_U \ge h \Delta$.
    The Chernoff bound provides that
    \begin{align*}
        \Pr{X_U \ge \Delta h} \le \ab(\frac{\e \lambda k h}{\Delta h})^{\Delta h} = \ab(\frac{\e}{6k^6})^{\Delta h}
    \end{align*}
    since $\Delta = 6k^7 \lambda$.
    By the union bound,
    \begin{align*}
        \Pr{\abs{\V^{\high}} \ge h} &\le \binom{n}{h} \ab(\frac{\e}{6k^6})^{\Delta h} \\
        &\le \ab(\frac{\e n}{h})^h \ab(\frac{\e}{6k^6})^{\Delta h} \\
        &\le \ab(\frac{\e}{q_0} \ab(\frac{\e}{6k^6})^\Delta)^h.
    \end{align*}
    For $k \ge 30$, by choices of parameters \eqref{eq:parameter-choice},
    \begin{align*}
        \frac{\e}{q_0} \ab(\frac{\e}{6k^6})^{\Delta} \le \e^{-1}.
    \end{align*}
    Hence we conclude the lemma.
\end{proof}

Another important property is the bounded growth of the core for every high-degree variable subset.
\begin{lemma} \label{lem:bounded-core-growth}
    Assume $r \ge 1, \abs{\V^{\high}} \le r$ and the expansion event of \Cref{lem:expansion}.
    For every $S \subseteq \V^{\high}$, it holds that
    \begin{align*}
        \abs{\Core(S)} \le \frac{k - 4}{5} \abs{S} \le \frac{k}{2} \abs{S}.
    \end{align*}
\end{lemma}
\begin{proof}
    By construction of $\Core(S)$, every selected clause has at most $k - 9$ new variables added to the core.
    Therefore, after any $t \le r$ selections, by \Cref{lem:expansion},
    \begin{align*}
        (k - 4)t \le \abs{\vbl(\text{selected clauses})} \le \abs{S} + (k - 9)t
    \end{align*}
    which implies that $5t \le \abs{S}$.
    It also implies that the process cannot reach $t = r$ otherwise it must hold $5r \le \abs{S} \le r$.
    Therefore
    \begin{align*}
        \abs{\Core(S)} \le \abs{S} + (k - 9)\frac{\abs{S}}{5} = \frac{k - 4}{5} \abs{S},
    \end{align*}
    which concludes this lemma.
\end{proof}

\subsubsection{Small components in the bad core}
A key property of $\V^{\bad}$ is that with high probability, its components have small size.
Let $B_1, \ldots, B_s$ be the maximal connected components of $\G[\V^{\bad}]$.
\begin{lemma} \label{lem:component-growth-size}
    Assume the joint event of \Cref{lem:expansion,lem:incident-clause-estimate,lem:high-degree-tail-bound}.
    For any non-empty subset $J \subseteq [s]$ and $U = \bigcup_{i \in J} B_i$, it holds that
    \begin{align*}
        \abs{U} \le \frac{k}{2} \abs{U \cap \V^{\high}}.
    \end{align*}
    Moreover, each component is of size at most $\log_2 n$.
\end{lemma}
\begin{proof}
    The key ingredient to prove the first inequality is the following identity:
    \begin{align} \label{eq:component-identity}
        \Core\left(U \cap \V^{\high}\right) = U.
    \end{align}

    To show $\Core\left(U \cap \V^{\high}\right) \subseteq U$, suppose that there is a clause containing at least $9$ variables in $U$.
    All its variables should be in $\V^{\bad}$ by the construction of $\V^{\bad}$ and they are connected in $\G[\V^{\bad}]$.
    Note that each $B_i$ is a maximal connected component for $i = 1, \ldots, s$.
    These variables must lie in one component.
    Therefore they are in $U$ since the clause meets $U$.

    Now we prove that $U \subseteq \Core\left(U \cap \V^{\high}\right)$.
    Take a sequence of clause additions constructing $\V^{\bad}$ from $\V^{\high}$.
    Then we retain only the additions whose clauses lie in $U$.
    It is clear that this is a legal sequence to construct $\Core\left(U \cap \V^{\high}\right)$ from $U \cap \V^{\high}$, and the sequence finally produces all of $U$ since $U \subseteq \V^{\bad}$.
    Hence $U \subseteq \Core\left(U \cap \V^{\high}\right)$.

    Putting two directions together, we conclude \eqref{eq:component-identity}.
    By \Cref{lem:bounded-core-growth},
    \begin{align*}
        \abs{U} = \abs{\Core\left(U \cap \V^{\high}\right)} \le \frac{k}{2} \abs{U \cap \V^{\high}}
    \end{align*}
    which is exactly the first conclusion.

    To bound the size of each component, suppose that there is a component $B_i \subseteq \V^{\bad}$ such that $\abs{B_i} > \log_2 n$.
    By the inequality above, it holds that
    \begin{align*}
        \abs{B_i} \le \frac{k}{2} \abs{B_i \cap \V^{\high}}.
    \end{align*}
    Consider the number of clauses incident to $B_i \cap \V^{\high}$.
    Since every variable in $B_i \cap \V^{\high}$ has degree at least $\Delta$,
    \begin{align*}
        \Delta \abs{B_i \cap \V^{\high}} \le \sum_{v \in B_i \cap \V^{\high}} \deg(v) \le k I\left(B_i\right).
    \end{align*}
    The last inequality holds since every clause incident to $B_i$ contains at most $k$ distinct variables in $B_i \cap \V^{\high}$.
    By \Cref{lem:incident-clause-estimate},
    \begin{align*}
        \frac{2\Delta \abs{B_i}}{k} \le \Delta \abs{B_i \cap \V^{\high}} \le k C_{\INC} \abs{B_i} + 4k \ln n.
    \end{align*}
    By the assumption $\abs{B_i} > \log_2 n$, $\ln n < \abs{B_i} \cdot \ln 2$.
    Then we have
    \begin{align*}
        \frac{2\Delta}{k} \abs{B_i} < k\left(C_{\INC} + 4\ln 2\right) \abs{B_i}
    \end{align*}
    giving that $1 < k^2\left(C_{\INC} + 4\ln 2\right)/(2\Delta)$.
    By \eqref{eq:parameter-choice},
    \begin{align*}
        \frac{k^2\left(C_{\INC} + 4\ln 2\right)}{2\Delta} = \frac{25}{3k^3} + \frac{\ln 2}{3k^5\lambda} < 1
    \end{align*}
    for $k \ge 30$ and $\lambda \ge 1$.
    This contradiction concludes that $\abs{B_i} \le \log_2 n$ for every $i = 1, \ldots, s$.
\end{proof}

Now we check the feasibility of the core formula with arbitrary literal signs, which provides a guarantee of validity for updates to the block dynamics in our sampling algorithm.
As a first step, we show that, under the preceding structural properties, every connected component of core clauses is small.
\begin{proposition} \label{prop:small-core-clause-component}
    Assume the joint event of \Cref{lem:expansion,lem:incident-clause-estimate,lem:high-degree-tail-bound}.
    For every integer $L \in [\log_2{n}, r]$ and $i = 1, \ldots, s$, there are at most $L$ core clauses meeting $B_i$.
\end{proposition}
\begin{proof}
    Suppose that there exists $B_i$ met by at least $L$ core clauses.
    Select $L$ of them to form $S$.
    By \Cref{lem:expansion}, since $\abs{S}=L\le r$ and $k-4>1$, it holds that
    \[ \abs{\vbl(S)} \ge (k-4)\abs{S} = (k-4)L > L. \]
    On the other hand, every clause in $S$ is a core clause meeting $B_i$, so all its variables lie in $B_i$. Hence $\vbl(S)\subseteq B_i$. By \Cref{lem:component-growth-size},
    \begin{align*}
        \abs{\vbl(S)} \le \abs{B_i} \le \log_2 n \le L,
    \end{align*}
    which leads to a contradiction.
\end{proof}

Now we are ready to place the following lemma.
\begin{lemma} \label{lem:core-feasible-assignment}
    Assume the joint event of \Cref{lem:expansion,lem:incident-clause-estimate,lem:high-degree-tail-bound} and $r\ge\ceil{\log_2 n}$.
    Then the core formula is satisfiable with any choice of literal signs, and its feasible set is
    \begin{align*}
        \Omega \triangleq \prod_{i = 1}^s \Omega_i \neq \emptyset
    \end{align*}
    where $\Omega_i$ consists of all feasible assignments to $B_i$ satisfying all core clauses meeting that component.

    Precisely, let $c_i$ be the number of core clauses in $B_i$ and $r_i$ be the number of mixed clauses incident to $B_i$.
    Then it holds that
    \begin{align*}
        I(B_i) \le \frac{4\ln{n}}{1 - 25/(3k^3)}, \quad c_i \le \frac{\abs{B_i}}{k - 4}
    \end{align*}
    and the expected cost $U_i$ of the rejection sampler to generate a feasible assignment in $\Omega_i$ is at most 
    \begin{align*}
        U_i \le \widetilde O\ab((\abs{B_i} + k c_i + 1) \exp\ab(c_i \frac{2^{-k + 4}}{1 - 2^{-k + 4}})) = \widetilde O\ab(\poly{k, \lambda} n^\kappa)
    \end{align*}
    where
    \begin{align*}
        \kappa = \frac{2^{-k + 4}}{1 - 2^{-k + 4}} \cdot \frac{2k^2}{\Delta (1 - 25/(3k^3)) (k - 4)}.
    \end{align*}
\end{lemma}
\begin{proof}
    For every $i = 1, \ldots, s$, let $\C_i$ be the collection of core clauses meeting $B_i$.
    Construct a bipartite graph with:
    \begin{itemize}
        \item one left vertex for each clause label $a \in \C_i$;
        \item one right vertex for each variable $v \in B_i$;
        \item an edge $a \sim v$ if and only if $v \in \vbl(a)$.
    \end{itemize}
    For every $S \subseteq \C_i$, the neighbors of $S$ are exactly $\vbl(S)$.
    By \Cref{prop:small-core-clause-component}, for every subset $S \subseteq \C_i$, $\abs{S} \le \abs{\C_i} \le r$, thus $\abs{\vbl(S)} \ge (k - 4) \abs{S}$.
    Thus by Hall's theorem, there is an injection $f_i: \C_i \to B_i$ such that $f_i(a)$ occurs in the clause $C_a$.
    
    Now we construct a satisfying assignment based on the matching.
    Given any choice of literal signs, for every clause $a \in \C_i$,
    \begin{itemize}
        \item choose one occurrence of $f_i(a)$ in $C_a$;
        \item set $f_i(a)$ making $C_a$ true.
    \end{itemize}
    Since $f_i$ is injective, different clauses prescribe values to different variables.
    Finally, we set all unmatched variables in $B_i$ arbitrarily.
    Hence we conclude that $\Omega_i \neq \emptyset$ for every $i = 1, \ldots, s$.
    
    Since $B_1, \ldots, B_s$ are disjoint and thus each core clause belongs to exactly one component, then it holds that every assignment in $\Omega \triangleq \prod_{i = 1}^s \Omega_i$ is feasible for the full core formula.
    On the other hand, given a feasible assignment of the core formula, its restriction to $B_i$ also satisfies $\C_i$ for each $i = 1, \ldots, s$, meaning that it belongs to $\Omega$.
    Thus we conclude that the feasible set of the core formula is exactly $\Omega$.

    Since $B_i$ is connected, the incidence bound gives $I(B_i)\le C_{\INC}|B_i|+4\ln n$. The degree sum and \Cref{lem:component-growth-size} give $|B_i|\le k^2 I(B_i)/(2\Delta)$. Substituting this into the incidence bound and using $C_{\INC}k^2/(2\Delta)=25/(3k^3)$ yields
    \[
        I(B_i)\le\frac{4\ln n}{1-25/(3k^3)},\qquad
        |B_i|\le\frac{2k^2\ln n}{\Delta(1-25/(3k^3))}.
    \]
    Furthermore, the expansion property provides that $c_i \le \abs{B_i}/(k - 4)$.

    Replace each core clause on the left side of the bipartite graph by $k-4$ copies with the same neighbors. For any set of copies, let $T$ be its set of original clauses. The set contains at most $(k-4)|T|$ copies and has $|\vbl(T)|\ge(k-4)|T|$ neighbors. Hall's theorem therefore gives a matching saturating the copies. Each clause receives $k-4$ distinct matched variables, and the groups assigned to different clauses are disjoint.
    Choose one occurrence of each matched variable in its assigned clause. Under a uniform assignment, the event that at least one of these literals is true has probability $1-2^{-(k-4)}$ for each clause. These events are independent across clauses and their intersection satisfies the entire core component. Hence
    \begin{align*}
        \frac{\abs{\Omega_i}}{2^{\abs{B_i}}} \ge \ab(1 - 2^{-k + 4})^{c_i} \ge \exp\ab(-c_i \frac{2^{-k + 4}}{1 - 2^{-k + 4}}).
    \end{align*}
    This also implies that this rejection sampler terminates almost surely.
    To generate a candidate, it takes $\abs{B_i}$ fair bits and checks at most $k c_i$ literal occurrences.
    Therefore the expected cost is
    \begin{align*}
        U_i \le \widetilde O\ab((\abs{B_i} + k c_i + 1) \exp\ab(c_i \frac{2^{-k + 4}}{1 - 2^{-k + 4}})) = \widetilde O\ab(\poly{k, \lambda} n^\kappa)
    \end{align*}
    by above bounds for $\abs{B_i}$ and $c_i$.
\end{proof}

\subsubsection{Extension from the core formula to good variables}
Now we consider the good formula $\Psi$.
Let $\H_{\Psi}$ be the dependency graph of $\Psi$ and $\vbl^{\good}(\cdot)$ denote the variables in the good formula $\Psi$.
\begin{proposition} \label{prop:good-property}
    Assume the joint event of \Cref{lem:expansion,lem:incident-clause-estimate,lem:high-degree-tail-bound}.
    Then the following properties hold.
    \begin{itemize}
        \item The dependency graph $\H_{\Psi}$ of $\Psi$ has maximum degree $d \le k\Delta$ and thus $d + 1 \le \overline{D}$ by \eqref{eq:parameter-choice}.
        \item For every $S \subseteq \C^{\good}$ with $\abs{S} \le r$,
        \begin{align*}
            \abs{\vbl^{\good}(S)} \ge (k - 12) \abs{S}.
        \end{align*}
    \end{itemize}
\end{proposition}
\begin{proof}
    Observe that $\V^{\high} \subseteq \V^{\bad}$ and $\V^{\good} = \V \setminus \V^{\bad}$.
    Then for every variable $v \in \V^{\good}$, it holds that $\deg(v) < \Delta$.
    Therefore the maximum degree of $\H_{\Psi}$ is at most $k\Delta$.

    For every $S \subseteq \C^{\good}$ with $\abs{S} \le r$, by \Cref{lem:expansion},
    \begin{align*}
        \abs{\vbl^{\good}(S)} &\ge \abs{\vbl(S)} - \sum_{a \in S} \abs{\vbl(a) \cap \V^{\bad}} \\
        &\ge (k - 4) \abs{S} - 8 \abs{S} \\
        &\ge (k - 12) \abs{S}
    \end{align*}
    which concludes the second property.
\end{proof}

The following lemma shows that every feasible assignment to the core formula can be extended to a satisfying assignment of the original formula $\Phi$.
For every subset $F \subseteq \C^{\good}$, the deletion sub-formula of $\Psi$ induced by $F$ is defined as
\begin{align*}
    \Psi[F] \triangleq \bigwedge_{a \in F} C_a^{\good}.
\end{align*}
Let $\nu^{\good}$ be the uniform product measure on $\set{0, 1}^{\V^{\good}}$.
For $F\subseteq\C^{\good}$, write $\Good(F)$ for $\Good(\Psi[F])$ and define
\[
    Z_F\triangleq\nu^{\good}(\Good(F)),\qquad
    \mu_F(\cdot)\triangleq\nu^{\good}(\cdot\mid\Good(F))
\]
whenever $Z_F>0$. The following lemma guarantees this positivity for every $F$.
For $a \in \C^{\good}$, let $E_a$ be the event that the residual clause $C_a^{\good}$ is violated, and for $S \subseteq \C^{\good}$ define
\begin{align*}
    E_S = \bigcap_{a \in S} E_a, \qquad q_S = \nu^{\good}\left(E_S\right).
\end{align*}
\begin{lemma} \label{lem:core-assignment-good-variable-extension}
    Assume $r\ge1$ and the joint event of \Cref{lem:expansion,lem:incident-clause-estimate,lem:high-degree-tail-bound}.
    Put $p_* = 2^{-k + 12}$ and $C_* = 10^4$.
    If
    \begin{align} \label{eq:good-variable-condition}
        p_* \overline{D} \le \frac{1}{2C_*},
    \end{align}
    then every deletion sub-formula $\Psi[F]$ of $\Psi$ is satisfiable and it holds that for every $a \notin F$,
    \begin{align} \label{eq:violating-upper-bound}
        \mu_F(E_a) \le \e p_*.
    \end{align}
    Consequently for $J \subseteq F$, $Z_{F \setminus J} \le (1 - \e p_*)^{-\abs{J}} Z_F$.
    
    Moreover, for every $\sigma \in \Omega$, it extends to a satisfying assignment of $\Phi$.
    Let $D_{\sigma}$ be the collection of labels whose bad-variable literals are all false at $\sigma$.
    Then it holds that
    \begin{align*}
        \mu_{\Phi}\ab(\cdot \mid \V^{\bad} \gets \sigma) = \mu_{\Psi[D_{\sigma}]}(\cdot).
    \end{align*}
\end{lemma}
\begin{proof}
    Fix a deletion sub-formula $\Psi[F]$. Since $r \ge 1$ by assumption, applying \Cref{prop:good-property} to the singleton set $\set{a}$ gives $\abs{\vbl^{\good}(a)} \ge k-12$. Consequently, every clause in $\Psi[F]$ has violating probability at most $p_*$, and has at most $d + 1 \le \overline{D}$ dependencies by \Cref{prop:good-property}.
    Take $x_a = x \triangleq \e p_*$ in \Cref{thm:aymmetric-local-lemma} for every clause label $a \in F$.
    It holds that for every $a$, by $x \overline{D} \le \e/\left(2C_*\right) < 1/128$ and $x < 1/2$,
    \begin{align*}
        p_* = x \e^{-1} \le x (1 - x)^{d}.
    \end{align*}
    Therefore there exists a satisfying assignment for $\Psi[F]$.

    To show \eqref{eq:violating-upper-bound}, we prove it by induction on $\abs{F}$.
    First, since every deletion sub-formula $\Psi[F]$ is satisfiable, $Z_F > 0$.
    For $F = \emptyset$, it is clear to see $\mu_\emptyset(E_a) \le p_* \le \e p_*$.
    Now fix $F$ with $\abs{F} = t \ge 1$ and $a \notin F$.
    Split $F$ into the neighbors and non-neighbors of $a$:
    \begin{align*}
        S = F \cap N_{\H_\Psi}(a), \quad T = F \setminus S
    \end{align*}
    and order $S = \set{b_1, \ldots, b_\ell}$ where $\ell \le d$.
    
    The event $E_a$ uses variables disjoint from every clause in $T$.
    Hence
    \begin{align*}
        \nu^{\good}\ab(E_a \mid \Good(F)) = \frac{\nu^{\good}\ab(E_a \cap \Good(S) \mid \Good(T))}{\nu^{\good}\ab(\Good(S) \mid \Good(T))} \le \frac{p_*}{\nu^{\good}\ab(\Good(S) \mid \Good(T))}
    \end{align*}
    To bound the denominator, by the chain rule,
    \begin{align*}
        \nu^{\good}\ab(\Good(S) \mid \Good(T)) = \prod_{j = 1}^\ell \ab(1 - \nu^{\good}\ab(E_{b_j} \mid \Good(T \cup \set{b_1, \ldots, b_{j - 1}}))).
    \end{align*}
    By induction, each factor is at least $1 - \e p_*$.
    Therefore
    \begin{align*}
        \nu^{\good}\ab(\Good(S) \mid \Good(T)) \ge (1 - \e p_*)^\ell \ge (1 - \e p_*)^d.
    \end{align*}
    Thus we conclude that $\nu^{\good}\ab(E_a \mid \Good(F)) \le p_*/(1 - \e p_*)^d \le \e p_*$.

    To complete the proof, fix a partial assignment $\sigma \in \Omega$ on the core formula.
    Each non-core clause either is already satisfied by $\sigma$ and disappears, or retains its residual clause in $\Psi$.
    Then it is trivial to see the satisfiability under $\sigma$ and the identity of conditional law.
\end{proof}

\subsection{Connected subsets of good variables}
We bound the total weight of connected clause sets containing a fixed clause of $\Psi$.

For every good clause $a \in \C^{\good}$, let $\CS_{\Psi}(a)$ denote the nonempty connected sets in $\H_{\Psi}$ containing $a$.
For $k\ge2^{20}$ and $m/n\le 2^k/k^{16}$ we have $1\le 2^k/k^{16}$. Hence $\lambda\le 2\cdot 2^k/k^{16}$ and the rounded degree bound satisfies
\[
    \overline D\le 6k^8\lambda+2\le 14\cdot 2^k/k^8,
    \qquad p_*\overline D\le 14\cdot 2^{12}/k^8\le 2^{-40}.
\]
In particular, \eqref{eq:good-variable-condition} holds. Define
\begin{align*}
    \tau = \ln\frac{1}{C_* p_* \overline{D}}, \qquad x = \e p_*, \qquad \rho = \frac{4}{C_* \cdot \overline{D}}.
\end{align*}
Since $C_*=10^4$, the preceding bound gives $\tau\ge\ln(2^{40}/C_*)>1$, $x<1/2$, and $x\overline D\le\e/(2C_*)$.

We need the following fact to bound the number of connected sets with given size.
\begin{lemma}[{\cite[Lemma 2.1]{BCKL13}}] \label{lem:connected-set-count}
    In a graph of maximum degree at most $\Delta$, the number of connected sets of size $\ell$ containing a certain vertex is at most $(\e \Delta)^{\ell - 1}$.
\end{lemma}

We view the generation of $\Phi$ as the following three-stage process.
\begin{itemize}
    \item \textbf{The structure generation:} sample all variable-occurrences without signs and separate variables into the bad core part $\V^{\bad}$ and the good part $\V^{\good}$.
    
    \item \textbf{Bad-variable signs:} sample bad-variable occurrence signs independently at random.
    
    \item \textbf{Good-variable signs:} sample good-variable occurrence signs independently at random.
\end{itemize}
We remark here that the dependency graph $\H_\Phi$ and the primal graph $\G_\Phi$ depend only on the structure generation.

\begin{proposition} \label{prop:root-event}
    Assume the joint event of \Cref{lem:expansion,lem:incident-clause-estimate,lem:high-degree-tail-bound}. Suppose $m \ge 1$.
    Fix $\zeta \in (0, 1)$ such that $L \le r = \floor{q_0 n}$ for $L = \left\lceil \log_2 (32m/\zeta)\right\rceil$.
    Conditional on the structure generation and every choice of bad-variable signs, with probability at least $1 - \zeta$ over the good-variable signs,
    \begin{align} \label{eq:good-root-event}
        \sup_{a \in \C^{\good}}\; \sum_{\gamma \in \CS_{\Psi}(a)} q_\gamma \e^{(1 + \tau) \abs{\gamma}} \le \rho.
    \end{align}
\end{proposition}
\begin{proof}
    Write
    \begin{align*}
        A_a \triangleq \sum_{\gamma \in \CS_{\Psi}(a), \abs{\gamma} \le L} q_\gamma \e^{(1 + \tau) \abs{\gamma}}, \quad
        T_a \triangleq \sum_{\gamma \in \CS_{\Psi}(a), \abs{\gamma} > L} q_\gamma \e^{(1 + \tau) \abs{\gamma}}.
    \end{align*}
    Put $z = \e^2 / C_* < 10^{-3}$.
    Since $L \le r$, \Cref{prop:good-property} implies that, for every $S \subseteq \C^{\good}$ with $\abs{S} \le L$, the violating probability is at most
    \begin{align*}
        q_S \le 2^{-(k - 12) \abs{S}} = p_*^{\abs{S}}.
    \end{align*}
    Then with \Cref{lem:connected-set-count}, the contribution from $\abs{\gamma} \le L$ is at most
    \begin{align*}
        A_a \le \sum_{\ell \ge 1} \ab(\e \overline{D})^{\ell - 1} p_*^{\ell} \e^{(1 + \tau)\ell} = \sum_{\ell \ge 1} p_* \e^{1+\tau} \ab(\e^{2+\tau}\overline{D}p_*)^{\ell-1} = \sum_{\ell\ge 1}\frac{\e}{C_* \overline{D}} z^{\ell-1}  = \frac{\e}{C_* \overline{D}(1 - z)}.
    \end{align*}
    For the tail $T_a$ with $\abs{\gamma} > L$, interchange expectation over good-variable signs with integration against $\nu^{\good}$.
    For every fixed assignment, all good-variable literal occurrences in $S$ are independently false with probability $1/2$, even when variables repeat within or across clauses.
    The number of these occurrences is at least $\sum_{c\in S}|\vbl^{\good}(c)|$.
    Applying \Cref{prop:good-property} separately to each singleton $\{c\}$ (using $r\ge L\ge1$) therefore gives, for every $S\subseteq\C^{\good}$, without a restriction on $|S|$,
    \begin{align*}
        \E[\text{good signs}]{q_S} \le 2^{-\sum_{c \in S} \abs{\vbl^{\good}(c)}} \le p_*^{\abs{S}}.
    \end{align*}
    Therefore
    \begin{align*}
        \E{T_a} \le \sum_{\ell = L + 1}^{\infty} \ab(\e \overline{D})^{\ell - 1} p_*^{\ell} \e^{(1 + \tau) \ell} = \sum_{\ell \ge L + 1} p_* \e^{1+\tau} \ab(\e^{2+\tau}\overline{D}p_*)^{\ell-1} = \sum_{\ell\ge L + 1}\frac{\e}{C_* \overline{D}} z^{\ell-1} = \frac{\e z^{L}}{C_* \cdot \overline{D} (1 - z)}.
    \end{align*}
    The Markov inequality gives that
    \begin{align*}
        \Pr{\max_{a \in \C^{\good}} \; T_a > \frac{1}{C_* \cdot \overline{D}}} \le 3m z^L < \zeta.
    \end{align*}
    Consequently, with probability at least $1 - \zeta$, for every $a \in \C^{\good}$,
    \begin{align*}
        A_a + T_a \le \frac{\e}{C_* \cdot \overline{D}(1 - z)} + \frac{1}{C_* \cdot \overline{D}} \le \rho,
    \end{align*}
    which concludes the proposition.
\end{proof}

Putting all structural properties together, we present the following criterion to verify whether a $k$-SAT formula $\Phi$ is nice.
\begin{condition}[Nice formula] \label{cond:nice-formula}
    With parameters defined as above, a formula $\Phi$ on $n$ Boolean variables with $m$ clauses with width $k$ is said to be \emph{nice} if the following conditions hold.
    \begin{itemize}
        \item \textbf{Parameter conditions:} $k \ge 2^{20}$, $m/n \le 2^k/k^{16}$, and $r\ge\max\{1,\ceil{\log_2 n}\}$. In particular, $\tau\ge1$, $x < 1/2$ and $x \overline{D} \le \e/(2C_*)$.
        \item \textbf{Structural conditions:} \Cref{lem:expansion,lem:incident-clause-estimate,lem:high-degree-tail-bound} hold for $\Phi$.
        \item \textbf{Polymer-root estimate:} for every $a \in \C^{\good}$,
        \begin{align*}
            \sum_{S \in \CS_\Psi(a)} q_S \exp\ab((1 + \tau) \abs{S}) \le \rho.
        \end{align*}
    \end{itemize}
\end{condition}

\begin{lemma} \label{lem:random-k-SAT-nice}
    Fix an integer $k \ge 2^{20}$ and a real density $0 < \alpha \le 2^k/k^{16}$.
    Let $\Phi \sim \Phi(k, n,\floor{\alpha n})$ be a random $k$-SAT formula on $n$ variables and $\floor{\alpha n}$ clauses.
    Then $\Phi$ satisfies \Cref{cond:nice-formula} with probability $1 - o_n(1)$ for all sufficiently large $n$.
\end{lemma}
\begin{proof}
    For fixed $k,\alpha$, $\lambda\le1+\alpha$, so $q_0$ is bounded below by a positive constant and $r\ge\max\{1,\ceil{\log_2 n}\}$ for all sufficiently large $n$. The other parameter conditions follow from the density bound.
    By \Cref{lem:expansion,lem:incident-clause-estimate,lem:high-degree-tail-bound}, union bound gives that for all sufficiently large $n$, the structural conditions hold with probability at least
    \begin{align*}
        1 - \ab(\frac{2n_0}{n^2} + \frac{4}{3} 4^{-n^{1/3}} + n^{-3} + \e^{-\ceil{q_0 n}}) = 1 - o_n(1).       
    \end{align*}
    For the polymer-root estimate, by \Cref{prop:root-event} with choice of $\zeta = n^{-2}$, \eqref{eq:good-root-event} holds with probability at least $1 - n^{-2}$ for every $a \in \C^{\good}$.
    Putting everything together concludes this lemma.
\end{proof}

\section{Rapid Mixing of Heat-bath Dynamics for the Bad Core}
\label{sec:core-fast-sampling}

We prove a spectral-gap bound of order $1/(s n^{O(2^{-k})})$ for the heat-bath dynamics on the bad core. The small correction in the exponent comes from retaining the exponential tilt of interactions that touch at least two core components.
Recall that $B_1, \ldots, B_s$ are the connected components of $\G_\Phi[\V^{\bad}]$, $\Omega_i$ is the
set of assignments to $B_i$ satisfying its core clauses, and $\Omega = \prod_{i = 1}^s \Omega_i$.
A coordinate of $\Omega$ is an assignment to an entire component $B_i$.

Write $\mu^{\bad}$ for the marginal of $\mu_\Phi$ on $\V^{\bad}$.
We assume the core feasibility and extension condition
\begin{equation}
 \forall\, i\in[s],\quad \Omega_i\ne\emptyset,\qquad
 \forall\sigma\in\Omega,\ \exists\eta\in\{0,1\}^{\V}
 \text{ such that }\eta|_{\V^{\bad}}=\sigma\text{ and }\eta\models\Phi.
 \label{eq:core-feasibility-extension}
\end{equation}
Thus $\mu^{\bad}$ is strictly positive on $\Omega$. Condition
\eqref{eq:core-feasibility-extension} is established in the structural
section, with $n$ sufficiently large for its cutoff conditions.

For a random core configuration $X=\left(X_1,\ldots,X_s\right)$ and a fixed
configuration $\sigma=\left(\sigma_1,\ldots,\sigma_s\right)\in\Omega$, each coordinate
$\sigma_i=\sigma|_{B_i}\in\Omega_i$ denotes the assignment to the entire
core component $B_i$. For $j\in[s]$, define
\[
 X_{-j}=\left(X_i\right)_{i\in[s]\setminus\{j\}},\qquad
 \sigma_{-j}=\left(\sigma_i\right)_{i\in[s]\setminus\{j\}}.
\]
Thus the event $X_{-j}=\sigma_{-j}$ means that $X_i=\sigma_i$ for every
$i\ne j$, or equivalently that the assignments on
$\V^{\bad}\setminus B_j$ agree with $\sigma$.
For $A\subseteq[s]$, write $\sigma_A=\left(\sigma_i\right)_{i\in A}$ and
$\sigma_{-A}=\sigma_{[s]\setminus A}$, and use the same convention for
random configurations.

Let $\K$ denote the transition kernel of the random-scan block
heat-bath chain on $\Omega$. From a configuration $\sigma\in\Omega$,
the chain chooses $j\in[s]$ uniformly and replaces the assignment
on $B_j$ by a sample from
$\mu^{\bad}\left(\,\cdot\mid X_{-j}=\sigma_{-j}\right)$,
leaving all other core assignments unchanged.
For a probability measure $\pi$ on $\Omega$ and a real-valued function
$f:\Omega\to\mathbb R$, let $X\sim\pi$ and set
\[
 V_j^\pi(f)=
 \mathbb E_\pi\!\left[\operatorname{Var}_\pi\left(f(X)\mid X_{-j}\right)\right].
\]
For $s\ge1$, the Dirichlet form of $\K$ with respect to $\mu^{\bad}$,
following the definition in Section~\ref{sec:preliminary}, is
\[
 \mathcal E_{\K}(f,f) = \frac12\sum_{\sigma,\omega\in\Omega}
       \mu^{\bad}(\sigma)\K(\sigma,\omega)
       \bigl(f(\omega)-f(\sigma)\bigr)^2.
\]
Viewing each block $B_j$ as one coordinate, the standard heat-bath
Dirichlet-form identity (see
\cite[Definition~3.3(a)]{CLV23}) gives
\begin{equation}
 \mathcal E_{\K}(f,f)
 =\frac1s\sum_{j=1}^s V_j^{\mu^{\bad}}(f).
 \label{eq:core-dirichlet}
\end{equation}
Each block update contributes its expected conditional variance, and the
factor $1/s$ comes from the uniform choice of the updated block.
For any constant $C>0$, the variance inequality
\begin{equation}\label{eq-conditiononnorm}
 \operatorname{Var}_{\mu^{\bad}}(f)
 \le C\sum_{j=1}^s V_j^{\mu^{\bad}}(f)
\end{equation}
for every $f:\Omega\to\mathbb R$ is equivalent to
$\operatorname{Var}_{\mu^{\bad}}(f)\le sC\,\mathcal E_{\K}(f,f)$.
Each block heat-bath update satisfies detailed balance with respect to
$\mu^{\bad}$, so their average $\K$ is reversible. For $|\Omega|>1$,
the Rayleigh--Ritz characterization of the spectral gap in
\Cref{sec:preliminary} gives
\[
 \gap(\K)=\inf_{\substack{f:\Omega\to\mathbb R\\
                  \operatorname{Var}_{\mu^{\bad}}(f)>0}}
       \frac{\mathcal E_{\K}(f,f)}
            {\operatorname{Var}_{\mu^{\bad}}(f)}
 \ge\frac1{sC}.
\]
We will prove \eqref{eq-conditiononnorm} with $C=2n^{\gamma_{\mathrm{mix}}}$, where
$\gamma_{\mathrm{mix}}=O(2^{-k})$ is made explicit below. If $|\Omega|=1$, every
function is constant, and we adopt the convention $\gap(\K)=1$.

\subsection{The core marginal and its polymer representation}
\label{subsec:core-polymer}

We use the parameters of the structural section:
\begin{equation}
 \lambda=1+m/n,\qquad
 \Delta=6k^7\lambda,\qquad
 C_{\INC}=100k^2\lambda,\qquad
 \overline D=\ceil{6k^8\lambda}+1,\qquad
 \rho=\frac4{10^4\overline D}.
 \label{eq:core-parameters}
\end{equation}
We assume that $\Phi$ satisfies \Cref{cond:nice-formula}, \IE, $\Phi$ is a \emph{nice} formula.
To make our analysis clear, we list here these deterministic properties.
\begin{enumerate}
    \item For every nonempty $J \subseteq [s]$, the variable set $W_J \triangleq \bigcup_{i \in J} B_i$ satisfies
    \begin{equation}
        \abs{W_J} \le \frac{k}{2}\abs{W_J \cap \V^{\high}}, \qquad
        \forall\, v \in \V^{\high}, \quad \deg(v) \ge \Delta.
        \label{eq:core-component-growth-input}
    \end{equation}
    
    \item Every connected variable set $U$ in the original primal graph $\G_\Phi$ satisfies
    \begin{equation}
        I(U) \le C_{\INC}\abs{U} + 4\ln n.
        \label{eq:core-incidence-input}
    \end{equation}
    
    \item The residual dependency graph $\H_\Psi$ satisfies
    \begin{equation}
        \forall\, a \in \C^{\good}, \quad \abs{N_{\H_\Psi}[a]} \le \overline{D}.
        \label{eq:core-dependency-input}
    \end{equation}

    \item Let $\CS_{\Psi}(a)$ consist of the connected clause sets containing $a$. By \Cref{prop:root-event}, a nice formula satisfies the following polymer-root estimate with the value of $\tau$ from \Cref{sec:random-sat-structure}:
    \begin{equation}
      \tau\ge1, \quad \forall a\in\C^{\good},\qquad
      \sum_{S\in\CS_{\Psi}(a)} q_S\exp\ab((1+\tau)|S|) \le \rho.
      \label{eq:core-polymer-root}
    \end{equation}
\end{enumerate}
Throughout this section, polymer connectivity is in $\H_\Psi$, whose edges correspond to shared \emph{good} variables, not in the original dependency graph $\H_\Phi$.

Then we have the following theorem.

\begin{theorem}[Spectral gap of the core heat bath]
\label{thm:core-spectral-gap}
Suppose $n\ge2$ and $k\ge30$, with the parameters defined in
\eqref{eq:core-parameters}. Assume
\eqref{eq:core-feasibility-extension},
\eqref{eq:core-component-growth-input},
\eqref{eq:core-incidence-input},
\eqref{eq:core-dependency-input}, and
\eqref{eq:core-polymer-root}. Define
\begin{equation}
 c_k=1-\frac{25}{3k^3},\qquad
 \rho_0=\rho e^{-2\tau},\qquad
 \gamma_{\mathrm{mix}}=\frac{64\rho_0}{c_k}.
 \label{eq:core-sharp-exponent}
\end{equation}
If $s\ge1$, then, for every $f:\Omega\to\mathbb R$,
\begin{equation}
 \operatorname{Var}_{\mu^{\bad}}(f)
 \le 2n^{\gamma_{\mathrm{mix}}}\sum_{j=1}^s V_j^{\mu^{\bad}}(f),
 \qquad
 \gap(\K)\ge\frac1{2s n^{\gamma_{\mathrm{mix}}}}\ge\frac1{2n^{1+\gamma_{\mathrm{mix}}}}.
 \label{eq:core-gap}
\end{equation}
Consequently, for every $\sigma_0\in\Omega$ and integer $t\ge0$,
\begin{equation}
 \TV{\K^t\left(\sigma_0,\cdot\right)}{\mu^{\bad}}
 \le \frac12\,2^{n/2}\exp\left(-\frac{t}{2s n^{\gamma_{\mathrm{mix}}}}\right).
 \label{eq:core-tv}
\end{equation}
\end{theorem}

In the regime of \Cref{thm:fast-sampler}, we have
$\tau=\ln(1/(C_*p_*\overline D))$ and
$p_*\overline D\le2^{-40}$ by \Cref{sec:random-sat-structure}.
Consequently,
\begin{equation}
 \rho_0=4C_*p_*^2\overline D,\qquad
 \gamma_{\mathrm{mix}}=\frac{256C_*}{c_k}(p_*\overline D)p_*
 \le p_* = 2^{-k+12},
 \label{eq:core-uniform-exponent}
\end{equation}
since $C_*=10^4$ and $c_k\ge99/100$. Thus the variance constant
is at most $2n^{p_*}$, uniformly over the allowed densities.

We first derive precisely the interaction expansion needed for the proof.
For $\sigma\in\Omega$ put
\[
 \chi_a(\sigma)=
 \id{\text{all bad-variable literals of $C_a$ are false at $\sigma$}},
 \qquad D_\sigma=\{a\in\C^{\good}:\chi_a(\sigma)=1\}.
\]
If $C_a$ contains no bad-variable literals, we set
$\chi_a(\sigma)=1$, since its satisfaction depends entirely on the
good variables.
All bad variables in an original
clause lie in one component $B_i$, so a nonconstant $\chi_a$ depends on
one core coordinate. The probability of a good-variable extension is
\begin{equation}
 Z(\sigma)=
 \mathbb E_{\nu^{\good}}\prod_{a\in\C^{\good}}
    \left(1-\chi_a(\sigma)\one_{E_a}\right),
 \qquad
 \mu^{\bad}(\sigma)=\frac{Z(\sigma)}
                         {\sum_{\omega\in\Omega}Z(\omega)}.
 \label{eq:core-extension-weight}
\end{equation}
In particular, this law need not be uniform on $\Omega$.

A polymer is a nonempty set of clause labels
$\gamma\subseteq\C^{\good}$ such that the induced graph
$\H_\Psi[\gamma]$ is connected; it need not be a maximal connected
component of $\H_\Psi$.
Two polymers are incompatible if they intersect or an edge of
$\H_\Psi$ joins them; a polymer is incompatible with itself.
We assign to each polymer the following signed weight, called its
\emph{activity}:
\begin{equation}
 z_\gamma(\sigma)=(-1)^{|\gamma|}q_\gamma
                         \prod_{a\in\gamma}\chi_a(\sigma).
 \label{eq:core-polymer-activity}
\end{equation}
The factor $q_\gamma$ is the joint probability of the violation events
indexed by $\gamma$. The product of the $\chi_a(\sigma)$ ensures that
this activity is zero if any clause in $\gamma$ is already satisfied
by its bad-variable literals.

For completeness, we derive the polymer representation directly from
\eqref{eq:core-extension-weight}. For fixed $\sigma$, expanding the
product and using
$\mathbb E_{\nu^{\good}}\prod_{a\in S}\one_{E_a}=q_S$ gives
\begin{equation}
 Z(\sigma)=
 \sum_{S\subseteq\C^{\good}}(-1)^{|S|}q_S
                           \prod_{a\in S}\chi_a(\sigma).
 \label{eq:core-inclusion-exclusion}
\end{equation}
For $S=\emptyset$, we use $q_\emptyset=1$ and the empty product is one.

For each nonempty $S$, let $\mathcal F(S)$ be the collection of vertex
sets of the connected components of the induced graph $\H_\Psi[S]$.
These components are polymers and are pairwise compatible. Distinct
components share no good variable, since such a variable would give
an edge of $\H_\Psi$ between them. Hence their joint violation events
are independent under the product measure $\nu^{\good}$, and
\[
 q_S=\prod_{\gamma\in\mathcal F(S)}q_\gamma.
\]
No independence is assumed among the violation events within a single
polymer. Since $\mathcal F(S)$ partitions $S$, the corresponding
summand in \eqref{eq:core-inclusion-exclusion} factors as
\[
 (-1)^{|S|}q_S\prod_{a\in S}\chi_a(\sigma)
 =\prod_{\gamma\in\mathcal F(S)}
    \left((-1)^{|\gamma|}q_\gamma
          \prod_{a\in\gamma}\chi_a(\sigma)\right)
 =\prod_{\gamma\in\mathcal F(S)}z_\gamma(\sigma).
\]
Conversely, for any set $\mathcal F$ of pairwise compatible polymers,
its members are exactly the connected components of
$\H_\Psi[\bigcup_{\gamma\in\mathcal F}\gamma]$.
Thus $S\mapsto\mathcal F(S)$, with $\mathcal F(\emptyset)=\emptyset$,
is a bijection between clause subsets and sets of pairwise compatible
polymers. Reindexing \eqref{eq:core-inclusion-exclusion} yields
\begin{equation}
 Z(\sigma)=
 \sum_{\mathcal F\ {\rm pairwise\ compatible}}
       \prod_{\gamma\in\mathcal F}z_\gamma(\sigma).
 \label{eq:core-polymer-partition}
\end{equation}
The sum includes the empty set of polymers, whose contribution is one.

For an ordered tuple
$\Gamma=\left(\gamma_1,\ldots,\gamma_\ell\right)$ of polymers, not
necessarily pairwise compatible and with repetitions allowed, let
$I(\Gamma)$ and $\varphi(\Gamma)$ denote
its incompatibility graph and Ursell coefficient, respectively, as
defined in Section~\ref{sec:preliminary}. Define
\begin{equation}
 \begin{aligned}
 U_\Gamma&=\bigcup_{i=1}^{\ell}\gamma_i,\qquad
 K_\Gamma=\sum_{i=1}^{\ell}|\gamma_i|,\qquad
 a_\Gamma&=\frac{\varphi(\Gamma)}{\ell!}
                   \prod_{i=1}^{\ell}(-1)^{|\gamma_i|}q_{\gamma_i},\qquad
 Y_\Gamma(\sigma)&=\prod_{a\in U_\Gamma}\chi_a(\sigma).
 \end{aligned}
 \label{eq:core-cluster-data}
\end{equation}
We index the cluster-expansion terms $a_\Gamma Y_\Gamma(\sigma)$ by
the ordered polymer tuples $\Gamma$.

By \eqref{eq:core-cluster-data}, we have the following lemma.
\begin{lemma}\label{lem-core-cluster-connectivity}
If $a_\Gamma\ne0$, then $\H_\Psi[U_\Gamma]$ is connected.
\end{lemma}
\begin{proof}
Suppose that $a_\Gamma\ne0$. By \eqref{eq:core-cluster-data}, this implies
$\varphi(\Gamma)\ne0$. If the incompatibility graph $I(\Gamma)$ were
disconnected, there would be no connected spanning subgraph in the sum
defining $\varphi(\Gamma)$ in \eqref{eq:ursell-coefficient}, and hence
$\varphi(\Gamma)=0$, a contradiction. Thus $I(\Gamma)$ is connected.

Each polymer $\gamma_i$ is a connected set of clause labels in
$\H_\Psi$. If $i$ and $j$ are adjacent in $I(\Gamma)$, then
$\gamma_i$ and $\gamma_j$ are incompatible: they either intersect, or
an edge of $\H_\Psi$ joins a clause in $\gamma_i$ to a clause in
$\gamma_j$. In either case, $\gamma_i\cup\gamma_j$ is connected in
$\H_\Psi$. Since $I(\Gamma)$ is connected, following paths in
$I(\Gamma)$ shows that the full union
$U_\Gamma=\bigcup_{i=1}^{\ell}\gamma_i$ is connected in $\H_\Psi$.
\end{proof}

The next lemma is the key bridge from the polymer representation of
$Z(\sigma)$ to a controlled interaction expansion of the core marginal.
It writes $\ln Z(\sigma)$ as a sum of cluster terms
$a_\Gamma Y_\Gamma(\sigma)$ and bounds the total contribution of terms
touching any fixed clause. This absolute-summability bound will let us
regroup the interactions by core components and pass to finite truncations.

\begin{lemma}
\label{lem:core-cluster-expansion}
Under \eqref{eq:core-polymer-root},
\begin{equation}
 \ln Z(\sigma)=\sum_\Gamma a_\Gamma Y_\Gamma(\sigma); \qquad \qquad
 \forall\, a\in\C^{\good},\quad
 \sum_{\Gamma:a\in U_\Gamma}|a_\Gamma|e^{\tau K_\Gamma}\le\rho.
 \label{eq:core-cluster-root}
\end{equation}
Furthermore, we have \[
 \sum_\Gamma\sup_{\sigma\in\Omega}|a_\Gamma Y_\Gamma(\sigma)|
 \le m\rho<\infty.
\]
Hence, the series $\sum_\Gamma a_\Gamma Y_\Gamma(\sigma)$ converges absolutely
and uniformly for $\sigma\in\Omega$.
\end{lemma}
\begin{proof}
Apply Theorem~\ref{thm:kotecky-preiss} with
\[
 u_\gamma=q_\gamma e^{\tau|\gamma|},
 \qquad A(\gamma)=|\gamma|.
\]
For every polymer $\gamma$ and every $\sigma\in\Omega$, each
$\chi_a(\sigma)$ belongs to $\{0,1\}$.
Hence, by \eqref{eq:core-polymer-activity},
\[
 |z_\gamma(\sigma)|e^{\tau|\gamma|}
 =q_\gamma e^{\tau|\gamma|}
      \prod_{a\in\gamma}\chi_a(\sigma)
 \le q_\gamma e^{\tau|\gamma|}
 =u_\gamma.
\]
Thus the upper bound $u_\gamma$ is independent of the core
configuration $\sigma$.
Every polymer $\gamma'$ incompatible with $\gamma$ contains at least one
clause in $N_{\H_\Psi}[\gamma]$. Therefore, using
\eqref{eq:core-polymer-root}, \eqref{eq:core-dependency-input} and \eqref{eq:core-parameters},
\begin{align*}
 \sum_{\gamma'\not\sim\gamma}u_{\gamma'}e^{|\gamma'|}=\sum_{\gamma'\not\sim\gamma}
       q_{\gamma'}e^{(1+\tau)|\gamma'|}\le\sum_{a\in N_{\H_\Psi}[\gamma]}
       \sum_{\gamma'\in\CS_\Psi(a)}
       q_{\gamma'}e^{(1+\tau)|\gamma'|}\le |N_{\H_\Psi}[\gamma]|\rho
 \le\overline D\rho|\gamma|
 \le|\gamma|.
\end{align*}
Thus
\eqref{eq:kotecky-preiss-condition} holds.

Applying Theorem~\ref{thm:kotecky-preiss} to the partition function
\eqref{eq:core-polymer-partition}, with $w_\gamma=z_\gamma(\sigma)$,
gives the absolutely convergent series
\eqref{eq:polymer-cluster-expansion-prelim}. Thus, for every
$\sigma\in\Omega$,
\begin{equation}
 \ln Z(\sigma)
 =\sum_{\ell\ge1}\frac1{\ell!}
   \sum_{\gamma_1,\ldots,\gamma_\ell}
      \varphi\left(\gamma_1,\ldots,\gamma_\ell\right)
      \prod_{i=1}^{\ell}z_{\gamma_i}(\sigma),
 \label{eq:core-log-polymer-expansion}
\end{equation}
where the inner sum ranges over all ordered tuples of polymers, with
repetitions allowed.

For a tuple $\Gamma=\left(\gamma_1,\ldots,\gamma_\ell\right)$ and $a\in U_\Gamma$,
let $r_a(\Gamma)\triangleq|\{i\in[\ell]:a\in\gamma_i\}|\ge1$.
Since $\chi_a(\sigma)\in\{0,1\}$, the definition of $Y_\Gamma$ in
\eqref{eq:core-cluster-data} gives
\[
 \prod_{i=1}^{\ell}\prod_{a\in\gamma_i}\chi_a(\sigma)
 =\prod_{a\in U_\Gamma}\chi_a(\sigma)^{r_a(\Gamma)}
 =\prod_{a\in U_\Gamma}\chi_a(\sigma)
 =Y_\Gamma(\sigma).
\]
Substituting the activities from \eqref{eq:core-polymer-activity}
and using the definition of $a_\Gamma$ in
\eqref{eq:core-cluster-data}, we obtain
\begin{align*}
 \frac{\varphi(\Gamma)}{\ell!}
      \prod_{i=1}^{\ell}z_{\gamma_i}(\sigma)
 =\frac{\varphi(\Gamma)}{\ell!}
      \left(\prod_{i=1}^{\ell}(-1)^{|\gamma_i|}q_{\gamma_i}\right)
      \left(\prod_{i=1}^{\ell}\prod_{a\in\gamma_i}\chi_a(\sigma)\right)=a_\Gamma Y_\Gamma(\sigma).
\end{align*}
Consequently, \eqref{eq:core-log-polymer-expansion} becomes
\[
 \ln Z(\sigma)
 =\sum_{\ell\ge1}\sum_{\Gamma=\left(\gamma_1,\ldots,\gamma_\ell\right)}
      a_\Gamma Y_\Gamma(\sigma)
 =\sum_\Gamma a_\Gamma Y_\Gamma(\sigma),
\]
which proves the expansion asserted in \eqref{eq:core-cluster-root}.

For every fixed polymer $\gamma_0$, the estimate
\eqref{eq:kotecky-preiss-pinned}, with $A\left(\gamma_0\right)=|\gamma_0|$, gives
\begin{equation}
 \sum_{\ell\ge1}\frac1{(\ell-1)!}
  \sum_{\gamma_2,\ldots,\gamma_\ell}
   |\varphi\left(\gamma_0,\gamma_2,\ldots,\gamma_\ell\right)|
   \prod_{i=2}^{\ell}u_{\gamma_i}
 \le e^{|\gamma_0|}.
 \label{eq:core-pinned-majorant}
\end{equation}
Fix a clause $a\in\C^{\good}$ and put
\begin{equation}
 S_a\triangleq
 \sum_{\Gamma:a\in U_\Gamma}|a_\Gamma|e^{\tau K_\Gamma}.
 \label{eq:core-clause-root-sum}
\end{equation}
By \eqref{eq:core-cluster-data} and the definition of $u_\gamma$,
every tuple $\Gamma=\left(\gamma_1,\ldots,\gamma_\ell\right)$ satisfies
\[
 |a_\Gamma|e^{\tau K_\Gamma}
 =\frac{|\varphi(\Gamma)|}{\ell!}
      \prod_{i=1}^{\ell}u_{\gamma_i}.
\]
Also, $a\in U_\Gamma$ if and only if $a\in\gamma_i$ for at least one
position $i$. 
Thus,
\[
 \id{a\in U_\Gamma}
 \le\sum_{i=1}^{\ell}\id{a\in\gamma_i}.
\]
The right-hand side counts every position whose polymer contains $a$,
including distinct occurrences of the same polymer. Since all terms
below are nonnegative, this gives
\begin{align}\label{eq:core-clause-root-marking}
 S_a=\sum_{\ell\ge1}\frac1{\ell!}
    \sum_{\gamma_1,\ldots,\gamma_\ell}
      \id{a\in U_\Gamma}|\varphi(\Gamma)|
      \prod_{i=1}^{\ell}u_{\gamma_i}\le\sum_{\ell\ge1}\frac1{\ell!}
    \sum_{\gamma_1,\ldots,\gamma_\ell}
      \left(\sum_{i=1}^{\ell}\id{a\in\gamma_i}\right)
      |\varphi(\Gamma)|\prod_{i=1}^{\ell}u_{\gamma_i}.
\end{align}
For $\ell\ge1$ and $i\in[\ell]$, define the sum associated with
position $i$ by
\[
 T_{\ell,i}(a)\triangleq
 \sum_{\gamma_1,\ldots,\gamma_\ell}
    \id{a\in\gamma_i}
    |\varphi\left(\gamma_1,\ldots,\gamma_\ell\right)|
    \prod_{j=1}^{\ell}u_{\gamma_j}.
\]
Swapping positions $i$ and $1$ bijectively maps the tuples with
$a\in\gamma_i$ to those with $a\in\gamma_1$, while preserving
the Ursell coefficient and the product $\prod_j u_{\gamma_j}$. Hence,
\[
 T_{\ell,i}(a)=T_{\ell,1}(a)
 \qquad\text{for every }\ell\ge1\text{ and }i\in[\ell].
\]
Writing the right-hand side of \eqref{eq:core-clause-root-marking}
in terms of these sums therefore gives
\begin{equation}
 \begin{aligned}
 S_a
 \le\sum_{\ell\ge1}\frac1{\ell!}
        \sum_{i=1}^{\ell}T_{\ell,i}(a)=\sum_{\ell\ge1}\frac{\ell}{\ell!}T_{\ell,1}(a)
  =\sum_{\ell\ge1}\frac1{(\ell-1)!}T_{\ell,1}(a).
 \end{aligned}
 \label{eq:core-first-position-sum}
\end{equation}
In $T_{\ell,1}(a)$, rename the first polymer $\gamma_1$ as $\gamma_0$
and factor out $u_{\gamma_0}$. This yields
\[
 T_{\ell,1}(a)
 =\sum_{\gamma_0:\,a\in\gamma_0}u_{\gamma_0}
   \sum_{\gamma_2,\ldots,\gamma_\ell}
    |\varphi\left(\gamma_0,\gamma_2,\ldots,\gamma_\ell\right)|
    \prod_{j=2}^{\ell}u_{\gamma_j}.
\]
Substituting this identity into \eqref{eq:core-first-position-sum}
and interchanging the nonnegative sums over $\ell$ and $\gamma_0$
gives
\[
 S_a\le
 \sum_{\gamma_0:\,a\in\gamma_0}u_{\gamma_0}
 \left[
  \sum_{\ell\ge1}\frac1{(\ell-1)!}
   \sum_{\gamma_2,\ldots,\gamma_\ell}
    |\varphi\left(\gamma_0,\gamma_2,\ldots,\gamma_\ell\right)|
    \prod_{i=2}^{\ell}u_{\gamma_i}
 \right].
\]
Applying \eqref{eq:core-pinned-majorant} to the bracketed sum and then
\eqref{eq:core-polymer-root} yields
\begin{equation}
 S_a\le\sum_{\gamma_0:\,a\in\gamma_0}u_{\gamma_0}e^{|\gamma_0|}
 =\sum_{\gamma_0\in\CS_\Psi(a)}
      q_{\gamma_0}e^{(1+\tau)|\gamma_0|}
 \le\rho.
 \label{eq:core-clause-root-bound}
\end{equation}
Since $a\in\C^{\good}$ was arbitrary, this proves
$\sum_{\Gamma:a\in U_\Gamma}|a_\Gamma|e^{\tau K_\Gamma}\le\rho$
for every such $a$, as asserted in \eqref{eq:core-cluster-root}.

Finally, we show that the series $\sum_\Gamma a_\Gamma Y_\Gamma(\sigma)$
converges absolutely and uniformly on $\Omega$.
Since every polymer is a nonempty subset of $\C^{\good}$,
the definition of $U_\Gamma$ in \eqref{eq:core-cluster-data} gives
$\emptyset\ne U_\Gamma\subseteq\C^{\good}$ for every tuple.
By \eqref{eq:core-clause-root-sum},
\begin{equation}
 \begin{aligned}
 \sum_{a\in\C^{\good}}S_a&=\sum_{a\in\C^{\good}}
     \sum_{\Gamma:a\in U_\Gamma}
       |a_\Gamma|e^{\tau K_\Gamma}=\sum_{a\in\C^{\good}}\sum_\Gamma
       \id{a\in U_\Gamma}|a_\Gamma|e^{\tau K_\Gamma}\\&=\sum_\Gamma
       \left(\sum_{a\in\C^{\good}}\id{a\in U_\Gamma}\right)
       |a_\Gamma|e^{\tau K_\Gamma}=\sum_\Gamma
       |U_\Gamma|\,|a_\Gamma|e^{\tau K_\Gamma}.
 \end{aligned}
 \label{eq:core-clause-counting}
\end{equation}
Since \eqref{eq:core-polymer-root} implies $\tau\ge0$ and
\eqref{eq:core-cluster-data} gives $K_\Gamma\ge1$, we have
$e^{\tau K_\Gamma}\ge1$. Using $|U_\Gamma|\ge1$,
\eqref{eq:core-clause-counting}, and the bound $S_a\le\rho$ just
proved in \eqref{eq:core-clause-root-bound}, we obtain
\begin{align*}
 \sum_\Gamma|a_\Gamma|\le\sum_\Gamma|a_\Gamma|e^{\tau K_\Gamma}\le\sum_\Gamma
       |U_\Gamma|\,|a_\Gamma|e^{\tau K_\Gamma}=\sum_{a\in\C^{\good}}S_a\le |\C^{\good}|\rho\le m\rho.
\end{align*}
The last inequality uses $\C^{\good}\subseteq\C$ and $|\C|=m$.
By \eqref{eq:core-cluster-data}, $|Y_\Gamma(\sigma)|\le1$ for every
$\sigma\in\Omega$. Consequently,
\[
 \sum_\Gamma\sup_{\sigma\in\Omega}|a_\Gamma Y_\Gamma(\sigma)|
 \le\sum_\Gamma|a_\Gamma|\le m\rho<\infty.
\]
The Weierstrass $M$-test therefore gives absolute and uniform
convergence of $\sum_\Gamma a_\Gamma Y_\Gamma(\sigma)$ on $\Omega$.
\end{proof}

\subsection{Variance tensorization for finite-interaction measures}
\label{subsec:core-swap-comparison}

The cluster expansion above represents $\log Z(\sigma)$ as a sum of
polymer-cluster interactions. In any finite truncation, terms involving no
core component contribute only a constant, terms involving one component
are absorbed into its local weight, and terms coupling multiple components
become labelled hyperedges. The following lemma gives a variance bound for
measures with finitely many such interactions. It controls the global
variance by the sum of the coordinatewise conditional variances, with a
prefactor determined by the expected sizes and incident interaction
strengths of connected components in an associated random hypergraph.
We
will apply it to finite truncations of the cluster expansion and then pass
to the limit.

\begin{lemma}[Component-swap variance comparison]
\label{lem:core-swap-comparison}
Let $s\ge1$ be an integer, let $\Omega_1,\ldots,\Omega_s$ be nonempty
finite sets, and put $\Omega=\prod_{i=1}^s\Omega_i$.
For each $i\in[s]$, let $\ell_i:\Omega_i\to(0,\infty)$ be a local
weight function. Let $\mathcal J$ be a finite set of labelled hyperedges
on vertex set $[s]$. Each $e\in\mathcal J$ has a nonempty vertex set
$V_e\subseteq[s]$; distinct hyperedges may have the same vertex set,
but are distinguished by their labels.
For each $e\in\mathcal J$, fix a real coefficient $a_e\in\mathbb R$
and an indicator function
$Y_e:\prod_{i\in V_e}\Omega_i\to\{0,1\}$.
Define the probability measure $\pi$ on $\Omega$ by
\begin{equation}
 \pi(\sigma)=\frac1{Z_\pi}
 \left(\prod_{i=1}^s\ell_i\left(\sigma_i\right)\right)
 \exp\left(\sum_{e\in\mathcal J}a_eY_e\left(\sigma_{V_e}\right)\right),
 \qquad \sigma\in\Omega,
 \label{eq:core-swap-measure}
\end{equation}
where $Z_\pi$ is the normalizing constant.
This measure is well defined and strictly positive. For $A\subseteq[s]$,
define
\begin{equation}
 R(A)=\sum_{\substack{e\in\mathcal J\\V_e\cap A\ne\emptyset}}|a_e|.
 \label{eq:core-swap-oscillation}
\end{equation}
Let $\zeta$ be a random labelled hypergraph on vertex set $[s]$, and write
$\zeta_e=\id{e\in\zeta}$ for its edge indicators. These indicators are
independent and satisfy
\begin{equation}
 \forall e\in\mathcal J,\qquad
 \mathbb P_\zeta(\zeta_e=1)=1-e^{-2|a_e|}.
 \label{eq:core-zeta-law}
\end{equation}
For $j\in[s]$, let $\mathcal A_j$ be the vertex set of the connected
component of $\zeta$ containing $j$; in particular,
$\mathcal A_j=\{j\}$ when $j$ is isolated.
Then, for every $f:\Omega\to\mathbb R$,
\begin{equation}
 \operatorname{Var}_\pi(f)
 \le\left(
   \max_{j\in[s]}\mathbb E_\zeta
   \left[2^{|\mathcal A_j|}e^{R\left(\mathcal A_j\right)}\right]
   \right)\sum_{j=1}^s V_j^\pi(f).
 \label{eq:core-swap-tensorization}
\end{equation}
\end{lemma}

If $a_e=0$, then the label $e$ contributes neither to $\pi$ nor to
$R(A)$, and its hyperedge is absent from $\zeta$ almost surely because
$1-e^{-2|a_e|}=0$. We therefore remove all such labels. For every remaining
$e\in\mathcal J$, define $w_e$ and the function
$h_e:\prod_{i\in V_e}\Omega_i\to\{0,1\}$ by
\begin{equation}
 w_e=e^{|a_e|}-1,\qquad
 h_e\left(\sigma_{V_e}\right)=
 \begin{cases}
  Y_e\left(\sigma_{V_e}\right),&a_e>0,\\
  1-Y_e\left(\sigma_{V_e}\right),&a_e<0,
 \end{cases}
 \qquad \sigma\in\Omega.
 \label{eq:core-positive-auxiliary-definitions}
\end{equation}
Then we have the following lemma.

\begin{lemma}\label{lem-pisigma-expansion}
For all $\sigma\in\Omega$, we have
$\pi(\sigma)\propto\left(\prod_{i=1}^s\ell_i\left(\sigma_i\right)\right)
          \prod_{e\in\mathcal J}\bigl(1+w_eh_e\left(\sigma_{V_e}\right)\bigr)$.
\end{lemma}

Let $\eta=\left(\eta_e\right)_{e\in\mathcal J}\in\{0,1\}^{\mathcal J}$ and define
\begin{equation}
 \widehat\pi(\sigma,\eta)\propto
 \prod_{i=1}^s\ell_i\left(\sigma_i\right)
 \prod_{e:\eta_e=1}w_eh_e\left(\sigma_{V_e}\right).
 \label{eq:core-positive-lift}
\end{equation}
Then we have the following result.

\begin{lemma}\label{lem-pisigmatau-expansion}
For each $\sigma\in\Omega$,
\begin{equation}
 \sum_{\eta\in\{0,1\}^{\mathcal J}}
 \widehat\pi(\sigma,\eta)=\pi(\sigma).
 \label{eq:core-positive-lift-exact-marginal}
\end{equation}  
Moreover, for each $\sigma\in\Omega$ and $\eta\in\{0,1\}^{\mathcal J}$,
\begin{equation}
\begin{aligned}
 \widehat\pi\left(\eta\mid X=\sigma\right)
 =\prod_{e:\eta_e=1}
   \frac{w_eh_e\left(\sigma_{V_e}\right)}
        {1+w_eh_e\left(\sigma_{V_e}\right)}
  \prod_{e:\eta_e=0}
   \frac{1}{1+w_eh_e\left(\sigma_{V_e}\right)}.
\end{aligned}
 \label{eq:core-positive-lift-conditional-factorization}
\end{equation}
\end{lemma}

Draw independent copies $(X,\eta),\left(X',\xi\right)\sim\widehat\pi$. Let $g$ be
the labelled hypergraph on $[s]$ whose edge $e$ is present exactly when
$\eta_e\vee\xi_e=1$.
For $e\in\mathcal J$, we call the hyperedge labelled $e$ \emph{open}
in $g$ if $\eta_e\vee\xi_e=1$, and \emph{closed} if $\eta_e=\xi_e=0$.
For a hypergraph $H$ on $[s]$, let $\operatorname{Comp}(H)$ denote the
collection of vertex sets of its connected components, including
singleton sets for isolated vertices. These sets form a partition of
$[s]$. 
Then we have the following lemma.

\begin{lemma}
\label{lem:core-random-component-flip}
For every
$C\in\operatorname{Comp}(g)$, let $\varepsilon_C$ be an independent
fair bit. Let $X^\varepsilon$ be the configuration obtained by swapping
exactly those components $C$ for which $\varepsilon_C=1$, and let
$\varepsilon^C$ denote the vector obtained by flipping only the bit
indexed by $C$.
Then, for every $f:\Omega\to\mathbb R$,
\begin{equation}
 \operatorname{Var}_\pi(f)
 \le\frac12\mathbb E\!\left[
   \sum_{C\in\operatorname{Comp}(g)}\mathbb E_\varepsilon\!\left[
     \left(f\left(X^\varepsilon\right)
           -f\left(X^{\varepsilon^C}\right)\right)^2\right]\right].
 \label{eq:core-component-forms}
\end{equation}
\end{lemma}

The proofs of Lemmas \ref{lem-pisigma-expansion}--\ref{lem:core-random-component-flip} are provided in the appendix. Now we can prove the main result of this subsection.
\begin{proof}[Proof of Lemma~\ref{lem:core-swap-comparison}]
For each nonempty $C\subseteq[s]$, define
\begin{equation}
 \theta_C=\sup_{\sigma\in\Omega}
   \mathbb P\left(C\in\operatorname{Comp}(g)\mid X=\sigma\right),
 \qquad
 \widetilde f_C(\sigma)
 =f(\sigma)-\mathbb E_\pi\!\left[f(X)\mid X_{-C}=\sigma_{-C}\right].
 \label{eq:core-centered-block-function}
\end{equation}
The probability defining $\theta_C$ conditions only on the first
configuration $X=\sigma$, not on $\eta$ or the second sample
$\left(X',\xi\right)$. The remaining randomness comes from $\eta$
with its conditional law given $X=\sigma$ and the independent second
sample $\left(X',\xi\right)$ with law $\widehat\pi$. 
For $C\in\operatorname{Comp}(g)$, passing from $\varepsilon$ to
$\varepsilon^C$ flips only the swap bit for $C$ and leaves the bits
for all other components unchanged. Since the components are disjoint,
every coordinate outside $C$ still takes its value from the same
configuration, $X$ or $X'$. Thus, pointwise,
\[
 \left(X^{\varepsilon^C}\right)_{-C}
 =\left(X^\varepsilon\right)_{-C}.
\]
Consequently, 
\[
 \begin{aligned}
 &\mathbb E_\pi\!\left[
   f(X)\mid X_{-C}=\left(X^\varepsilon\right)_{-C}\right]=\mathbb E_\pi\!\left[
   f(X)\mid X_{-C}=\left(X^{\varepsilon^C}\right)_{-C}\right].
 \end{aligned}
\]
Substituting into the definition of $\widetilde f_C$ in
\eqref{eq:core-centered-block-function} yields 
\[
 \widetilde f_C\left(X^\varepsilon\right)
 -\widetilde f_C\left(X^{\varepsilon^C}\right)
 =f\left(X^\varepsilon\right)-f\left(X^{\varepsilon^C}\right).
\]
For a fixed sample pair with $C\in\operatorname{Comp}(g)$, the uniform
law of $\varepsilon$ is invariant under
$\varepsilon\mapsto\varepsilon^C$. Thus
\begin{equation}
 \begin{aligned}
 &\quad\quad\mathbb E_\varepsilon\!\left[
   \left(f\left(X^\varepsilon\right)
         -f\left(X^{\varepsilon^C}\right)\right)^2\right]=\mathbb E_\varepsilon\!\left[
   \left(\widetilde f_C\left(X^\varepsilon\right)
         -\widetilde f_C\left(X^{\varepsilon^C}\right)\right)^2\right]\\
 &\quad\le2\mathbb E_\varepsilon\!\left[
   \left(\widetilde f_C\left(X^\varepsilon\right)\right)^2
   +\left(\widetilde f_C\left(X^{\varepsilon^C}\right)\right)^2\right]
 =4\mathbb E_\varepsilon\!\left[
   \left(\widetilde f_C\left(X^\varepsilon\right)\right)^2\right].
 \end{aligned}
 \label{eq:core-centered-swap-bound}
\end{equation}
Hence, \eqref{eq:core-component-forms} gives
\begin{equation}
 \operatorname{Var}_\pi(f)
 \le2\mathbb E\!\left[
   \sum_{C\in\operatorname{Comp}(g)}\mathbb E_\varepsilon\!\left[
     \left(\widetilde f_C\left(X^\varepsilon\right)\right)^2\right]\right].
 \label{eq:core-centered-component-form}
\end{equation}

Conditional on $g$, the collection $\operatorname{Comp}(g)$ is fixed.
Fix $g$ and $\varepsilon$.
For each $\sigma\in\Omega$, by \eqref{eq:core-conditional-swapped-pair-law},
\[
 \begin{aligned}
 \mathbb P\left(X^\varepsilon=\sigma\mid g,\varepsilon\right)
 &=\sum_{\omega\in\Omega}\mathbb P\left(
   X^\varepsilon=\sigma,\left(X'\right)^\varepsilon=\omega
   \mid g,\varepsilon\right)=\sum_{\omega\in\Omega}\mathbb P\left(
   X=\sigma,X'=\omega\mid g\right)\\
 &=\mathbb P\left(X=\sigma\mid g\right).
 \end{aligned}
\]
By the law of total probability conditional on $g$,
\begin{equation}
 \forall\sigma\in\Omega,\qquad
 \mathbb P\left(X^\varepsilon=\sigma\mid g\right)
 =\mathbb P\left(X=\sigma\mid g\right).
 \label{eq:core-conditional-swapped-marginal-law}
\end{equation}

Condition on $g$, so that $\operatorname{Comp}(g)$ is fixed.
By \eqref{eq:core-centered-component-form} and \eqref{eq:core-conditional-swapped-marginal-law},
\begin{equation}
 \operatorname{Var}_\pi(f)
 \le2\mathbb E\!\left[
   \sum_{C\in\operatorname{Comp}(g)}
     \left(\widetilde f_C\left(X\right)\right)^2\right].
 \label{eq:core-unswapped-component-form}
\end{equation}
Since the marginal law of $X$ is $\pi$, the law of total expectation gives
\begin{equation}
 \mathbb E\!\left[
   \sum_{C\in\operatorname{Comp}(g)}
     \left(\widetilde f_C\left(X\right)\right)^2\right]
 =\sum_{\emptyset\ne C\subseteq[s]}\mathbb E_\pi\!\left[
   \mathbb P\left(C\in\operatorname{Comp}(g)\mid X\right)
   \left(\widetilde f_C\left(X\right)\right)^2\right].
 \label{eq:core-component-conditioning}
\end{equation}
Substituting \eqref{eq:core-component-conditioning} into
\eqref{eq:core-unswapped-component-form} and using the definition of
$\theta_C$ in \eqref{eq:core-centered-block-function} yields
\begin{equation}
 \operatorname{Var}_\pi(f)
 \le2\sum_{\emptyset\ne C\subseteq[s]}
   \theta_C\mathbb E_\pi\!\left[
     \left(\widetilde f_C\left(X\right)\right)^2\right].
 \label{eq:core-weighted-centered-bound}
\end{equation}
By the definition of $\widetilde f_C$ in
\eqref{eq:core-centered-block-function} and the definition of
conditional variance, for every nonempty $C\subseteq[s]$,
\begin{equation}
 \begin{aligned}
 \mathbb E_\pi\!\left[
   \left(\widetilde f_C\left(X\right)\right)^2\mid X_{-C}\right]
 &=\mathbb E_\pi\!\left[
   \left(f(X)-\mathbb E_\pi\!\left[f(X)\mid X_{-C}\right]\right)^2
   \mid X_{-C}\right]=\operatorname{Var}_\pi\left(f(X)\mid X_{-C}\right).
 \end{aligned}
 \label{eq:core-centered-block-second-moment}
\end{equation}
The law of total expectation, followed by
\eqref{eq:core-centered-block-second-moment}, gives
\begin{equation}
 \begin{aligned}
 \mathbb E_\pi\!\left[
   \left(\widetilde f_C\left(X\right)\right)^2\right]
 &=\mathbb E_\pi\!\left[
   \mathbb E_\pi\!\left[
     \left(\widetilde f_C\left(X\right)\right)^2
     \mid X_{-C}\right]\right]=\mathbb E_\pi\!\left[
   \operatorname{Var}_\pi\left(f(X)\mid X_{-C}\right)\right].
 \end{aligned}
 \label{eq:core-centered-block-averaged-variance}
\end{equation}
Substituting \eqref{eq:core-centered-block-averaged-variance} into
\eqref{eq:core-weighted-centered-bound} gives
\begin{equation}
 \operatorname{Var}_\pi(f)
 \le2\sum_{\emptyset\ne C\subseteq[s]}\theta_C\mathbb E_\pi\!\left[
   \operatorname{Var}_\pi\left(f(X)\mid X_{-C}\right)\right].
 \label{eq:core-block-form}
\end{equation}

Fix the outside configuration $X_{-C}=y$, write $\pi_C^y$ for its conditional law on $C$, and set $f_y(x_C)=f(x_C,y)$. Let $\nu_C(x_C)\propto\prod_{i\in C}\ell_i(x_i)$. In the density $q_y=\pi_C^y/\nu_C$, all interaction factors disjoint from $C$ are constant. Each remaining indicator can change the log density by at most $|a_e|$. Thus, with $U_y=\max q_y$ and $L_y=\min q_y>0$, we have $U_y/L_y\le e^{R(C)}$.

For any law $\rho_C$ on $C$, the averaged conditional variance in coordinate $j\in C$ satisfies
\[
 V_j^{\rho_C}(f_y)=\inf_\psi
     \mathbb E_{\rho_C}\!\left[(f_y(X_C)-\psi(X_{C\setminus\{j\}}))^2\right].
\]
The analogous infimum over constants gives $\operatorname{Var}_{\rho_C}(f_y)$. Since $L_y\nu_C\le\pi_C^y\le U_y\nu_C$, these formulas and the Efron--Stein inequality for the product measure $\nu_C$ give
\[
 \operatorname{Var}_{\pi_C^y}(f_y)
 \le U_y\operatorname{Var}_{\nu_C}(f_y)
 \le U_y\sum_{j\in C}V_j^{\nu_C}(f_y)
 \le e^{R(C)}\sum_{j\in C}V_j^{\pi_C^y}(f_y).
\]
Averaging over $y$ gives
\begin{equation}\label{eq:core-local-variance-comparison}
 \mathbb E_\pi\!\left[\operatorname{Var}_\pi(f(X)\mid X_{-C})\right]
 \le e^{R(C)}\sum_{j\in C}V_j^\pi(f),
\end{equation}
because conditioning on $X_{-C}$ and $X_{C\setminus\{j\}}$ fixes exactly $X_{-j}$.

To bound $\theta_C$ defined in
\eqref{eq:core-centered-block-function}, put $p_e=1-e^{-|a_e|}$.
By \eqref{eq:core-positive-lift-conditional-factorization}, for every
$e\in\mathcal J$,
\begin{align}
 \mathbb P\left(\eta_e=1\mid X=\sigma\right)
 &\quad=\sum_{\eta_{-e}\in
       \{0,1\}^{\mathcal J\setminus\{e\}}}
   \widehat\pi\left(\eta_e=1,\eta_{-e}\mid X=\sigma\right)\notag\\
 &\quad=\frac{w_eh_e\left(\sigma_{V_e}\right)}
              {1+w_eh_e\left(\sigma_{V_e}\right)}
   \prod_{d\in\mathcal J\setminus\{e\}}
   \left(
    \frac{1}{1+w_dh_d\left(\sigma_{V_d}\right)}
    +\frac{w_dh_d\left(\sigma_{V_d}\right)}
           {1+w_dh_d\left(\sigma_{V_d}\right)}
   \right)\notag\\
 &\quad=\frac{w_eh_e\left(\sigma_{V_e}\right)}
              {1+w_eh_e\left(\sigma_{V_e}\right)}.
 \label{eq:core-lift-single-edge-probability}
\end{align}
Since
$h_e\left(\sigma_{V_e}\right)\in\{0,1\}$, we have
\begin{equation}
 \mathbb P\left(\eta_e=1\mid X=\sigma\right) = \frac{w_eh_e\left(\sigma_{V_e}\right)}
      {1+w_eh_e\left(\sigma_{V_e}\right)}
 =\begin{cases}
   0,
      &h_e\left(\sigma_{V_e}\right)=0,\\[2mm]
   \displaystyle\frac{w_e}{1+w_e},
      &h_e\left(\sigma_{V_e}\right)=1.
  \end{cases}
 \label{eq:core-lift-single-edge-cases}
\end{equation}
Using $w_e=e^{|a_e|}-1$ from
\eqref{eq:core-positive-auxiliary-definitions} and
$p_e=1-e^{-|a_e|}$, we have
\begin{equation}
 \frac{w_e}{1+w_e}
 =\frac{e^{|a_e|}-1}{e^{|a_e|}}
 =1-e^{-|a_e|}=p_e.
 \label{eq:core-p-w-relation}
\end{equation}
Substitution into \eqref{eq:core-lift-single-edge-cases} therefore yields
\begin{equation}
 \mathbb P\left(\eta_e=1\mid X=\sigma\right)
 =p_eh_e\left(\sigma_{V_e}\right).
 \label{eq:core-lift-bernoulli-parameter}
\end{equation}
Likewise, for every $\omega\in\Omega$,
\begin{equation}
 \mathbb P\left(\xi_e=1\mid X'=\omega\right)
 =p_eh_e\left(\omega_{V_e}\right).
 \label{eq:core-second-lift-bernoulli-parameter}
\end{equation}

To compare $g$ with the random hypergraph $\zeta$ defined in
\eqref{eq:core-zeta-law}, we place all relevant random objects on a
common probability space. In this construction, $(X,\eta)$ and
$(X',\xi)$ are independent, each has distribution $\widehat\pi$, and
the auxiliary random hypergraph has the distribution of $\zeta$.

First draw $X,X'$ independently from $\pi$. For every label
$e\in\mathcal J$, introduce two auxiliary random variables $U_e,U'_e$,
each uniformly distributed on $[0,1]$. All $2|\mathcal J|$ variables
are mutually independent and are independent of $(X,X')$.
Define 
\begin{equation}
\forall e\in\mathcal J,\qquad \eta_e=\id{U_e\le p_eh_e\left(X_{V_e}\right)},\qquad
 \xi_e=\id{U'_e\le p_eh_e\left(X'_{V_e}\right)}.
 \label{eq:core-lift-uniform-realization}
\end{equation}
Indeed, $0\le p_eh_e\left(\sigma_{V_e}\right)\le p_e<1$, so independence
of $U_e$ from $X$ and uniformity of $U_e$ give
\[
 \mathbb P\left(\eta_e=1\mid X=\sigma\right)
 =\mathbb P\left(U_e\le p_eh_e\left(\sigma_{V_e}\right)\right)
 =p_eh_e\left(\sigma_{V_e}\right).
\]
Thus, \eqref{eq:core-lift-bernoulli-parameter} is satisfied. Moreover,
conditional on $X=\sigma$, each quantity
$p_eh_e\left(\sigma_{V_e}\right)$ in the comparison
$U_e\le p_eh_e\left(\sigma_{V_e}\right)$ is deterministic. Since the
variables $(U_e)_{e\in\mathcal J}$ are independent, the bits
$(\eta_e)_{e\in\mathcal J}$ are therefore conditionally independent.
Consequently, for every
$\varepsilon\in\{0,1\}^{\mathcal J}$,
\begin{align}
 \mathbb P\left(\eta=\varepsilon\mid X=\sigma\right)
 &=\prod_{e:\varepsilon_e=1}
      p_eh_e\left(\sigma_{V_e}\right)
   \prod_{e:\varepsilon_e=0}
      \left(1-p_eh_e\left(\sigma_{V_e}\right)\right).
 \label{eq:core-uniform-lift-product-law}
\end{align}
Since $h_e\left(\sigma_{V_e}\right)\in\{0,1\}$, using
\eqref{eq:core-p-w-relation}, the two Bernoulli probabilities can be
rewritten as
\[
 p_eh_e\left(\sigma_{V_e}\right)
 =\frac{w_eh_e\left(\sigma_{V_e}\right)}
        {1+w_eh_e\left(\sigma_{V_e}\right)},
 \qquad
 1-p_eh_e\left(\sigma_{V_e}\right)
 =\frac{1}{1+w_eh_e\left(\sigma_{V_e}\right)}.
\]
Substituting these identities into
\eqref{eq:core-uniform-lift-product-law} gives
\[
 \mathbb P\left(\eta=\varepsilon\mid X=\sigma\right)
 =\prod_{e:\varepsilon_e=1}
   \frac{w_eh_e\left(\sigma_{V_e}\right)}
        {1+w_eh_e\left(\sigma_{V_e}\right)}
  \prod_{e:\varepsilon_e=0}
   \frac{1}{1+w_eh_e\left(\sigma_{V_e}\right)},
\]
which is exactly the conditional distribution in
\eqref{eq:core-positive-lift-conditional-factorization}.

By \eqref{eq:core-positive-lift-exact-marginal}, the $X$-marginal of
$\widehat\pi$ is $\pi$. Therefore, the conditional probability formula
gives, for every $\sigma\in\Omega$ and
$\varepsilon\in\{0,1\}^{\mathcal J}$,
\begin{equation}
 \widehat\pi(\sigma,\varepsilon)
 =\pi(\sigma)\widehat\pi\left(\varepsilon\mid X=\sigma\right).
 \label{eq:core-lift-joint-factorization}
\end{equation}
In our construction, $X\sim\pi$, and the conditional law of $\eta$ given
$X=\sigma$ is the same as the conditional law under $\widehat\pi$.
Consequently,
\begin{equation}
 \mathbb P\left(X=\sigma,\eta=\varepsilon\right)
 =\pi(\sigma)\widehat\pi\left(\varepsilon\mid X=\sigma\right)
 =\widehat\pi(\sigma,\varepsilon),
 \label{eq:core-coupled-lift-law}
\end{equation}
where the last equality is \eqref{eq:core-lift-joint-factorization}.
The same argument shows that $(X',\xi)\sim\widehat\pi$.
Moreover, $(X,(U_e)_{e\in\mathcal J})$ and
$(X',(U'_e)_{e\in\mathcal J})$ are independent, so $(X,\eta)$ and
$(X',\xi)$ are independent random pairs, each with joint distribution
$\widehat\pi$. Thus they have exactly the distribution assumed when
$g$ was defined. Define $g$ from these bits as before, with label $e$
open precisely when $\eta_e\vee\xi_e=1$.

Next, define
\begin{equation}
 \zeta_e=\id{U_e\le p_e\ \text{or}\ U'_e\le p_e},
 \qquad e\in\mathcal J,
 \label{eq:core-dominating-edge-indicator}
\end{equation}
and let $\zeta$ be the labelled hypergraph that includes $e$ exactly
when $\zeta_e=1$. By independence and uniformity of $U_e,U'_e$,
\begin{equation}
 \begin{aligned}
 \mathbb P(\zeta_e=1)
 &=1-\mathbb P\left(U_e>p_e,\ U'_e>p_e\right)=1-(1-p_e)^2
 =1-e^{-2|a_e|}.
 \end{aligned}
 \label{eq:core-dominating-edge-probability}
\end{equation}
Different indicators $\zeta_e$ depend on disjoint pairs $(U_e,U'_e)$,
so they are mutually independent. They are also independent of
$(X,X')$, since their definitions use only the auxiliary variables.
Consequently, this $\zeta$ has exactly the independent-hypergraph law
specified in the lemma.

Recall $0\le p_eh_e\left(\sigma_{V_e}\right)\le p_e<1$. By \eqref{eq:core-lift-uniform-realization} and \eqref{eq:core-dominating-edge-indicator},
\[
 \begin{aligned}
 \eta_e=1
 &\ \Longrightarrow\ U_e\le p_eh_e\left(X_{V_e}\right)\le p_e
 \ \Longrightarrow\ \zeta_e=1,\\
 \xi_e=1
 &\ \Longrightarrow\ U'_e\le p_eh_e\left(X'_{V_e}\right)\le p_e
 \ \Longrightarrow\ \zeta_e=1.
 \end{aligned}
\]
It follows that the set of
labelled hyperedges of $g$ is a subset of the set of labelled hyperedges
of $\zeta$. 
Recall that $\mathcal A_j$ is the vertex set of the connected component
of $\zeta$ containing $j$.
Fix a nonempty deterministic set $C\subseteq[s]$ and a vertex $j\in C$.
On the event $\{C\in\operatorname{Comp}(g)\}$, the set $C$ is the vertex
set of a connected component of $g$. Hence every vertex of $C$ is
connected to $j$ in $g$, and therefore also in $\zeta$. 
In other words, 
$C\in\operatorname{Comp}(g)$ only if $C\subseteq\mathcal A_j$.
Consequently, for every $\sigma\in\Omega$,
\[
 \begin{aligned}
 \mathbb P\left(C\in\operatorname{Comp}(g)\mid X=\sigma\right)
 \le\mathbb P\left(C\subseteq\mathcal A_j\mid X=\sigma\right)=\mathbb P_\zeta\left(C\subseteq\mathcal A_j\right),
 \end{aligned}
\]
where the equality uses independence of $\zeta$ from $X$.
Using the definition of $\theta_C$ in
\eqref{eq:core-centered-block-function}, for every $j\in C$,
\begin{align}
 \theta_C
 &=\sup_{\sigma\in\Omega}
   \mathbb P\left(C\in\operatorname{Comp}(g)\mid X=\sigma\right)
 \le\sup_{\sigma\in\Omega}
   \mathbb P_\zeta\left(C\subseteq\mathcal A_j\right)
 =\mathbb P_\zeta\left(C\subseteq\mathcal A_j\right).
 \label{eq:core-theta-component-bound}
\end{align}

Finally, substituting \eqref{eq:core-local-variance-comparison} into
\eqref{eq:core-block-form} gives
\begin{align}
 \operatorname{Var}_\pi(f)
 &\le 2\sum_{\emptyset\ne C\subseteq[s]}
   \theta_C\mathbb E_\pi\!\left[
   \operatorname{Var}_\pi\left(f(X)\mid X_{-C}\right)\right]\le 2\sum_{\emptyset\ne C\subseteq[s]}
   \theta_Ce^{R(C)}\sum_{j\in C}V_j^\pi(f)
 \notag\\
 &=2\sum_{j=1}^s V_j^\pi(f)
   \sum_{\substack{\emptyset\ne C\subseteq[s]\\j\in C}}
   \theta_Ce^{R(C)}=\sum_{j=1}^s
   \left(
   2\sum_{\substack{\emptyset\ne C\subseteq[s]\\j\in C}}
   \theta_Ce^{R(C)}\right)V_j^\pi(f).
 \label{eq:core-collected-variance-bound}
\end{align}
Fix $j\in[s]$. By \eqref{eq:core-theta-component-bound}, we have
\begin{align}
 &\quad 2\sum_{\substack{\emptyset\ne C\subseteq[s]\\j\in C}}
   \theta_Ce^{R(C)}
\le 2\sum_{\substack{\emptyset\ne C\subseteq[s]\\j\in C}}
   \mathbb P_\zeta\left(C\subseteq\mathcal A_j\right)e^{R(C)}
 =2\mathbb E_\zeta\!\left[
   \sum_{\substack{\emptyset\ne C\subseteq[s]\\j\in C}}
   \id{C\subseteq\mathcal A_j}e^{R(C)}\right]
 \notag\\
 &=2\mathbb E_\zeta\!\left[
   \sum_{\substack{C\subseteq\mathcal A_j\\j\in C}}e^{R(C)}\right]
   \le2\mathbb E_\zeta\!\left[
   2^{|\mathcal A_j|-1}e^{R\left(\mathcal A_j\right)}\right]=\mathbb E_\zeta\!\left[
   2^{|\mathcal A_j|}e^{R\left(\mathcal A_j\right)}\right]
 \le\max_{i\in[s]}\mathbb E_\zeta\!\left[
   2^{|\mathcal A_i|}e^{R\left(\mathcal A_i\right)}\right].
 \label{eq:core-coefficient-bound}
\end{align}
Indeed, for each realization of $\mathcal A_j$, exactly
$2^{|\mathcal A_j|-1}$ subsets of $\mathcal A_j$ contain $j$, and the
definition of $R$ in \eqref{eq:core-swap-oscillation} implies
$R(C)\le R(\mathcal A_j)$ whenever
$C\subseteq\mathcal A_j$.
Substitution into \eqref{eq:core-collected-variance-bound} proves
\eqref{eq:core-swap-tensorization}.
\end{proof}

\subsection{The exploration bound and the spectral gap}
\label{subsec:core-exploration}
Recall that $\Gamma=(\gamma_1,\ldots,\gamma_\ell)$ is the ordered polymer
tuple whose associated data $U_\Gamma,K_\Gamma,a_\Gamma,Y_\Gamma$ are
defined in \eqref{eq:core-cluster-data}.
For each tuple $\Gamma$, let $\vbl\left(U_\Gamma\right)$ denote the set
of all variables appearing in the original clauses $C_a$ with
$a\in U_\Gamma$, including both good and bad variables.
Define its \emph{geometric scope} to be the set of core-component
indices touched by these clauses:
\begin{equation}
 J_\Gamma=\{i\in[s]:B_i\cap\vbl\left(U_\Gamma\right)\ne\emptyset\}.
 \label{eq:core-geometric-scope}
\end{equation}
Here $\sigma_{J_\Gamma}=\left(\sigma_i\right)_{i\in J_\Gamma}$ denotes the tuple
of assignments to the core components indexed by $J_\Gamma$.
Equivalently, it is the restriction of $\sigma$ to
$\bigcup_{i\in J_\Gamma}B_i$.
The value of $Y_\Gamma(\sigma)$ defined in \eqref{eq:core-cluster-data}
is determined by this restricted assignment.
Indeed, each factor $\chi_a(\sigma)$ in
$Y_\Gamma(\sigma)=\prod_{a\in U_\Gamma}\chi_a(\sigma)$ depends only on
the bad variables in $C_a$. By \eqref{eq:core-geometric-scope}, all
these variables belong to $\bigcup_{i\in J_\Gamma}B_i$.
We write
$Y_\Gamma\left(\sigma_{J_\Gamma}\right)$ for that value.

Every original clause meets at most one core component: all its bad
variables are mutually adjacent in $\G_\Phi$. Therefore
\begin{equation}
 |J_\Gamma|\le |U_\Gamma|\le K_\Gamma.
 \label{eq:core-scope-size}
\end{equation}
In particular, an interaction involving at least two components has
$K_\Gamma\ge2$. Retaining this lower bound in
\eqref{eq:core-cluster-root} gives, for every clause label $a$,
\begin{equation}
 \sum_{\substack{\Gamma:a\in U_\Gamma\\|J_\Gamma|\ge2}}
       |a_\Gamma|
 \le e^{-2\tau}
       \sum_{\Gamma:a\in U_\Gamma}|a_\Gamma|e^{\tau K_\Gamma}
 \le\rho_0.
 \label{eq:core-multiblock-root}
\end{equation}
The factor $e^{-2\tau}$ is essential for the improved running-time
exponent; bounding this sum only by $\rho$ would lose that precision.

Fix an arbitrary finite truncation set
$\Lambda\subseteq\{\Gamma:|J_\Gamma|\ge2\}$ with
$|\Lambda|<\infty$.
The finite set $\Lambda$ indexes the candidate labelled hyperedges on
vertex set $[s]$: the candidate with label $\Gamma\in\Lambda$ has vertex
set $J_\Gamma$. The random hypergraph $\zeta$ from
Lemma~\ref{lem:core-swap-comparison} is obtained by retaining these
candidates independently. Equivalently, we identify $\zeta$ with the
random subset
$\{\Gamma\in\Lambda:\zeta_\Gamma=1\}$,
where $\zeta_\Gamma$ is the indicator that the candidate labelled
$\Gamma$ is retained, with
$\mathbb P(\zeta_\Gamma=1)=1-e^{-2|a_\Gamma|}$ by
\eqref{eq:core-zeta-law}.
Define
\[
 R(A)=\sum_{\substack{\Gamma\in\Lambda\\
                       J_\Gamma\cap A\ne\emptyset}}|a_\Gamma|.
\]
For any $j\in[s]$, recall that $\mathcal A_j$ is the vertex set of the
connected component of $\zeta$ containing $j$.
The main result of this subsection is the following lemma.

\begin{lemma}\label{lem-sizeajraj}
$\max_j\mathbb E_\zeta
       \left[2^{|\mathcal A_j|}e^{R\left(\mathcal A_j\right)}\right]\le 2n^{\gamma_{\mathrm{mix}}}$.
\end{lemma}

We identify
$\mathcal A_j$ by an ordinary component search, using good clauses to
locate the hyperedges that touch the part already found.
Start with the single block index $j$. Choose a good clause $a$ that
contains a variable from one of the blocks found so far and has not been
used before. For every label $\Gamma$ not checked earlier
with $a\in U_\Gamma$, check whether $\zeta_\Gamma=1$. If it is, the
hyperedge $J_\Gamma$ is present in $\zeta$, so add all indices in
$J_\Gamma$ to the set found so far; if $\zeta_\Gamma=0$, add nothing.
Repeat this until every good clause incident to the blocks found so far
has been used. Thus the procedure checks only hyperedges that touch the
growing component, although in the worst case this may include all
labels in $\Lambda$.

We use a fixed deterministic order whenever several clauses are
available. Let $a_t$ be the clause chosen at the $t$-th step,
$\mathcal R_t$ the block indices found before that step, and
$\mathcal E_t$ the interaction labels checked before that step. Put
\[
 \mathcal R_0=\{j\},\qquad
 \mathcal E_0=\emptyset.
\]
The labels checked for the first time at step $t$ are
\begin{equation}
 \mathcal D_t=
 \left\{\Gamma\in\Lambda\setminus\mathcal E_t:
   a_t\in U_\Gamma\right\}.
 \label{eq:core-exploration-new-labels}
\end{equation}
After checking their indicators, set
\[
 \begin{aligned}
  \mathcal E_{t+1}&=\mathcal E_t\cup\mathcal D_t,\qquad \qquad
  \mathcal R_{t+1}&=\mathcal R_t\cup 
    \bigcup_{\substack{\Gamma\in\mathcal D_t\\
                        \zeta_\Gamma=1}}J_\Gamma.
 \end{aligned}
\]
Let $Q$ be the number of steps before the procedure stops, and define
\begin{align}\label{eq-defpq}
  \mathcal P_Q \triangleq \{a_0,\ldots,a_{Q-1}\}.
\end{align}
No clause is chosen twice, so
$Q=|\mathcal P_Q|\le|\C^{\good}|\le m$.
Actually, the final set $\mathcal R_Q$ is exactly $\mathcal A_j$.
\begin{lemma}\label{lem-rqequalaj}
$\mathcal R_Q=\mathcal A_j$.
\end{lemma}

Recall the definition of $K_\Gamma$ in \eqref{eq:core-cluster-data}. 
For $0\le t\le Q$, define
\[
 T_t=\sum_{\Gamma\in\mathcal E_t}K_\Gamma\zeta_\Gamma,
 \qquad T=T_Q.
\]
We have the following bound on $T$.
\begin{lemma} \label{lem:core-exploration-mgf}
$\mathbb E e^{T-15\rho_0 Q}\le1$.
\end{lemma}

For $A\subseteq[s]$, define the union of the corresponding core blocks by
\begin{equation}
 W_A=\bigcup_{i\in A}B_i.
 \label{eq:core-block-union}
\end{equation}
For simplicity, let $W=W_{\mathcal A_j}$.
Recall the definitions of $\Delta$ and $I(W)$ in \eqref{eq:core-parameters} and \eqref{eq:incident-clause-count}, respectively. 
Then we have the following lemmas.

\begin{lemma} \label{lem:core-oscillation-incidence} 
For $A\subseteq[s]$, we have $R(A)\le\rho_0 I(W_A)$.
\end{lemma}

\begin{lemma}\label{lem-wleksquareovertwodeltaiw}
$|W|\le \frac{k^2}{2\Delta}I(W)$.
\end{lemma}

\begin{lemma}\label{lem-winlq-wplust}
There is a connected variable set in $\G_\Phi$ containing $W$ and having
at most $|W|+T$ variables.
\end{lemma}

The proofs of Lemmas \ref{lem-rqequalaj}--\ref{lem-winlq-wplust} are provided in the appendix.
Now we can prove \Cref{lem-sizeajraj}.

\begin{proof}[Proof of \Cref{lem-sizeajraj}]
Every queried clause $a_t$ meets $W_{\mathcal R_t}\subseteq W$, so
\begin{equation}
 Q=|\mathcal P_Q|\le I(W).
 \label{eq:core-Q-IW}
\end{equation}
By \Cref{lem-winlq-wplust}, there is a connected variable set $L$
containing $W$ with $|L|\le |W|+T$. Applying the connected-incidence
bound \eqref{eq:core-incidence-input} to $L$ and then
\Cref{lem-wleksquareovertwodeltaiw} gives
\[
 I(W)\le I(L)\le C_{\INC}(|W|+T)+4\ln n
 \le\frac{C_{\INC}k^2}{2\Delta}I(W)+C_{\INC}T+4\ln n.
\]
For $k\ge30$, the parameter definitions give
\begin{equation}
 \frac{C_{\INC}k^2}{2\Delta}=\frac{25}{3k^3}\le\frac1{100},
 \qquad
 C_{\INC}\rho=\frac{k^2\lambda}{25\overline D}
 \le\frac1{150k^6}\le\frac1{100}.
 \label{eq:core-numerical-slack}
\end{equation}
Thus the first coefficient is $1-c_k$, and
\begin{equation}
 I(W)\le\frac{C_{\INC}T+4\ln n}{c_k}.
 \label{eq:core-IW-bound}
\end{equation}

The exploration starts from one block. Each checked open label adds at
most $|J_\Gamma|\le K_\Gamma$ new blocks, by
\eqref{eq:core-scope-size}. Consequently,
\[
 |\mathcal A_j|\le1+T.
\]
This counts discovered components directly, without charging the size
of the starting block. Using \Cref{lem:core-oscillation-incidence},
\eqref{eq:core-Q-IW}, and \eqref{eq:core-IW-bound}, we obtain
\begin{align}
 \ln\!\left(2^{|\mathcal A_j|}e^{R(\mathcal A_j)}\right)
       +15\rho_0Q
 &\le (\ln2)(1+T)+16\rho_0I(W)\notag\\
 &\le\ln2+
   \left(\ln2+\frac{16\rho_0C_{\INC}}{c_k}\right)T
   +\frac{64\rho_0}{c_k}\ln n\notag\\
 &\le\ln2+T+\gamma_{\mathrm{mix}}\ln n.
 \label{eq:core-direct-charge}
\end{align}
For the last inequality, $\rho_0\le\rho$ and
\eqref{eq:core-numerical-slack} imply
$\ln2+16\rho_0C_{\INC}/c_k\le\ln2+16/99<1$.
Exponentiating \eqref{eq:core-direct-charge} and applying
\Cref{lem:core-exploration-mgf} yields
\[
 \mathbb E_\zeta\!\left[
   2^{|\mathcal A_j|}e^{R(\mathcal A_j)}\right]
 \le2n^{\gamma_{\mathrm{mix}}}\mathbb E_\zeta e^{T-15\rho_0Q}
 \le2n^{\gamma_{\mathrm{mix}}}.
\]
The bound is uniform in $j$ and in the finite truncation $\Lambda$.
\end{proof}

\subsection{Proof of Theorem~\ref{thm:core-spectral-gap}}

Now we can prove the main theorem of this section.
\begin{proof}[Proof of Theorem~\ref{thm:core-spectral-gap}]
Assume $s\ge1$. 
Recall that a polymer $\gamma$ is a nonempty set of clause labels
$\gamma\subseteq\C^{\good}$ whose induced subgraph $\H_\Psi[\gamma]$
is connected. A cluster label $\Gamma$ is an ordered tuple
$\Gamma=(\gamma_1,\ldots,\gamma_\ell)$ of polymers, with repetitions
allowed and finite length $\ell\ge1$.
Let $\mathfrak L$ be the set of all cluster labels, and define
\[
\mathfrak L_{\ge2}=\{\Gamma\in\mathfrak L:|J_\Gamma|\ge2\}.
\]
The clause-label set
$\C^{\good}$ is finite, so it has only finitely many polymers (each
polymer is a nonempty subset of $\C^{\good}$). Consequently, for each
fixed length $\ell$ there are only finitely many possible tuples
$(\gamma_1,\ldots,\gamma_\ell)$. Taking the union over the positive
integer lengths shows that the set of all cluster labels, and hence its
subfamily $\mathfrak L_{\ge2}$, is countable. Choose an increasing
sequence of finite sets $\Lambda_N\subseteq\mathfrak L_{\ge2}$ such that
$\bigcup_N\Lambda_N=\mathfrak L_{\ge2}$.

We partition the terms of the expansion of $\ln Z(\sigma)$ in
\eqref{eq:core-cluster-root}, indexed by $\Gamma\in\mathfrak L$, according
to whether $J_\Gamma$ is empty, a singleton, or contains at least two
indices. The last class is indexed by $\mathfrak L_{\ge2}$.
Absolute convergence in Lemma~\ref{lem:core-cluster-expansion}
justifies this regrouping.
If $J_\Gamma=\emptyset$, none of the original clauses indexed by
$U_\Gamma$ contains a bad variable. Thus $\chi_a(\sigma)=1$ for every
$a\in U_\Gamma$, so $Y_\Gamma(\sigma)=1$ and
$a_\Gamma Y_\Gamma(\sigma)=a_\Gamma$.
If $J_\Gamma=\{i\}$, we abbreviate $Y_\Gamma\left(\sigma_{J_\Gamma}\right)$
as $Y_\Gamma\left(\sigma_i\right)$ and define the local weight
\[
 \ell_i\left(\sigma_i\right)=
 \exp\left(\sum_{\Gamma\in\mathfrak L:J_\Gamma=\{i\}}
                      a_\Gamma Y_\Gamma\left(\sigma_i\right)\right).
\]
Absolute summability in Lemma~\ref{lem:core-cluster-expansion} ensures
that this exponent is finite, so $\ell_i\left(\sigma_i\right)>0$.
Exponentiating the regrouped expansion gives
\[
 \begin{aligned}
 Z(\sigma)
 &=\exp\left(\sum_{\Gamma\in\mathfrak L:J_\Gamma=\emptyset}a_\Gamma\right)
   \left(\prod_{i=1}^s\ell_i\left(\sigma_i\right)\right)
   \exp\left(\sum_{\Gamma\in\mathfrak L_{\ge2}}
                    a_\Gamma Y_\Gamma\left(\sigma_{J_\Gamma}\right)\right).
 \end{aligned}
\]
The first exponential factor is independent of $\sigma$ and cancels
in the ratio $Z(\sigma)/\sum_{\omega\in\Omega}Z(\omega)$ in
\eqref{eq:core-extension-weight}. Hence
\begin{equation}
 \mu^{\bad}(\sigma)\propto
 \left(\prod_{i=1}^s\ell_i\left(\sigma_i\right)\right)
 \exp\left(\sum_{\Gamma\in\mathfrak L_{\ge2}}
                    a_\Gamma Y_\Gamma\left(\sigma_{J_\Gamma}\right)\right).
 \label{eq:core-interaction-representation}
\end{equation}
Define
\begin{equation}
 \pi_N(\sigma)=\frac1{Z_N}
 \left(\prod_{i=1}^s\ell_i(\sigma_i)\right)
 \exp\left(H_N(\sigma)\right),\qquad
 H_N(\sigma)=\sum_{\Gamma\in\Lambda_N}
      a_\Gamma Y_\Gamma(\sigma_{J_\Gamma}),
 \label{eq:core-finite-interaction-measure}
\end{equation}
where $Z_N$ is the sum of the numerator over $\Omega$. Thus, for each
$N$, $\pi_N$ is a strictly positive probability measure on the same finite
state space $\Omega$.

For each $N$, apply the component-swap comparison
\Cref{lem:core-swap-comparison} to $\pi_N$ using the identification
\[
 \mathcal J=\Lambda_N,\qquad
 e=\Gamma,\qquad V_e=J_\Gamma,\qquad
 a_e=a_\Gamma,\qquad Y_e=Y_\Gamma.
\]
Under this identification, the oscillation in
\eqref{eq:core-swap-oscillation} is exactly the quantity $R(A)$ used in
\Cref{lem-sizeajraj}, restricted to labels in $\Lambda_N$. The edge
indicators have the same independent law
$\mathbb P(\zeta_\Gamma=1)=1-e^{-2|a_\Gamma|}$. Applying \Cref{lem-sizeajraj} to the finite label set $\Lambda_N$ gives
\begin{equation}
 \max_{j\in[s]}\mathbb E_\zeta\!\left[
   2^{|\mathcal A_j|}e^{R\left(\mathcal A_j\right)}\right]\le 2n^{\gamma_{\mathrm{mix}}}.
\end{equation}
Applying \Cref{lem:core-swap-comparison} with this bound yields, for every $f:\Omega\to\mathbb R$,
\begin{equation}
 \operatorname{Var}_{\pi_N}(f)
 \le 2n^{\gamma_{\mathrm{mix}}}\sum_{j=1}^s V_j^{\pi_N}(f).
 \label{eq:core-finite-poincare}
\end{equation}

We pass \eqref{eq:core-finite-poincare} to the limit. Write $H$ for the sum of all interaction terms indexed by $\mathfrak L_{\ge2}$. By \Cref{lem:core-cluster-expansion},
\[
 \sup_{\sigma\in\Omega}|H_N(\sigma)-H(\sigma)|
 \le\sum_{\Gamma\in\mathfrak L_{\ge2}\setminus\Lambda_N}|a_\Gamma|
 \longrightarrow 0.
\]
The local weights are fixed and positive. Since $\Omega$ is finite, the unnormalized weights in \eqref{eq:core-finite-interaction-measure} and their sum converge to those in \eqref{eq:core-interaction-representation}. Thus $\pi_N(\sigma)\to\mu^{\bad}(\sigma)$ for every $\sigma$.
For each $j$ the limiting outside marginal $\mu^{\bad}(X_{-j}=y)$ is positive for every feasible $y$, by \eqref{eq:core-feasibility-extension}. Hence the conditional laws given $X_{-j}=y$ also converge pointwise. Both the global variance and the averaged conditional variances are finite sums, so taking the limit in \eqref{eq:core-finite-poincare} gives
\begin{equation}
 \operatorname{Var}_{\mu^{\bad}}(f)
 =\lim_{N\to\infty}\operatorname{Var}_{\pi_N}(f)
 \le 2n^{\gamma_{\mathrm{mix}}}\lim_{N\to\infty}\sum_{j=1}^s V_j^{\pi_N}(f)
 =2n^{\gamma_{\mathrm{mix}}}\sum_{j=1}^s V_j^{\mu^{\bad}}(f).
 \label{eq:core-limit-poincare}
\end{equation}

By the heat-bath Dirichlet-form identity \eqref{eq:core-dirichlet},
\[
 2n^{\gamma_{\mathrm{mix}}}\sum_{j=1}^sV_j^{\mu^{\bad}}(f)
 =2s n^{\gamma_{\mathrm{mix}}}\,\mathcal E_{\K}(f,f).
\]
Thus \eqref{eq:core-limit-poincare} gives
$\operatorname{Var}_{\mu^{\bad}}(f)\le
2s n^{\gamma_{\mathrm{mix}}}\,\mathcal E_{\K}(f,f)$. If $|\Omega|>1$, the Rayleigh--Ritz
characterization in \Cref{sec:preliminary} then yields
\[
 \gap(\K)=\inf_{\operatorname{Var}_{\mu^{\bad}}(f)>0}
       \frac{\mathcal E_{\K}(f,f)}
            {\operatorname{Var}_{\mu^{\bad}}(f)}
       \ge\frac1{2s n^{\gamma_{\mathrm{mix}}}}.
\]
Each of the $s$ blocks is a nonempty, pairwise disjoint subset of the
$n$ variables, so $s\le n$ and consequently $1/(2s n^{\gamma_{\mathrm{mix}}})\ge1/(2n^{1+\gamma_{\mathrm{mix}}})$. If
$|\Omega|=1$, all functions on $\Omega$ are constant, so both sides of
the variance inequality \eqref{eq:core-limit-poincare} vanish. We use the
convention $\gap(\K)=1$.

It remains to prove the total-variation estimate \eqref{eq:core-tv} when $|\Omega|>1$.
For each $j$, let $P_j$ be the conditional-resampling operator, with
transition kernel
\[
 P_j(\sigma,\omega)=\one_{\{\omega_{-j}=\sigma_{-j}\}}
 \mu^{\bad}(X_j=\omega_j\mid X_{-j}=\sigma_{-j}),
 \qquad \sigma,\omega\in\Omega.
\]
View each $P_j$ as an operator on $L^2(\mu^{\bad})$, equipped with the
inner product $\langle f,g\rangle_{\mu^{\bad}}=\mathbb E_{\mu^{\bad}}[fg]$.
For any $g:\Omega\to\mathbb R$, its action is
\[
 (P_jg)(\sigma)=\sum_{\omega\in\Omega}P_j(\sigma,\omega)g(\omega)
 =\mathbb E_{\mu^{\bad}}[g(X)\mid X_{-j}=\sigma_{-j}].
\]
Thus $P_j$ is the orthogonal projection onto the subspace of functions
depending only on $X_{-j}$. In particular, $P_j=P_j^*=P_j^2$, where
$P_j^*$ is the adjoint with respect to this inner product; hence $P_j$
is positive semidefinite. The Markov operator corresponding to $\K$ is
$P=s^{-1}\sum_{j=1}^sP_j$, so it too is self-adjoint and positive
semidefinite. Since $\mu^{\bad}$ has full support on the product space
$\Omega$, each block can be resampled to any of its values with positive
probability. Thus $\K$ is irreducible, and the eigenvalue $1$ is simple.
The gap bound in \eqref{eq:core-gap} then implies that every eigenvalue
of $P$ on the mean-zero subspace lies in $[0,1-1/(2s n^{\gamma_{\mathrm{mix}}})]$. Consequently,
for every mean-zero $g$,
\begin{equation}\label{eq:core-l2-contraction}
 \|P^tg\|_{L^2(\mu^{\bad})}
 \le\left(1-\frac1{2s n^{\gamma_{\mathrm{mix}}}}\right)^t
       \|g\|_{L^2(\mu^{\bad})}.
\end{equation}

Fix $\sigma_0\in\Omega$ and define
\begin{equation}\label{eq:core-initial-density}
 g_0(\sigma)\triangleq\frac{\one_{\{\sigma_0\}}(\sigma)}{\mu^{\bad}(\sigma_0)}-1.
\end{equation}
This is the initial density relative to $\mu^{\bad}$, minus its constant part.
In particular, \eqref{eq:core-initial-density} gives
$\mathbb E_{\mu^{\bad}}[g_0]=1-1=0$, so the contraction
\eqref{eq:core-l2-contraction} applies to $g_0$.
For $t\ge0$, write
$h_t(\omega)=\K^t(\sigma_0,\omega)/\mu^{\bad}(\omega)$ for the density
at time $t$. Since $\K^0$ is the identity kernel,
\[
\begin{aligned}
 h_0(\omega) - 1
 &=\frac{\K^0(\sigma_0,\omega)}{\mu^{\bad}(\omega)} - 1
  =\frac{\one_{\{\omega=\sigma_0\}}}{\mu^{\bad}(\omega)} - 1
  =\frac{\one_{\{\omega=\sigma_0\}}}{\mu^{\bad}(\sigma_0)} - 1 = g_0(\omega).
\end{aligned}
\]
The last equality holds because the indicator vanishes
unless $\omega=\sigma_0$. Detailed balance gives
\[
 h_{t+1}(\omega)
 =\frac{1}{\mu^{\bad}(\omega)}
   \sum_{\sigma\in\Omega}\mu^{\bad}(\sigma)h_t(\sigma)
        \K(\sigma,\omega)
 =\sum_{\sigma\in\Omega}\K(\omega,\sigma)h_t(\sigma)
 =(Ph_t)(\omega).
\]
Since $P$ is a Markov operator, it preserves the constant function $1$,
so $P1=1$. Subtracting this identity from the recurrence above gives
\[
 h_{t+1}-1=Ph_t-P1=P(h_t-1).
\]
Starting from $h_0-1=g_0$ and iterating this recurrence, we obtain
\begin{equation}\label{eq:core-density-evolution}
 h_t-1=P^t(h_0-1)=P^tg_0.
\end{equation}
By \eqref{eq:core-initial-density}, $g_0(\sigma_0)=1/\mu^{\bad}(\sigma_0)-1$, whereas
$g_0(\sigma)=-1$ for every $\sigma\ne\sigma_0$. Therefore,
\begin{equation}\label{eq:core-initial-density-norm}
\begin{aligned}
 \|g_0\|_2^2
 &=\sum_{\sigma\in\Omega}\mu^{\bad}(\sigma)g_0(\sigma)^2=\mu^{\bad}(\sigma_0)
   \left(\frac1{\mu^{\bad}(\sigma_0)}-1\right)^2
   +\sum_{\sigma\ne\sigma_0}\mu^{\bad}(\sigma)\\
 &=\mu^{\bad}(\sigma_0)
   \left(\frac1{\mu^{\bad}(\sigma_0)}-1\right)^2
   +1-\mu^{\bad}(\sigma_0)=\frac1{\mu^{\bad}(\sigma_0)}-1.
\end{aligned}
\end{equation}
By the total-variation identity \eqref{eq:def-TV-distance},
\[
 \TV{\K^t(\sigma_0,\cdot)}{\mu^{\bad}}
 =\frac12\mathbb E_{\mu^{\bad}}[|h_t-1|].
\]
Using \eqref{eq:core-density-evolution}, this becomes
\[
 \TV{\K^t(\sigma_0,\cdot)}{\mu^{\bad}}
 =\frac12\mathbb E_{\mu^{\bad}}[|P^tg_0|].
\]
Since $\mu^{\bad}$ is a probability measure, Cauchy--Schwarz gives
\[
 \frac12\mathbb E_{\mu^{\bad}}[|P^tg_0|]
 \le\frac12\|P^tg_0\|_2.
\]
Applying the contraction \eqref{eq:core-l2-contraction} (which applies to
$g_0$ because it is mean-zero by \eqref{eq:core-initial-density}), we get
\[
 \TV{\K^t(\sigma_0,\cdot)}{\mu^{\bad}}
 \le\frac12\|P^tg_0\|_2
 \le\frac12\left(1-\frac1{2s n^{\gamma_{\mathrm{mix}}}}\right)^t\|g_0\|_2.
\]
Finally, substituting the norm calculation \eqref{eq:core-initial-density-norm}
yields
\begin{equation}\label{eq:core-tv-intermediate-bound}
 \TV{\K^t(\sigma_0,\cdot)}{\mu^{\bad}}
 \le\frac12\sqrt{\frac1{\mu^{\bad}(\sigma_0)}-1}
       \left(1-\frac1{2s n^{\gamma_{\mathrm{mix}}}}\right)^t.
\end{equation}
By \eqref{eq:core-feasibility-extension}, every $\sigma_0$ has at least
one satisfying full extension. Since $\mu^{\bad}$ is the induced marginal
described by \eqref{eq:core-extension-weight} and there are at most $2^n$
full assignments in total, it follows that
$\mu^{\bad}(\sigma_0)\ge2^{-n}$.
Consequently, the prefactor is bounded as follows:
\begin{equation}\label{eq:core-initial-density-prefactor}
 \sqrt{\frac1{\mu^{\bad}(\sigma_0)}-1}\le2^{n/2}.
\end{equation}
Moreover, we also have
\begin{equation}\label{eq:core-exponential-decay}
 \left(1-\frac1{2s n^{\gamma_{\mathrm{mix}}}}\right)^t
 \le\exp\left(-\frac{t}{2s n^{\gamma_{\mathrm{mix}}}}\right).
\end{equation}
Combining \eqref{eq:core-initial-density-prefactor} with
\eqref{eq:core-exponential-decay} and substituting both bounds into
\eqref{eq:core-tv-intermediate-bound} proves
\eqref{eq:core-tv}. If $|\Omega|=1$, the total-variation distance is
identically zero.
The theorem is proved.
\end{proof}

\section{Rejection-repair Sampler for the Good Variables} \label{secd:rejection-repair-sampler}

In this section, we provide the sampling algorithm for the good formula $\Psi$.
The algorithm adapts the near-linear time exact sampler in \cite{LWYZZ26} to the random $k$-SAT setting.
We replace the two-tree framework of \cite{LWYZZ26} with a similar one on ordinary connected sets of constraints, with a proper truncation on their size to implement the final fast sampler for the ordinary formula $\Phi$ in \Cref{thm:fast-sampler}.

Suppose that $\Phi$ is a nice formula satisfying \Cref{cond:nice-formula}.
For simplicity, we use $\F = \C^{\good}$ to denote the family of good clauses.
For a clause label $a \in \C^{\good}$, we use it to denote $C_a^{\good}$ and $B_a$ to denote the bad event violating its residual good clause.
Let $\H$ be a graph on $\F$ in which a pair of constraints intersecting with each other are adjacent.
For $F \subseteq \F$, let $\Good(F)$ be the event that all bad events in $F$ do not occur and write
\begin{align*}
    Z_F \triangleq \nu^{\good}\left(\Good(F)\right), \quad \mu_F(\cdot) = \nu^{\good}\left(\cdot \;|\; \Good(F)\right).
\end{align*}
To guarantee the existence of measures, assume that $Z_{\F} > 0$.

For a non-empty subset $F \subseteq \F$ and $a \in F$, redefine $\CS_F(a)$ to consist of all connected subsets of $\H[F]$ containing $a$.
For every $S \in \CS_F(a)$, define
\begin{equation} \label{eq:rejection-repair-notations}
\begin{gathered}
    \V_S \defeq \bigcup_{c \in S} \vbl^{\good}(c), \qquad E_S \defeq \bigcap_{c \in S} B_c, \qquad q_S \defeq \nu^{\good}(E_S), \\
    F^- \defeq F \setminus \set{a}, \qquad F_S \defeq F \setminus N_F[S], \qquad \ell(S) \defeq \abs{N_F[S] \setminus \set{a}} \le \overline{D} \cdot \abs{S}
\end{gathered}
\end{equation}
where $N_F[S]$ is the closed neighborhood of $S$ in $\H[F]$.
We order $N_F[S] \setminus \set{a} = \set{b_{S, 1}, \ldots, b_{S, \ell(S)}}$ and put
\begin{align*}
    \forall i = 0, \ldots, \ell(S), \quad F_{S, i} = F_S \cup \set{b_{S, 1}, \ldots, b_{S, i}}
\end{align*}
and thus $F_{S, 0} = F_S$, $F_{S, \ell(S)} = F^-$.
Furthermore, define
\begin{align*}
    w_S \defeq q_S \frac{Z_{F_S}}{Z_{F^-}}.
\end{align*}
For $q_S > 0$, let $u_S$ be the unique local assignment specified by $E_S$.
Define the auxiliary distribution $\eta_S$ as the conditional distribution of $\mu_{F_S}$ on $u_S$ fixed, \IE,
\begin{align*}
    \eta_S = \nu^{\good}\ab(\cdot \mid E_S \cap \Good(F_S)).
\end{align*}

\subsection{Inclusion-exclusion identity}
The following inclusion-exclusion identity for connected subsets, in a similar form as \cite[Lemma 6.1]{LWYZZ26}, plays a crucial role in the design and analysis.
\begin{lemma} \label{lem:connected-subset-inclusion-exclusion}
    For every $F \subseteq \F$ and $a \in F$,
    \begin{align*}
        \mu_{F^-} \one_{B_a} = \sum_{S \in \CS_F(a)} (-1)^{\abs{S} - 1} w_S \eta_S.
    \end{align*}
    Moreover, define
    \begin{align*}
        \OO = \OO_{F, a} \defeq \sum_{S \in \CS_F(a)\;:\; \abs{S} \text{ odd}} w_S, \qquad
        \EE = \EE_{F, a} \defeq \sum_{S \in \CS_F(a)\;:\; \abs{S} \text{ even}} w_S.
    \end{align*}
    Then $\OO - \EE = \mu_{F^-}(B_a)$.
\end{lemma}
\begin{proof}
    Fix an assignment $\sigma$.
    Expand the product over $F^-$, grouping each subset $J$ containing $a$ by the component $S$ of $a$ in $H[J]$. For fixed $S$, the possible sets $J$ are exactly $S\cup T$ for $T\subseteq F_S$. With $\one_{E_\emptyset}=1$ this gives
    \begin{align*}
        \one_{B_a}\one_{\Good(F^-)}
        &=\sum_{J\subseteq F:\,a\in J}(-1)^{|J|-1}\one_{E_J}\\
        &=\sum_{S\in\CS_F(a)}(-1)^{|S|-1}\one_{E_S}
                    \sum_{T\subseteq F_S}(-1)^{|T|}\one_{E_T}\\
        &=\sum_{S\in\CS_F(a)}(-1)^{|S|-1}\one_{E_S}\one_{\Good(F_S)}.
    \end{align*}
    It remains to identify the measure associated with each summand.
    Note that every clause in $F_S$ uses variables only outside $\V_S$.
    Therefore $E_S$ and $\Good(F_S)$ depend on disjoint sets of coordinates and are independent under the product measure $\nu^{\good}$.

    If $q_S = 0$, the corresponding term is measured zero and can be omitted.
    Suppose $q_S > 0$.
    Recall the definition of $u_S$, the unique local assignment on $\V_S$ specified by $E_S$.
    Thus
    \begin{align*}
        E_S = \set{\sigma \midc \sigma|_{\V_S} = u_S}, \qquad q_S = 2^{-\abs{\V_S}}.
    \end{align*}
    Independence tells
    \begin{align*}
        \eta_S = \delta_{u_S} \otimes \nu_{\V^{\good} \setminus \V_S}^{\good}\ab(\cdot \mid \Good(F_S)).
    \end{align*}
    Then for every event $A$,
    \begin{align*}
        \nu^{\good}(A \cap E_S \cap \Good(F_S)) = q_S Z_{F_S} \eta_S(A).
    \end{align*}
    Integrate over $A$ against $\nu^{\good}$, divide by $Z_{F^-}$, apply the inclusion-exclusion identity for $\one_{B_a} \one_{\Good(F^-)}$ and we obtain that
    \begin{align*}
        \mu_{F^-}(A \cap B_a) &= \frac{\nu^{\good}(A \cap B_a \cap \Good(F^-))}{Z_{F^-}} \\
        &= \sum_{S \in \CS_F(a)} (-1)^{\abs{S} - 1} \frac{q_S Z_{F_S}}{Z_{F^-}} \eta_S(A) \\
        &= \sum_{S \in \CS_F(a)} (-1)^{\abs{S} - 1} w_S \eta_S(A).
    \end{align*}
    This proves the first argument.

    For the second argument, taking $A$ to be the entire assignment space and using $\eta_S(A) = 1$ provides that
    \begin{align*}
        \mu_{F^-}(B_a) = \sum_{S \in \CS_F(a)} (-1)^{\abs{S} - 1} w_S = \OO - \EE.
    \end{align*}
    Thus we conclude the lemma.
\end{proof}

\subsection{Parity pairing}
The second preparation is the \emph{parity-pairing} lemma.
The following canonical exploration is frequently employed in our design and analysis.
Given a fixed total order on $F$, the exploration of $S \in \CS_F(a)$ starts with the accepted set $\set{a}$ and unlabeled other vertices.
Repeatedly explore and label the first unlabeled neighbor of the accepted set, accepting it when it is in $S$, otherwise rejecting it.
When there is no candidate, stop the exploration.
By connectivity of $S$, the final accepting set is exactly $S$ and the labeled set is $N_F[S] \setminus \set{a}$.
Moreover, for an assignment $\sigma$, define $Q(\sigma) \triangleq \set{c \in F \mid \sigma \in B_c}$ as the constraints $\sigma$ violates.
\begin{lemma}[Parity pairing] \label{lem:parity-pairing}
    Suppose that $\sigma \in E_S \cap \Good\left(F_S\right)$.
    When $Q(\sigma) \neq \set{a}$, let $d_\sigma(S)$ be the last member of $Q(\sigma) \setminus \set{a}$ explored by $S$, and define $\iota_\sigma(S) \triangleq S \triangle \set{d_\sigma(S)}$.
    Then $\iota_\sigma(S) \in \CS_F(a)$, $\sigma \in E_{\iota_\sigma(S)} \cap \Good\left(F_{\iota_\sigma(S)}\right)$ and $\iota_\sigma\left(\iota_\sigma(S)\right) = S$.
    Trivially $S$ and $\iota_\sigma(S)$ have opposite parity.
    When $Q(\sigma)=\set{a}$, the only admissible ordinary label is $\set{a}$. Extend the pairing to the exceptional label $\star$ by setting $\iota_\sigma(\set{a})=\star$ and $\iota_\sigma(\star)=\set{a}$.
    Moreover, these operations give a measure-preserving bijection between the two sides of
    \begin{align} \label{eq:measure-parity-pairing}
        \mu_{F^-} \one_{B_a} + \sum_{S \in \CS_F(a) \;:\; \abs{S} \text{ even}} w_S \eta_S \longleftrightarrow \sum_{S \in \CS_F(a) \;:\; \abs{S} \text{ odd}} w_S \eta_S.
    \end{align}
    In particular, for every event $A$,
    \begin{align} \label{eq:event-parity-pairing}
        \mu_{F^-}(A \cap B_a) + \sum_{S \in \CS_F(a) \;:\; \abs{S} \text{ even}} w_S \eta_S(A) = \sum_{S \in \CS_F(a) \;:\; \abs{S} \text{ odd}} w_S \eta_S(A).
    \end{align}
\end{lemma}
\begin{proof}
    The statement $\sigma \in E_S \cap \Good\left(F_S\right)$ is equivalent to
    \begin{align*}
        S \subseteq Q(\sigma) \subseteq N_F[S].
    \end{align*}
    Therefore all constraints in $Q(\sigma)$ are labeled in the exploration by $S$.

    When $Q(\sigma) = \set{a}$, there must hold $S \subseteq \set{a} \subseteq N_F[S]$, meaning that $S$ must be $\set{a}$.
    Otherwise, write $S' = S \triangle \set{d_\sigma(S)}$.
    We verify the following properties of $S'$.
    \begin{enumerate}
        \item $S' \in \CS_F(a)$.
        Let $A$ be the accepted set immediately before $d_\sigma(S)$ is explored.
        Note that there is no accepted vertex after $d_\sigma(S)$ is explored (since $S \subseteq Q(\sigma)$).
        Therefore, it must be $S' = A$ or $S' = A \cup \set{d_\sigma(S)}$ and $d_\sigma(S)$ must be adjacent to $A$.
        It is clear that $a \in A$ by the definition of $d_\sigma(S)$.
        Hence we conclude that $a \in S'$ and $S'$ is connected, leading to the property $S' \in \CS_F(a)$.

        \item $\sigma \in E_{S'} \cap \Good\left(F_{S'}\right)$.
        It is equivalent to show that
        \begin{align*}
            S' \subseteq Q(\sigma) \subseteq N_F[S'].
        \end{align*}
        Since $A \subseteq S \subseteq Q(\sigma)$ and $d_\sigma(S) \in Q(\sigma)$, we directly obtain that $S' \subseteq Q(\sigma)$.
        To prove $Q(\sigma) \subseteq N_F[S']$, recall the exploration of $S$ and the definition that $d_\sigma(S)$ is the last member of $Q(\sigma) \setminus \set{a}$ in the exploration.
        Then all clauses in $Q(\sigma) \setminus \set{a}$ were explored and thus labeled by the time $d_\sigma(S)$ was explored.
        Hence $Q(\sigma) \subseteq N_F[A] \subseteq N_F[S']$ by $A \subseteq S'$.

        \item $\iota_\sigma\left(\iota_\sigma(S)\right) = S$.
        Compare the explorations of $S$ and $S'$.
        Their membership decisions differ only at $d_\sigma(S)$.
        Then until $d_\sigma(S)$ is explored, they share the same accepted set and label set.
        Moreover, under the fixed total order on $F$, their exploration sequences remain identical through the selection of $d_\sigma(S)$.

        At $d_\sigma(S)$, the accept/reject decision is reversed.
        Since either $S' = A$ or $S' = A \cup \set{d_\sigma(S)}$ holds, any subsequent candidate is rejected.
        Therefore, we obtain that $d_{\sigma}\left(S'\right) = d_\sigma(S)$.
        Hence we conclude the property that $\iota_\sigma\left(\iota_\sigma(S)\right) = S$.
    \end{enumerate}

    Now we verify \eqref{eq:measure-parity-pairing}.
    Observe that
    \begin{align*}
        w_S \eta_S(\sigma) = \frac{\nu^{\good}(\sigma)}{Z_{F^-}} \id{S \subseteq Q(\sigma) \subseteq N_F[S]}.
    \end{align*}
    Consider all admissible labels of $\sigma$:
    \begin{align*}
        \Label_{F, a}(\sigma) \defeq \set{S \in \CS_F(a) \mid S \subseteq Q(\sigma) \subseteq N_F[S]}
    \end{align*}
    adding one exceptional state $\star$ when $Q(\sigma) = \set{a}$.
    Then all pairs of $(S, \sigma)$ with $S \in \Label_{F, a}(\sigma)$ share the same mass $\nu^{\good}(\sigma)/Z_{F^-}$.
    When $Q(\sigma) \neq \set{a}$, the mapping $(S, \sigma) \mapsto (\iota_\sigma(S), \sigma)$ is clearly a bijection from odd admissible labels to even admissible labels and preserves the point mass at $\sigma$.

    When $Q(\sigma) = \set{a}$, the only admissible label is $S = \set{a}$ and we pair $(\set{a}, \sigma)$ with $(\star, \sigma)$.
    Direct calculation provides that for the label $\set{a}$,
    \begin{align*}
        w_{\set{a}} \eta_{\set{a}}(\sigma) = \frac{\nu^{\good}(\sigma)}{Z_{F^-}}.
    \end{align*}
    For the exceptional state $\star$ with $\sigma$, since $\sigma$ satisfies all constraints in $F^-$ but violates $a$,
    \begin{align*}
        \ab(\mu_{F^-} \one_{B_a})(\sigma) = \mu_{F^-}(\sigma) = \frac{\nu^{\good}(\sigma)}{Z_{F^-}} = w_{\set{a}} \eta_{\set{a}}(\sigma).
    \end{align*}
    Then this equality shows that the mapping preserves the point mass at $\sigma$.
    
    Consequently, the above mapping indicates a measure-preserving bijection between the two sides of \eqref{eq:measure-parity-pairing} and thus the lemma is concluded.
\end{proof}

\subsection{Ideal and truncated source}
Set
\begin{align*}
    h = 25\overline{D}, \qquad g(S) = \begin{cases}
        \ceil{\ell(S)/h}, & \ell(S) > 0 \\
        0 & \ell(S) = 0
    \end{cases}, \qquad A_S = 2^{g(S)}, \qquad a_0 = \frac{\ln 2}{25}.
\end{align*}
When comparing different requests, we write $A_S(H,c)$ for the weight $A_S$ computed with
$\ell(S)=\abs{N_H[S]\setminus\{c\}}$ in the request $(H,c)$. The violation probability $q_S$ depends only on the clauses in $S$.

\begin{proposition} \label{prop:local-weight-bounds}
    For every $F \subseteq \F$ and $a \in F$,
    \begin{align*}
        \sum_{S \in \CS_F(a)} A_S q_S \le 2\rho, \qquad \sum_{S \in \CS_F(a)} A_S q_S \ell(S) \le \frac{2}{C_*}, \qquad \sum_{S \in \CS_F(a)} A_S q_S \abs{S}^j \le 2j! \cdot \rho = O_j(\rho).
    \end{align*}
\end{proposition}
\begin{proof}
    Observe that for every $S \in \CS_F(a)$,
    \begin{align*}
        g(S) \le 1 + \frac{\ell(S)}{25 \overline{D}} \le 1 + \frac{\abs{S}}{25}.
    \end{align*}
    Therefore, $A_S \le 2\e^{a_0\abs{S}}$.
    Set $\delta \defeq 1 + \tau - a_0$ and by $\tau \ge 1$, $\delta \ge 2 - \ln{2}/25 > \ln{4}$.
    By \Cref{prop:root-event}, for every non-negative function $f$, it holds that
    \begin{align*}
        \sum_{S \in \CS_F(a)} A_S q_S f(\abs{S}) &\le 2\sum_{S \in \CS_F(a)} q_S \exp\ab((1 + \tau)\abs{S}) f(\abs{S}) \e^{-\delta \abs{S}} \\
        &\le 2\rho \cdot \sup_{t \ge 1}\; f(t) \e^{-\delta t}.
    \end{align*}

    Taking $f(t) = 1$, we conclude $\sum_{S \in \CS_F(a)} A_S q_S \le 2\rho$.
    Taking $f(t) = t$ together with $\ell(S) \le \overline{D} \abs{S}$, we obtain that
    \begin{align*}
        \sum_{S \in \CS_F(a)} A_S q_S \ell(S) &\le \sum_{S \in \CS_F(a)} A_S q_S \cdot \overline{D} \abs{S} \\
        &\le 2 \overline D \rho \sup_{t \ge 1} t \e^{-\delta t}.
    \end{align*}
    Using $\delta > \ln{4}$, for $t \ge 1$, $t \e^{-\delta t} \le t 4^{-t} \le 1/4$.
    Consequently,
    \begin{align*}
        \sum_{S \in \CS_F(a)} A_S q_S \ell(S) \le 2\overline D\rho\cdot\tfrac14 = \frac{2}{C_*}.
    \end{align*}
    Finally, fix an integer $j \ge 0$.
    The exponential series provides that $\e^{\delta t} \ge (\delta t)^j/j!$.
    Therefore taking $f(t) = t^j$, we have
    \begin{align*}
        \sum_{S \in \CS_F(a)} A_S q_S \abs{S}^j &\le 2\rho \cdot \sup_{t \ge 1} \; t^j \e^{-\delta t} \le 2\rho \cdot \sup_{t \ge 1}\; t^j \frac{j!}{(\delta t)^j} \le 2j! \cdot \rho = O_j(\rho).
    \end{align*}
    Thus we conclude the proposition.
\end{proof}

The ideal source $\Source_{\infty}(F, a)$ outputs $S$ with mass $A_S q_S$ and outputs the empty label with remaining mass.
\Cref{prop:local-weight-bounds} guarantees the existence of $\Source_{\infty}(F, a)$.
However, it might be hard to implement the ideal source.
Instead, we use the following executable source.
\begin{definition}[Executable source] \label{def:executable-source}
    For $1 \le K \le r$, the executable source $\ExecSource_K(F, a)$ processes as follows.
    \begin{enumerate}
        \item Pass a $\Ber(4p_*)$ gate.
        If it fails, return the empty label.

        \item Run the canonical exploration from the root $a$, accepting each newly queried non-root label independently with probability $\vartheta \defeq 1/(4\overline{D})$.
        Return the empty label as soon as the number of accepted labels exceeds $K$.

        \item Assume that the set of all accepted labels is $S$.
        If $S$ has inconsistent violation requirements, return the empty label.
        Otherwise return $S$ with probability
        \begin{align} \label{eq:acceptance-probability}
            \frac{A_S q_S}{4 p_* \vartheta^{\abs{S} - 1} (1 - \vartheta)^{\abs{N_F[S] \setminus S}}}
        \end{align}
        and return the empty label on rejection.
    \end{enumerate}
\end{definition}

For a subroutine $\AA$, we use $\Cost(\AA)$ to denote its running cost, \IE, the number of its used bit operations.
We provide the accuracy and cost guarantee of $\ExecSource_{\cdot}(\cdot, \cdot)$ in the following proposition.
\begin{proposition} \label{prop:executable-source-accuracy-cost}
    The acceptance probability in \eqref{eq:acceptance-probability} is in $[0, 1]$ and
    \begin{align*}
        \forall S \in \CS_F(a), \quad \Pr{\ExecSource_K(F, a) = S} = A_S q_S \id{\abs{S} \le K}.
    \end{align*}
    Therefore the total variation between the executable source and the ideal source is
    \begin{align*}
        t_K(F,a) \defeq \TV{\ExecSource_K(F, a)}{\Source_{\infty}(F, a)} \le d_K
    \end{align*}
    where $d_K = 2\rho R^{-K}$ with $R = 2/(C_* p_* \overline{D}) \ge 4$.
    Moreover, the expected cost of each call is
    \begin{align*}
        \E{\Cost\ab(\ExecSource_K(F, a))} = \widetilde{O}\ab(\poly{k, \overline{D}}).
    \end{align*}
\end{proposition}
\begin{proof}
    We first verify the acceptance probability in \eqref{eq:acceptance-probability} is in $[0, 1]$.
    For every set $S$ that reaches the acceptance test, we have $\abs S \le K \le r$. Hence, by \Cref{prop:good-property}, $q_S \le 2^{-(k - 12)\abs{S}} = p_*^{\abs{S}}$.
    Together with $\abs{N_F[S] \setminus S} \le \overline{D} \abs{S}$ and $(1 - 1/(4\overline{D}))^{\overline{D} \abs{S}} \ge \e^{-\abs{S}/2}$, we obtain that the acceptance probability is at most
    \begin{align*}
        \frac{1}{2} \e^{(a_0 + 1/2) \abs{S}} \ab(4p_*\overline{D})^{\abs{S} - 1} < 0.85\ab(7p_*\overline{D})^{\abs{S} - 1} \le 1.
    \end{align*}
    By \Cref{def:executable-source}, $\ExecSource_K(F, a)$ only outputs the empty label or labels of size at most $K$.
    For a fixed non-empty $S \in \CS_F(a)$ with $\abs{S} \le K$, the probability that the canonical exploration generates it is $\vartheta^{\abs{S} - 1}(1 - \vartheta)^{\abs{N_F[S] \setminus S}}$.
    Therefore, by \Cref{def:executable-source},
    \begin{align*}
        \Pr{\ExecSource_K(F, a) = S} = 4p_* \cdot \vartheta^{\abs{S} - 1} (1 - \vartheta)^{\abs{N_F[S] \setminus S}} \frac{A_S q_S}{4 p_* \vartheta^{\abs{S} - 1} (1 - \vartheta)^{\abs{N_F[S] \setminus S}}} = A_S q_S.
    \end{align*}
    Hence by the proof of \Cref{prop:local-weight-bounds}, the total variation distance between $\ExecSource_K(F, a)$ and $\Source_\infty(F, a)$ is
    \begin{align*}
        t_K(F,a) = \sum_{S \in \CS_F(a) \;:\; \abs{S} > K} A_S(F,a) q_S \le 2 \rho \e^{-(\tau + 1 - a_0)K} \le d_K.
    \end{align*}

    Now we bound the expected cost of $\ExecSource_K(F, a)$.
    To generate a Bernoulli gate $\Ber(4p_*)$, it takes $O(k)$ bit operations.
    Conditional on passing the Bernoulli gate, let $T_K$ be the number of accepted labels when the exploration stops including the root.
    When a cutoff aborts the exploration, it immediately holds that $T_K = K + 1$.

    By \Cref{def:executable-source}, every accepted label exposes at most $\overline{D}$ possible new labels, each of which is accepted with probability $\vartheta = 1/(4\overline{D})$.
    We can couple the exploration with a Galton--Watson tree with offspring distribution $\xi \sim \Bin(\overline{D}, 1/(4\overline{D}))$.
    Let $T$ denote its total population including its root.
    Then $T_K \le T$ under this coupling.
    We mark here that the cutoff can only decrease the work.
    
    For the Galton--Watson tree, it holds that
    \begin{align*}
        \E{\xi} = \frac{1}{4}, \qquad \E{\xi(\xi - 1)} = \frac{\overline{D}\ab(\overline{D} - 1)}{16\overline{D}^2} \le \frac{1}{16}.
    \end{align*}
    Also the branching identity holds:
    \begin{align*}
        T = 1 + \sum_{i = 1}^{\xi} T^{(i)}
    \end{align*}
    where $T^{(1)}, \ldots, T^{(\xi)}$ are independent copies of $T$.
    We calculate $\E{T}$ and $\E{T^2}$ by the branching identity.
    Conditional on $\xi$,
    \begin{align*}
        \E{T \mid \xi} = 1 + \xi \E{T}
    \end{align*}
    and thus $\E{T} = 1 + \E{\xi} \E{T} = 1 + \E{T}/4$.
    This provides $\E{T} = 4/3$.
    For $\E{T^2}$, conditional on $\xi$,
    \begin{align*}
        \E{\ab(\sum_{i = 1}^\xi T^{(i)})^2 \;\middle|\; \xi} = \xi \E{T^2} + \xi(\xi - 1) \E{T}^2.
    \end{align*}
    Taking expectation on both sides and with the branching identity, we have
    \begin{align*}
        \E{T^2} = 1 + 2\E{\xi} \E{T} + \E{\xi} \E{T^2} + \E{\xi(\xi - 1)} \E{T}^2 \le 1 + \frac{2}{3} + \frac{1}{4} \E{T^2} + \frac{1}{9}.
    \end{align*}
    which gives $\E{T^2} \le 64/27$.
    By monotone convergence, $\E{T_K} \le 4/3$ and $\E{T_K^2} \le 64/27$.

    Now we turn to the exploration cost.
    Assume that we explore $t$ accepted labels.
    There are at most $\overline{D} t$ examined neighbors and thus the graph operations cost $\widetilde{O}(\overline{D} t)$ using local dictionaries for processed labels and an ordered frontier.
    Each acceptance coin has rational parameter $1/(4\overline{D})$, which can be generated in expected $O\ab(\log_2{(\overline{D} + 2)})$ bit operations.
    Charging this cost whenever a query is reached gives expected cost
    \begin{align*}
        \widetilde{O}\ab(\overline{D} \cdot \E{T_K}) = \widetilde{O}\ab(\overline{D}).
    \end{align*}

    With a completed set $S$ of size $t$, to examine its consistency, it contains at most $kt$ literal occurrences and each occurrence specifies the value its variable must take to violate the clause.
    Then it takes $\widetilde{O}(k t)$ bit operations to detect whether there is a conflict.
    When its consistency is verified, it holds that
    \begin{align*}
        q_S = 2^{-\abs{\V_S}}, \qquad \abs{\V_S} \le kt.
    \end{align*}
    It remains to bound the cost of the exact acceptance coin.
    For simplicity, put
    \begin{align*}
        b \defeq \abs{N_F[S] \setminus S}, \qquad \ell \defeq \ell(S) = t - 1 + b, \qquad g \defeq g(S).
    \end{align*}
    Its success probability is
    \begin{align*}
        r_S \defeq \frac{2^{g - \abs{\V_S}}}{4p_* \vartheta^{t - 1}(1 - \vartheta)^{b}} = 2^{g + k - 14 - \abs{\V_S}} \frac{\ab(4\overline{D})^{\ell}}{\ab(4\overline{D} - 1)^{b}}.
    \end{align*}
    It is an exact computable rational number $N/M$, with following bounds for parameters:
    \begin{align*}
        \ell, b \le \overline{D} t, \qquad g \le 1 + \frac{t}{25}, \qquad \abs{\V_S} \le kt.
    \end{align*}
    Its numerator and denominator have bit length at most $O\ab(\ab(k + \overline{D} \log_2(\overline{D} + 2)) t)$.
    Thus arithmetic calculations cost $O\ab(\ab(k + \overline{D} \log_2(\overline{D} + 2))^2 t^2)$ bit operations.
    To generate the $\Ber(r_S)$ coin exactly, sample a uniform integer $U \in \set{0, \ldots, M - 1}$ and return one when $U < N$.
    To sample from $\set{0, \ldots, M - 1}$, we use the rejection sampling for $O\ab(\ab(k + \overline{D} \log_2(\overline{D} + 2)) t)$ bits with expected at most two trials.
    Thus the expected arithmetic and coin cost is at most
    \begin{align*}
        \widetilde{O}\ab(\ab(k + \overline{D} \log_2(\overline{D} + 2))^2 t^2).
    \end{align*}
    When the cutoff occurs, it skips this computation and does not increase the cost.
    
    Write $B = k + \overline{D} \log_2(\overline{D} + 2)$.
    It holds that
    \begin{align*}
        \E{\Cost\ab(\ExecSource_K(F, a))} &\le O(k) + 4p_* \cdot \widetilde{O}\ab(\ab(k + \overline{D}) \E{T_K} + B^2 \E{T_K^2}) \\
        &= \widetilde{O}\ab(k + 4p_* B^2) = \widetilde{O}\ab(\poly{k, \overline{D}})
    \end{align*}
    which concludes the proposition.
\end{proof}

\subsection{Bernoulli factory and connected-subset proposals}
The weighted connected-subset proposals stated here serve as the replacement of the weighted two-tree proposals in \cite{LWYZZ26}.
We first introduce the following \emph{Bernoulli factory}, which is just a reimplementation of the \emph{reciprocal acceptance test} in \cite[Lemma 6.5]{LWYZZ26}.
\begin{definition}[Reciprocal test \cite{LWYZZ26}] \label{def:reciprocal-test}
    For a random oracle $\O$ which returns independent bits, we first implement an auxiliary subroutine $\Coin(\O)$ as follows.
    Repeatedly call $\O$.
    If the output is $1$, return $1$; otherwise flip a fair coin; return $0$ on heads and continue on tails.

    The main procedure $\Reciprocal(\O)$ is:
    \begin{enumerate}
        \item Flip a fair coin and return one on heads.
        
        \item Set $j = 1 + \Coin(\O)$ and $i = 2j - 1$.
        If $j = 2$, flip another fair coin and return one on heads.
        
        \item While $i > 0$, call $\Coin(\O)$.
        If it returns zero, set $i \gets i - 1$.
        Otherwise call $\Coin(\O)$ again.
        If the second call returns one, set $i \gets i + 2, j \gets j + 1$ and return one with probability $1/j$.

        \item Return zero when $i$ reaches $0$.
    \end{enumerate}
\end{definition}
\begin{lemma} \label{lem:reciprocal-test}
    Suppose that an oracle $\O$ returns independent bits with $\Pr{\O = 0} = b \ge 9/10$.
    The procedure $\Reciprocal(\O)$ returns a bit with success probability $1/(2b)$, terminates almost surely and uses at most $4$ expected calls to $\O$ and $O(1)$ additional operations.
\end{lemma}
The proof of \Cref{lem:reciprocal-test} is analogous to that of \cite[Lemma 6.5]{LWYZZ26}.
For completeness, we include it in \Cref{sec:reciprocal-test}.

Given a general source $\Source[\omega]$ with mass $\omega_S A_S$ at $S$ where $0 \le \omega_S \le q_S$, we construct the following proposals.
For a non-empty label $S$ from $\Source[\omega]$, recall that we order
\begin{align*}
    N_F[S] \setminus \set{a} = \set{b_1, \ldots, b_\ell}
\end{align*}
and define
\begin{align*}
    \forall j = 0, \ldots, \ell, \quad F_{S, j} = F_S \cup \set{b_1, \ldots, b_j}
\end{align*}
and thus $F_{S, 0} = F_S, F_{S, \ell} = F^-$.
\begin{definition}[Connected-subset proposal] \label{def:connected-subset-proposal}
    The subroutine $\Propose[\omega](F, a)$ runs as follows.
    \begin{enumerate}
        \item Draw $S$ from the source and return the empty label if it is empty.
        
        \item For every $j = 1, \ldots, \ell = \ell(S)$, child $j$ is a recursive proposal request
        \begin{align*}
            \Propose[\omega](F_{S, j}, b_j).
        \end{align*}
        Here $b_j \in F_{S,j}$ and $F_{S,j} \setminus \set{b_j} = F_{S,j-1}$, as required by the proposal interface.
        In index order, put every $h$ consecutive children in one group, with a possibly shorter last group. Since $h=25\overline D$ is an integer, this gives exactly $g(S)$ groups, including zero groups when $\ell(S)=0$.

        \item Design a two-proposal check as follows.
        Draw one child proposal.
        If it is empty, pass.
        Otherwise, draw an independent second proposal and pass precisely when exactly one of the two is non-empty and even.

        \item For each group $J$, define an oracle $\O_J$ that returns zero precisely when all its two-proposal checks pass.
        Run $\Reciprocal(\O_J)$.
        If the factory returns zero, return the empty label.

        \item For every child $j$, draw one fresh independent proposal and require it to be empty or odd.
        If any proposal is non-empty and even, return the empty label.
        If all checks pass, return the original set $S$.
    \end{enumerate}
\end{definition}

For simplicity, the version using $\ExecSource_K(F, a)$ is denoted by $\Propose_K(F, a)$.

\begin{lemma} \label{lem:connected-subset-proposal}
    Let $\widehat{\OO}_{H, c}$ and $\widehat{\EE}_{H, c}$ be the odd and even output masses of a request $\Propose[\omega](H, c)$ respectively and write $\widehat{r}_{H, c} \defeq \widehat{\OO}_{H, c} - \widehat{\EE}_{H, c}$.
    Then
    \begin{align*}
        \Pr{\Propose[\omega](F, a) = S} = \omega_S \prod_{j = 1}^{\ell(S)} \ab(1 - \widehat{r}_{F_{S, j}, b_j})^{-1}, \qquad \widehat{\OO}_{F, a} + \widehat{\EE}_{F, a} \le \rho.
    \end{align*}
    The expected number of source calls, including all descendants, is at most $2$. Excluding the work inside the sources, the expected cost is $\widetilde O(\poly{k,\overline D})$. For $\Propose_K$ the same bound includes source implementation.
    Specifically, when $\omega = q$ corresponds to the ideal source, the output mass at $S$ is exactly $w_S$.
\end{lemma}
\begin{proof}
    We employ induction on $\abs{F}$.
    When $\abs{F} = 1$, there are no child requests, $A_S = 1$ and the procedure simply returns its source output.
    Then by $\omega_S \le q_S$ and \Cref{prop:root-event}, the assertion holds.

    For the induction step, assume that every smaller request terminates and has non-empty output mass at most $\rho$.
    Consider the two-proposal check for child $j$.
    It passes in three disjoint cases:
    \begin{itemize}
        \item the first proposal is empty;
        \item the first proposal is non-empty and even, and the second is empty or odd;
        \item the first proposal is odd, and the second is non-empty and even.
    \end{itemize}
    Then its passing probability $p_j$ is
    \begin{align*}
        p_j &= \ab(1 - \widehat{\OO}_{F_{S, j}, b_j} - \widehat{\EE}_{F_{S, j}, b_j}) + \widehat{\EE}_{F_{S, j}, b_j}\ab(1 - \widehat{\EE}_{F_{S, j}, b_j}) + \widehat{\OO}_{F_{S, j}, b_j} \widehat{\EE}_{F_{S, j}, b_j} \\
        &= \ab(1 - \widehat{r}_{F_{S, j}, b_j})\ab(1 - \widehat{\EE}_{F_{S, j}, b_j}).
    \end{align*}
    For a group $J$, we use fresh randomness and make these checks independent.
    The group oracle returns zero with probability
    \begin{align*}
        p_J &= \prod_{j \in J} p_j = \prod_{j \in J} \ab(1 - \widehat{r}_{F_{S, j}, b_j}) \ab(1 - \widehat{\EE}_{F_{S, j}, b_j}).
    \end{align*}
    By induction and $\abs{J} \le h = 25\overline{D}$,
    \begin{align*}
        p_J \ge (1 - \rho)^{2\abs{J}} \ge (1 - \rho)^{50\overline{D}}.
    \end{align*}
    Using $\ln(1 - \rho) \ge -2\rho$ on $[0, 1/2]$, we obtain
    \begin{align*}
        p_J \ge \e^{-100\rho \overline{D}} = \e^{-0.04} \ge \frac{9}{10}.
    \end{align*}
    By \Cref{lem:reciprocal-test}, the reciprocal factory terminates almost surely and returns one with probability $1/(2p_J)$.
    Together with the termination of child requests, this proves the termination of $\Propose[\omega](F, a)$.

    Conditional on $\Source[\omega](F, a) = S$, the probability that all factories and checks pass is
    \begin{align*}
        \prod_J \frac{1}{2p_J} \prod_{j = 1}^{\ell(S)} \ab(1 - \widehat{\EE}_{F_{S, j}, b_j}) &= 2^{-g(S)} \prod_{j = 1}^{\ell(S)} \ab(1 - \widehat{r}_{F_{S, j}, b_j})^{-1}.
    \end{align*}
    Therefore, by law of expectation,
    \begin{align*}
        \Pr{\Propose[\omega](F, a) = S} &= A_S \omega_S \cdot 2^{-g(S)} \prod_{j = 1}^{\ell(S)} \ab(1 - \widehat{r}_{F_{S, j}, b_j})^{-1} \\
        &= \omega_S \prod_{j = 1}^{\ell(S)} \ab(1 - \widehat{r}_{F_{S, j}, b_j})^{-1}.
    \end{align*}

    To bound the mass, note that $\widehat{r}_{F_{S, j}, b_j} \le \rho$ for every $j = 1, \ldots, \ell(S)$.
    Hence we have
    \begin{align*}
        \Pr{\Propose[\omega](F, a) = S} \le \omega_S (1 - \rho)^{-\ell(S)} \le q_S \cdot \e^{2\rho\ell(S)}.
    \end{align*}
    For the mass of outputs, it holds that
    \begin{align*}
        \widehat{\OO}_{F, a} + \widehat{\EE}_{F, a} &= \sum_{S \in \CS_F(a)} \Pr{\Propose[\omega](F, a) = S} \\
        &\le \sum_{S \in \CS_F(a)} q_S \e^{2\rho \ell(S)} \\
        &\le \sum_{S \in \CS_F(a)} q_S \e^{2 \overline{D} \rho \abs{S}} \\
        &\le \sum_{S \in \CS_F(a)} q_S \e^{(1 + \tau) \abs{S}} \\
        &\le \rho
    \end{align*}
    by \Cref{prop:root-event} and $2\overline{D} \rho = 8/C_* < 1 + \tau$.

    To bound the number of source calls, let $N$ be the number of direct child proposal calls made by one invocation.
    Conditional on the source output $S$, each group factory uses at most $4$ oracle invocations in expectation.
    Each oracle invocation uses at most two proposals per child index, and the final checks use at most one additional proposal per index.
    Therefore
    \begin{align*}
        \E{N \mid S} \le 9\ell(S).
    \end{align*}
    We mark here that an empty label generates no children.
    Therefore,
    \begin{align*}
        \E{N} &\le 9\sum_{S \in \CS_F(a)} A_S \omega_S \ell(S) \\
        &\le 9 \sum_{S \in \CS_F(a)} A_S q_S \ell(S) \\
        &\le \frac{18}{C_*}.
    \end{align*}
    Set $M_s$ as the supremum of the expected total number of source calls for requests with at most $s$ clauses.
    Note that each invocation makes one call itself, and requests its direct children.
    Each child is specified by at most $s - 1$ clauses.
    Therefore,
    \begin{align*}
        \forall s \ge 1, \quad M_s \le 1 + \frac{18}{C_*} M_{s - 1}
    \end{align*}
    with boundary condition $M_0 = 0$.
    Solving the recursion, we conclude that
    \begin{align*}
        M_s \le \sum_{j = 0}^{s - 1} \ab(\frac{18}{C_*})^{j} \le \frac{1}{1 - 18/C_*} \le 2.
    \end{align*}
    Excluding child calls and source implementations, the expected work conditional on a nonempty source label $S$ is polynomial in $k,\overline D,|S|$: its neighborhood has at most $\overline D|S|$ indices and each group uses a bounded expected number of reciprocal-test iterations by \Cref{lem:reciprocal-test}. The empty label requires constant work. Averaging over $S$ using $\omega_S\le q_S$ and the moment bounds in \Cref{prop:local-weight-bounds} bounds the expected local work by $\widetilde O(\poly{k,\overline D})$ uniformly in the request. Charging this bound whenever a request is reached and using $M_s\le2$ gives the claimed cost outside the sources. For $\Propose_K$, \Cref{prop:executable-source-accuracy-cost} gives the same uniform bound for each executable source call, so it also bounds the full expected cost up to a constant factor.

    For the last argument, when $\omega = q$, we prove the ideal output law by induction on $\abs{F}$.
    The boundary condition $\abs{F} = 1$ is direct.
    Assume the inductive hypothesis.
    Then the child $j$ outputs with weight according to $w$ defined by $(F_{S, j}, b_j)$. 
    By \Cref{lem:connected-subset-inclusion-exclusion}, for every $j = 1, \ldots, \ell(S)$,
    \begin{align*}
        \widehat{r}_{F_{S, j}, b_j} = \mu_{F_{S, j - 1}}(B_{b_j}) = 1 - \frac{Z_{F_{S, j}}}{Z_{F_{S, j - 1}}}.
    \end{align*}
    Therefore,
    \begin{align*}
        \Pr{\Propose[\omega](F, a) = S} = q_S \prod_{j = 1}^{\ell(S)} \frac{Z_{F_{S, j - 1}}}{Z_{F_{S, j}}} = q_S \frac{Z_{F_S}}{Z_{F^-}} = w_S.
    \end{align*}
    Thus we conclude the lemma.
\end{proof}

For a proposal $\Propose[\omega](F, a)$, write
\begin{align*}
    \widehat O &\defeq \Pr{\Propose[\omega](F, a) \text{ is of odd size}}, \\
    \widehat E &\defeq \Pr{\Propose[\omega](F, a) \text{ is not empty and of even size}}
\end{align*}
and set $\widehat \Lambda \defeq \widehat O + \widehat E$ and $\widehat r \defeq \widehat O - \widehat E$.
We prepare the following parity bound for $\Propose_K(F, a)$.
\begin{lemma} \label{lem:parity-bound}
    For $\Propose_K(F, a)$ with $1 \le K \le r$, if $q_a > 0$, then $\widehat \Lambda > 0$ and $\widehat E/\widehat \Lambda < 1/64$.
    If $q_a = 0$, the proposal is empty almost surely.
    For $\widehat r$, it holds that $0 \le \widehat r \le \rho$.
\end{lemma}
\begin{proof}
    Suppose that $q_a = 0$.
    For every $S \in \CS_F(a)$, it also holds that $q_S = 0$ since $E_S \subseteq B_a$.
    Therefore the truncated source returns the empty label almost surely and so does $\Propose_K(F, a)$.
    Consequently $\widehat O = \widehat E = \widehat \Lambda = \widehat r = 0$.

    Now we assume that $q_a > 0$.
    By \Cref{lem:connected-subset-proposal}, every child request has non-empty output mass at most $\rho$.
    Denote by $\widehat r_{S, j}$ the signed mass difference for child $j$ when the initial source returns $S$, then $\abs{\widehat r_{S, j}} \le \rho$.
    By \Cref{lem:connected-subset-proposal},
    \begin{align*}
        \widehat{\omega}(S) \defeq \Pr{\Propose_K(F, a) = S} = q_S \id{\abs{S} \le K} \prod_{j = 1}^{\ell(S)} \ab(1 - \widehat r_{S, j})^{-1}.
    \end{align*}
    In particular,
    \begin{align*}
        (1 + \rho)^{-\ell(S)} \le \prod_{j = 1}^{\ell(S)} \ab(1 - \widehat r_{S, j})^{-1} \le (1 - \rho)^{-\ell(S)}.
    \end{align*}
    The singleton $S = \set{a}$ is a legal source output since $K \ge 1$.
    Hence
    \begin{align*}
        \widehat O \ge \widehat\omega(\set{a}) \ge q_a(1 + \rho)^{-\overline{D}} > 0.
    \end{align*}
    This already implies that $\widehat \Lambda > 0$.
    To bound the even mass, set $c_\rho = -\overline{D} \ln{(1 - \rho)}$.
    Since $\rho < 1/2$, $c_\rho \le 2\overline{D} \rho = 8/C_*$.
    Using $\ell(S) \le \overline{D} \abs{S}$, the upper bound yields that
    \begin{align*}
        \forall \abs{S} \le K, \quad \widehat\omega(S) \le q_S \e^{c_\rho \abs{S}}.
    \end{align*}
    Every returned set has size at most $K \le r$. Hence, by \Cref{prop:good-property}, $q_S \le p_*^{\abs{S}}$.
    Together with \Cref{lem:connected-set-count}, we estimate the even mass as
    \begin{align*}
        \widehat E \le \sum_{\ell = 2}^{K} \ab(\e \overline{D})^{\ell - 1} \ab(p_* \e^{c_\rho})^{\ell} \le p_* \e^{c_\rho} \frac{z}{1 - z}
    \end{align*}
    where $z = \e \overline{D} p_* \e^{c_\rho} \le 8 p_*\overline{D} \le 2^{-37}$.
    Note that $q_a = 2^{-\abs{\V_{\set{a}}}} \ge 2^{-k}$.
    Therefore, $p_*/q_a \le 2^{12}$.
    We compare $\widehat E$ and the mass at $\set{a}$ as follows:
    \begin{align*}
        \frac{\widehat E}{\widehat \Lambda} &\le \frac{p_* \e^{c_\rho}}{q_a (1 + \rho)^{-\overline{D}}} \frac{z}{1 - z} \\
        &\le 2^{12} \e^{c_\rho + \overline{D} \ln(1 + \rho)} \frac{z}{1 - z}.
    \end{align*}
    Since
    \begin{align*}
        c_\rho + \overline{D} \ln(1 + \rho) \le 2\overline{D} \rho + \overline{D} \rho = \frac{12}{C_*},
    \end{align*}
    the ratio is bounded by $\widehat E/\widehat \Lambda < 1/64$.
    Finally,
    \begin{align*}
        \widehat r = \widehat \Lambda - 2\widehat E > \frac{31}{32} \widehat \Lambda > 0.
    \end{align*}
    The other direction holds directly from $\widehat r \le \widehat \Lambda \le \rho$ by \Cref{lem:connected-subset-proposal}.
    Hence we conclude the complete lemma.
\end{proof}

Finally for every $F \subseteq \C^{\good}$, we provide the following accuracy bound.
\begin{proposition} \label{prop:proposal-accuracy-bound}
    For every $F \subseteq \C^{\good}$, $a \in F$ and $1\le K \le r$, consider $\Propose_K(F, a)$.
    It holds that
    \begin{align*}
        \abs{\widehat r - \mu_{F^-}(B_a)} \le 4d_K.
    \end{align*}
\end{proposition}
\begin{proof}
    We couple the ideal and truncated proposals by maximally coupling every new source request.
    Precisely, for any source request $(H, c)$, let $S_\infty$ have the ideal source law.
    Define
    \begin{align*}
        S_K = \begin{cases}
            S_\infty & \abs{S_\infty} \le K \\
            \emptyset & \abs{S_\infty} > K
        \end{cases}\;.
    \end{align*}
    Then by \Cref{prop:executable-source-accuracy-cost}, $S_K$ has the law $\ExecSource_K(H, c)$ and
    \begin{align*}
        \Pr{S_\infty \neq S_K} = t_K(H,c)
        = \sum_{S \in \CS_H(c) \;:\; \abs{S} > K} A_S(H,c) q_S \le d_K.
    \end{align*}
    Now we regard each source invocation as an atomic kernel.
    Run the ideal and truncated proposals together and share all other randomness.
    While their histories agree, use the coupling above at the common request $(H,c)$.
    Conditional on the history before that request, its discrepancy probability is $t_K(H,c)\le d_K$.
    After the first discrepancy, continue each execution with its own law.
    
    Let $N_\infty$ be the total number of source calls in the ideal execution.
    By \Cref{lem:connected-subset-proposal}, $\E{N_\infty} \le 2$.
    For $j \ge 1$, let $A_j$ be the event that the two histories still agree before the $j$-th ideal source request and this request exists.
    Then by the observation $A_j \subseteq \set{N_\infty \ge j}$,
    \begin{align*}
        \Pr{\text{some source discrepancy occurs}} &\le d_K \sum_{j = 1}^{\infty} \Pr{A_j} \\
        &\le d_K \sum_{j = 1}^{\infty} \Pr{N_\infty \ge j} \\
        &= d_K \E{N_\infty} \\
        &\le 2d_K.
    \end{align*}
    Define the signed parity function
    \begin{align*}
        f(S) = \begin{cases}
            0 & S = \emptyset \\
            1 & \abs{S} \text{ is odd} \\
            -1 & S \neq \emptyset \text{ and } \abs{S} \text{ is even}
        \end{cases}\;.
    \end{align*}
    For the truncated proposal,
    \begin{align*}
        \E[Y_K \sim \Propose_K(F, a)]{f(Y_K)} = \widehat O - \widehat E = \widehat r.
    \end{align*}
    For the ideal proposal, \Cref{lem:connected-subset-inclusion-exclusion} provides that
    \begin{align*}
        \E[Y_\infty \sim \Propose[q](F, a)]{f(Y_\infty)} = \sum_{S \in \CS_F(a)} (-1)^{\abs{S} - 1} w_S = \mu_{F^-}(B_a).
    \end{align*}
    Under our coupling, the two parity values agree unless a source discrepancy occurs, and their difference is always at most two.
    Then
    \begin{align*}
        \abs{\widehat r - \mu_{F^-}(B_a)} &\le \E{\abs{f(Y_K) - f(Y_\infty)}} \\
        &\le 2 \Pr{\text{some source discrepancy occurs}} \\
        &\le 4d_K.
    \end{align*}
    Therefore we conclude this accuracy bound.
\end{proof}

\subsection{Local marginal coin}
In this part, we aim to construct a local marginal coin estimating the probability $\mu_{F^-}(B_a)$.
The first preparation is the following doubling factory for a Bernoulli coin.
\begin{proposition}[Doubling factory] \label{prop:doubling-factory}
    Given independent Bernoulli bits $\Ber(p_0)$ with success probability $p_0 \in [0, 1/64]$, one can generate a $\Ber(2p_0)$ bit using at most $5$ input bits in expectation.
\end{proposition}
\begin{proof}
    Draw $4$ input bits.
    Return zero unless exactly one is $1$.
    Apply the reciprocal test to an oracle that returns zero when $3$ fresh input bits are all zero.
    Its zero probability is $b = (1 - p_0)^3 \ge (63/64)^3 > 9/10$ and thus by \Cref{lem:reciprocal-test}, the resulting success probability is $4p_0(1 - p_0)^3/(2(1 - p_0)^3) = 2p_0$.
    The factory uses at most $4$ oracle calls consuming $3$ input bits each in expectation.
    Therefore the total cost is bounded by $4 + 48p_0(1 - p_0)^3 \le 5$.
\end{proof}

\begin{definition}[Local marginal coin] \label{def:local-marginal-coin}
    For $a \in F \subseteq \C^{\good}$, the procedure $\LocalMarginalCoin_K(F, a)$ estimates $\mu_{F\setminus\set{a}}(B_a)$ and is implemented as follows.
    \begin{enumerate}
        \item Draw one proposal from $\Propose_K(F, a)$.
        
        \item If the proposal is empty, return zero.
        
        \item Define a fresh parity oracle: repeatedly draw independent new proposals until one is non-empty;
        return one for an even subset and zero for an odd set.
        
        \item Use the doubling factory in \Cref{prop:doubling-factory} for the parity oracle and return the complement of its output.
    \end{enumerate}
\end{definition}

\begin{proposition} \label{prop:executable-local-marginal-coin}
    The procedure $\LocalMarginalCoin_K(F, a)$ returns a $\Ber\ab(\widehat{r})$ bit using at most $6$ top-level proposals and $12$ source calls in expectation.
    Moreover, whenever $4d_K \le \zeta$ and $K \le r$, its expected cost is $T_{\coin}(\zeta) = \widetilde{O}\ab(\poly{k, \overline{D}})$ for every sub-deletion family $F \subseteq \C^{\good}$.
\end{proposition}
\begin{proof}
    For $\widehat \Lambda > 0$, the parity oracle has success probability $p_0 = \widehat E/\widehat \Lambda$.
    By \Cref{lem:parity-bound}, $p_0 \le 1/64$ and \Cref{prop:doubling-factory} applies.
    Consequently,
    \begin{align*}
        \Pr{\LocalMarginalCoin_K(F, a) = 1} = \widehat \Lambda \ab(1 - 2p_0) = \widehat r.
    \end{align*}
    When $\widehat \Lambda = 0$, the initial proposal returns the empty label almost surely and thus the procedure returns zero.
    The above argument concludes that $\LocalMarginalCoin_K(F, a)$ returns a $\Ber\ab(\widehat{r})$ bit.

    To count the top-level proposals, when $\widehat \Lambda = 0$, it is trivial to see that the procedure uses exactly one call.
    Otherwise suppose that $\widehat \Lambda > 0$.
    Each call for the fresh parity oracle requires a geometric number of proposals with mean $1/\widehat\Lambda$.
    Let $M$ be the number of parity-oracle calls requested by the doubling factory and $W_i$ be the number of proposals used by the $i$-th parity call.
    Then by \Cref{prop:doubling-factory}, $\E{M} \le 5$ and
    \begin{align*}
        \E{\sum_{i = 1}^M W_i} &= \sum_{i \ge 1} \E{\id{M \ge i} \cdot W_i} \\
        &= \frac{1}{\widehat\Lambda} \sum_{i = 1}^{\infty} \Pr{M \ge i} \\
        &\le \frac{5}{\widehat\Lambda}.
    \end{align*}
    Back to the whole procedure, there is always one initial proposal and the subsequent work is activated with probability exactly $\widehat\Lambda$.
    Thus
    \begin{align*}
        \E{\text{the number of proposal calls}} \le 1 + 5 = 6.
    \end{align*}
    Similarly, together with \Cref{lem:connected-subset-proposal}, the expected number of source calls is bounded by $12$.

    For the cost, by \Cref{prop:executable-source-accuracy-cost,lem:connected-subset-proposal}, a complete executable proposal has expected cost
    \begin{align*}
        B = \widetilde O\ab(\poly{k, \overline D}).
    \end{align*}
    Charging that conditional expected cost whenever a proposal is requested provides that the proposal work costs at most $6B$ in expectation.
    Together with \Cref{prop:doubling-factory}, we conclude that
    \begin{align*}
        \E{\Cost\ab(\LocalMarginalCoin_K(F, a))} = \widetilde O\ab(\poly{k, \overline D}).
    \end{align*}
    The cost bound uniformly holds for the cutoff.
    When $4d_K \le \zeta$, the cutoff $K$ will be picked as
    \begin{align*}
        K = \max\;\set{1, \left\lceil\log_R \frac{8\rho}{\zeta}\right\rceil}.
    \end{align*}
    Provided $K \le r$, all arguments hold and so does the cost bound.
\end{proof}

\subsection{A residual sampler with global abort and restart}
Now we provide the sampler for the residual formula $\Psi$.
For a better understanding, we first provide an ideal design just for analysis and then substitute it with an executable algorithm.

\subsubsection{Ideal insertion and deletion transports}
The ideal design uses the ideal source and thus the ideal proposing procedure $\Propose[q](\cdot, \cdot)$.
For a label $S$, its \emph{preparation} recursively deletes $b_{\ell(S)}, \ldots, b_1$ and then sets the variables in $\vbl^{\good}(S)$ to $u_S$.
This sends $\mu_{F^-}$ to $\eta_S$.
Its restoration redraws $\vbl^{\good}(S)$ as independent fair bits and recursively inserts $b_1, \ldots, b_{\ell(S)}$, sending $\eta_S$ to $\mu_{F^-}$.
Both preparation and restoration procedures employ the following update routine.
\begin{definition}[The update routine] \label{def:update}
    Insertion starts from an assignment satisfying $F^-$, and deletion starts from an assignment satisfying $F$. In either direction the individual pass below acts on assignments satisfying $F^-$, including those reached after a deletion has continued. Repeat the pass until it returns.
    \begin{enumerate}
        \item Independently draw $S$ from $\Propose[q](F, a)$.
        For deletion, return $X$ immediately if $S$ is non-empty and even.
        
        \item Flip a fair coin.
        On heads, start a new pass for deletion. For insertion on heads, return $X$ if $X\notin B_a$ and otherwise set $T\gets\Route(F,a,\star;X)$ and go to step~4.

        \item On tails, insertion starts a new pass if $S$ is empty or odd. Deletion returns $X$ if $S$ is empty. In the remaining tails cases (even $S$ for insertion, odd $S$ for deletion), prepare $S$ and set $T\gets\Route(F,a,S;X)$.

        \item For deletion, if the routed result is exceptional, start a new pass directly.
        Otherwise, restore $X$ from the routed label and start a new pass both for deletion and insertion.
    \end{enumerate}
    Here $\Route(F,a,S;X)=\iota_X(S)$ is the pairing of \Cref{lem:parity-pairing}. On exceptional input it returns $\set{a}$. On ordinary input the canonical exploration examines $N_F[S]$ to identify the partner.
    Preparation is bypassed on the exceptional insertion branch because $X\models F^-$ and $X\in B_a$ imply $Q(X)=\set{a}$. Thus $X$ already belongs to $E_{\set{a}}\cap\Good(F_{\set{a}})$ and can be routed directly from $\star$ to $\set{a}$ before restoration.
\end{definition}

\begin{lemma} \label{lem:ideal-update-accuracy-and-cost}
    The ideal routines terminate almost surely and satisfy
    \begin{align*}
        \mu_{F^-} \Update(F, a, +) = \mu_F, \qquad \mu_F \Update(F, a, -) = \mu_{F^-}.
    \end{align*}
    Under the respective input laws displayed above, each routine uses at most $16$ source requests in expectation, including all proposal and update descendants.
    Excluding implementation of the source, the expected cost is $\widetilde O\ab(\poly{k, \overline D})$.
\end{lemma}
\begin{proof}
    The proof is similar to \cite[Lemma 6.2]{LWYZZ26}.
    For simplicity, write $\mu^- = \mu_{F^-}$, $\mu^+ = \mu_F$, $r_a = \mu^-(B_a)$ and $\beta = (1 - r_a)/2$.
    We induct on $|F|$. All recursive transports act on smaller families, and the case $|F|=1$ has no child transports. Because the prescribed uniform laws are positive on their supports, almost-sure termination under a smaller-family input law also holds at each assignment in its support.
    Let $\OO, \EE$ be the odd and even ideal proposal masses.
    By \Cref{lem:connected-subset-inclusion-exclusion}, $\OO - \EE = r_a$.

    Evaluate one pass in either direction on the reference input law $\mu^-$. For insertion, the successful heads branch returns $\frac12\mu^-\one_{\neg B_a}=\beta\mu^+$.
    Empty or odd tails continue with measure $\frac{1}{2} (1 - \EE) \mu^-$.
    After preparation, routing and restoration, failed heads and even tails together continue with measure $\frac{1}{2} \OO \mu^-$.
    Therefore the total continuation measure is $(1 - \beta) \mu^-$.

    For deletion, the even-proposal and empty-tails branches return
    \begin{align*}
        \ab(\EE + \frac{1}{2}\ab(1 - \OO - \EE)) \mu^- = \beta \mu^-.
    \end{align*}
    The heads branch continues with measure $\frac{1}{2}(1 - \EE) \mu^-$.
    Routing odd tails in the reverse direction, restoring the even terms and continuing the exceptional term directly, contributes $\frac{1}{2}\EE \mu^- + \frac{1}{2}\mu^- \one_{B_a}$.

    Therefore, we obtain the one-pass identities
    \begin{equation} \label{eq:update-ledger}
        \begin{array}{c|cc}
             &\text{returned measure} &\text{continuation measure}\\
            \hline
            + & \beta \mu^+ & (1 - \beta) \mu^- \\
            - & \beta \mu^- & \mu^- - \beta \mu^+.
        \end{array}
    \end{equation}
    For insertion from $\mu^-$, summing the geometric continuation measures, we obtain the total entering measure $\beta^{-1} \mu^-$.
    Thus the returned law of the insertion is $\mu^+$.
    For deletion, let $\K^-$ be the continuation kernel on assignments satisfying $F^-$. From the deletion row of \eqref{eq:update-ledger}, $\mu^+=\beta^{-1}(\mu^--\mu^-\K^-)$. Thus for every $N\ge1$,
    \[
        \sum_{t=0}^{N-1}\mu^+(\K^-)^t
          =\beta^{-1}\bigl(\mu^--\mu^-(\K^-)^N\bigr).
    \]
    At every input the even-proposal and empty-tails branches return with total probability $\EE+(1-\OO-\EE)/2=\beta$, independently of the assignment. Hence $\mu^-(\K^-)^N$ has mass $(1-\beta)^N$. Taking the limit, the cumulative measure entering the deletion passes from the actual input law $\mu^+$ is
    \begin{equation}\label{eq:update-deletion-occupation}
        \sum_{t\ge0}\mu^+(\K^-)^t=\beta^{-1}\mu^-.
    \end{equation}
    Applying the one-pass return kernel to this measure gives $\beta^{-1}(\beta\mu^-)=\mu^-$. In both directions the returned measure has mass one, which proves almost-sure termination.

    For the cost, observe that the expected total number of passes is $\beta^{-1} \le 4$.
    Each proposal uses at most $2$ sources in expectation by \Cref{lem:connected-subset-proposal}.
    By \Cref{lem:core-assignment-good-variable-extension},
    \begin{align*}
        w_S \le q_S (1 - x)^{-\ell(S)} \le q_S \e^{b_x \abs{S}}
    \end{align*}
    where $b_x = -\overline{D} \ln(1 - x) \le 2\overline D x < 1/2$.
    Since $\tau \ge 1$, by \Cref{prop:root-event},
    \begin{align*}
        \sum_{S \in \CS_F(a)} w_S \ell(S) \le \overline D \sum_{S \in \CS_F(a)} \abs{S} q_S \e^{b_x\abs{S}} \le \overline D \rho.
    \end{align*}
    Likewise, for every fixed integer $j \ge 0$,
    \begin{align} \label{eq:moment-bound}
        \sum_{S \in \CS_F(a)} \abs{S}^j w_S \le O_j(\rho).
    \end{align}
    We now compute the input measures of the child updates. In either direction the cumulative measure entering the parent passes is $\beta^{-1}\mu^-$, by the insertion row of \eqref{eq:update-ledger} and by \eqref{eq:update-deletion-occupation}. Fix an ordinary label $S$ with $w_S>0$ and put $c_S=w_S/(2\beta)=w_S/(1-r_a)$. For insertion, preparation is used for even $S$ and restoration for odd $S$. For deletion these parities are reversed.

    Whenever the direction and parity call for preparation of $S$, the independent proposal and tails coin select it with probability $w_S/2$ per pass. Its cumulative input measure is therefore $c_S\mu^-$. By induction on the smaller families, the deletion of $b_{S,j}$ in that preparation receives measure $c_S\mu_{F_{S,j}}$: starting at $F_{S,\ell(S)}=F^-$, each deletion sends this measure to $c_S\mu_{F_{S,j-1}}$. Fixing the variables in $\V_S$ after the last deletion gives $c_S\eta_S$.

    For restoration the measure is determined after routing. By the measure-preserving bijection on assignment--label pairs in \Cref{lem:parity-pairing}, routing the even labels together with the exceptional branch gives mass $c_S\eta_S$ at each odd destination $S$. Indeed, the even inputs carry $c_T\eta_T$ at each label $T$, and the exceptional input carries $(2\beta)^{-1}\mu^-\one_{B_a}$. The reverse routing for deletion gives $c_S\eta_S$ at each even destination. Redrawing $\V_S$ sends this measure to $c_S\mu_{F_S}$ because no clause in $F_S$ uses those variables. Thus the insertion of $b_{S,j}$ in restoration receives $c_S\mu_{F_{S,j-1}}$ and sends it to $c_S\mu_{F_{S,j}}$.

    Each fixed child consequently receives a scalar multiple of its specified uniform input law when summed over all parent passes. Bounding preparation and restoration separately, the expected number of direct child updates is at most
    \[
        2\sum_{S\in\CS_F(a)}c_S\ell(S)
        \le 4\overline D\rho=\frac{16}{C_*},
    \]
    since $r_a\le x<1/2$. Let $M_s$ bound the expected total source count in either direction on families of size at most $s$ under the respective input law. Each of the above child measures contributes at most its mass times $M_{s-1}$ to that count. Together with the independent proposal at the start of each pass, this yields
    \begin{align*}
        M_s \le 2\beta^{-1} + \frac{16}{C_*} M_{s - 1}
    \end{align*}
    with boundary condition $M_0 = 0$.
    Thus we conclude that $M_s \le 8/(1 - 16/C_*) \le 16$.
    To bound the bit cost, processing an ordinary label requires polynomial local work in $k, \overline{D}$ and $\abs{S}$ and with \eqref{eq:moment-bound}, it takes $\poly{k, \overline{D}}$ bit operations in expectation.
    Then the additional work costs $\poly{k, \overline{D}}$ in expectation by a similar recursive bound as $M_s$.
\end{proof}

\subsubsection{Stopped sources and normalization after restart}
To implement an executable residual sampler, we replace the ideal source with a truncated one.
Direct replacement controls accuracy, but does not necessarily control the running cost after an assignment transport first deviates from its ideal law.
We instead use point-wise domination of successful histories.
\begin{definition} \label{def:stopped-source}
    The subroutine $\StoppedSource_K(F, a)$ runs as follows.
    \begin{enumerate}
        \item Draw a label $S$ from $\ExecSource_K(F, a)$.
        \item If $S$ is non-empty, directly return $S$.
        \item Flip a Bernoulli coin with rate
        \begin{align*}
            a_K \defeq \frac{d_K}{1 - 2\rho} = \frac{2\rho}{1 - 2\rho}R^{-K}.
        \end{align*}
        \item If the coin is on heads, abort the entire residual construction attempt; otherwise return the empty label.
    \end{enumerate}
\end{definition}

The implementation of $\Ber(a_K)$ coins can be done as: firstly flip a $\Ber(2\rho/(1 - 2\rho))$ coin; then flip at most $K$ independent $\Ber\ab(R^{-1})$ coins, stopping at the first failure.
Return one if and only if all coins are on heads.

\begin{proposition}[Domination of complete histories] \label{prop:history-domination}
    Start from a fresh fair good-variable assignment and insert every label in $F$ sequentially.
    At each source request $(H,c)$ use $\StoppedSource_K(H,c)$, and propagate every abort through the entire call stack.
    Define $\Lambda_F$ to be the successful-output sub-measure of one attempt, \IE,
    \begin{align*}
        \Lambda_F(A) = \Pr{\text{one attempt succeeds and returns an assignment in }A}.
    \end{align*}
    Let $p_F = \Lambda_F(\set{0, 1}^{\V^\good})$ be the total mass of $\Lambda_F$.
    Then
    \begin{align*}
        \Lambda_F \le \mu_F, \qquad 1 - p_F \le 32\abs{F}d_K.
    \end{align*}
    One attempt has expected cost $\widetilde O\ab(n + \abs{F}\poly{k, \overline D})$ in addition to a one-time precision setup cost.
\end{proposition}
\begin{proof}
    We couple the ideal construction and the stopped attempt.
    First, for the ideal construction, by \Cref{lem:ideal-update-accuracy-and-cost}, it terminates almost surely and its final output $X_\ideal$ satisfies $\Law(X_\ideal) = \mu_F$.
    Let $N_\ideal$ be the total number of source requests including all recursive descendants.
    Then
    \begin{align*}
        \E{N_\ideal} \le 16\abs F.
    \end{align*}
    Its expected work outside the sources is bounded by
    \begin{align*}
        \widetilde O\ab(n + \abs F\poly{k, \overline D}).
    \end{align*}

    Now we compare the stopped source with its ideal counterpart.
    Fix a request $(H,c)$ and write $B_K(H,c)=\Pr{\ExecSource_K(H,c)=\emptyset}$.
    By \Cref{prop:local-weight-bounds,prop:executable-source-accuracy-cost}, $B_K(H,c)\ge1-2\rho$.
    Truncation moves exactly $t_K(H,c)$ of the ideal source mass to the empty label, so its ideal mass is $B_K(H,c)-t_K(H,c)$. By the choice of $a_K$, its stopped mass satisfies
    \begin{align*}
        (1-a_K) B_K(H,c) \le B_K(H,c)-d_K \le B_K(H,c)-t_K(H,c).
    \end{align*}
    For every non-empty label of size at most $K$, its stopped probability equals its ideal probability.
    For large labels, the stopped probability is zero.
    Let $P_{H, c}$ be the ideal source law and
    \begin{align*}
        Q_{H, c}(S) = \Pr{\StoppedSource_K(H, c) = S}
    \end{align*}
    for every $S\in\CS_H(c)\cup\{\emptyset\}$. Thus, for all these labels,
    \begin{align*}
        Q_{H, c}(S) \le P_{H, c}(S).
    \end{align*}
    The conditional abort probability is
    \begin{align*}
        \beta_{H, c} = a_K B_K(H,c) \le \frac{d_K}{1 - 2\rho} \le 2d_K.
    \end{align*}
    For the coupling, draw a label $S$ from $P_{H,c}$ and retain it with probability
    \begin{align*}
        \theta_{H, c}(S) = \frac{Q_{H, c}(S)}{P_{H, c}(S)}
    \end{align*}
    for each $S$ with $P_{H,c}(S)>0$. Abort if the retention test fails.
    The retained mass at $S$ is $Q_{H,c}(S)$, so this gives the law of $\StoppedSource_K(H,c)$.

    Apply this coupling at every source request including recursive proposals, updates, and reciprocal tests.
    Share the initial fair good-variable assignment and all randomness outside source implementations.
    Until an abort, the two executions have identical histories and therefore request identical source inputs.
    When an abort occurs, we terminate the entire stopped attempt while allowing the ideal execution to continue.

    Let $N^\diamond$ be the number of source requests made by the stopped attempt including a request that aborts.
    Under this coupling, $N^\diamond \le N_\ideal$ almost surely.
    Therefore, the stopped attempt terminates almost surely.
    On success, the stopped output is exactly the ideal output.
    Thus for every $A$,
    \begin{align*}
        \Lambda_F(A) = \Pr{\text{stopped attempt succeeds and }X_\ideal \in A} \le \Pr{X_\ideal \in A} = \mu_F(A).
    \end{align*}
    This proves $\Lambda_F \le \mu_F$.

    Now we bound the abort probability.
    The preceding history determines each source request. Conditional on that history, its abort probability is $\beta_{H,c}\le2d_K$.
    Observe that the first abort stops the entire attempt.
    Thus by the domination $N^\diamond \le N_\ideal$ almost surely,
    \begin{align*}
        \Pr{\text{the attempt aborts}} &\le 2d_K \sum_{j \ge 1} \Pr{N^\diamond \ge j} \\
        &\le 2d_K \sum_{j \ge 1} \Pr{N_\ideal \ge j} \\
        &\le 2d_K \E{N_\ideal} \\
        &\le 32d_K \abs{F}.
    \end{align*}
    Therefore $1 - p_F \le 32 \abs{F} d_K$.

    To conclude the complete proposition, we provide the cost guarantee.
    The abort coin is implemented by a $\Ber(2\rho/(1 - 2\rho))$ coin and at most $K$ Bernoulli coins with rate $R^{-1}$.
    The expected number of these elementary coin calls is at most $1 + (2\rho)/(1 - 2\rho)\cdot 4/3 \le 7/3$.
    The rational coin parameters have local-parameter bit length and the counter costs are absorbed by the polylogarithmic factors.
    By \Cref{prop:executable-source-accuracy-cost}, for each actual stopped-source request, the expected cost is at most $\widetilde O\ab(\poly{k, \overline D})$.
    Therefore, the total expected cost for all source implementations is bounded as $\widetilde O\ab(\abs{F} \poly{k, \overline D})$.
    Moreover, all other work follows an ideal execution prefix under the coupling.
    Abort propagation adds at most constant work per active frame.
    Therefore, the expected cost of an entire stopped attempt is bounded by
    \begin{align*}
        \widetilde O\ab(n + \abs{F} \poly{k, \overline D}),
    \end{align*}
    which concludes the whole proposition.
\end{proof}

Now we are ready to provide the executable good-variable sampler.
\begin{definition}[Good-variable sampler] \label{def:good-variable-sampler}
    Let $F\subseteq\C^\good$ and let the rational tolerance $\delta\in(0,1)$ be given by a binary numerator and denominator of total bit length $L_\delta$.
    If $F=\emptyset$, return independent fair bits on $\V^\good$. Otherwise set
    \begin{align} \label{eq:truncated-length-choice}
        K_\delta \defeq \max\set{1, \left\lceil \log_R \frac{128 \rho (m + 1)}{\delta}\right\rceil}.
    \end{align}
    Fix an order on $F = \set{a_1, \ldots, a_{\abs F}}$ and set
    \begin{align*}
        F_0 = \emptyset, \qquad F_j = \set{a_1, \ldots, a_j}.
    \end{align*}
    The $\GoodSampler(F, \delta)$ runs as follows.
    \begin{enumerate}
        \item Draw a fresh assignment $X \sim \nu^\good$.
        
        \item For $j = 1, \ldots, \abs{F}$, run
        \begin{align*}
            X \gets \Update_K^\diamond(F_j, a_j, +; X)
        \end{align*}
        where $\Update_K^\diamond$ is the procedure obtained from $\Update$ by replacing all ideal sources with $\StoppedSource_K$.

        \item If any abort occurs, discard the current assignment and restart from the first step using independent randomness.
        
        \item If all $\abs{F}$ insertions finish, return $X$.
    \end{enumerate}
\end{definition}

\begin{lemma}[Executable good-variable sampler] \label{lem:good-variable-sampler}
    Let $0<\delta<1/2$ be rational. If $F=\emptyset$ or the cutoff $K_\delta$ in \eqref{eq:truncated-length-choice} is at most $r$, then $\GoodSampler(F,\delta)$ terminates almost surely and returns a satisfying assignment $X$ with $\TV{\Law(X)}{\mu_F}\le\delta$.
    For fixed $C>0$ and $\delta\ge n^{-C}$, its expected bit cost, including precision setup, is
    \[
        \widetilde O_C\ab(n+m\poly{k,\overline D}+L_\delta).
    \]
\end{lemma}
\begin{proof}
    For $F=\emptyset$ the returned law is exactly $\mu_\emptyset=\nu^\good$ and drawing the assignment costs $O(n)$. Assume henceforth that $F\ne\emptyset$ and put $K=K_\delta\le r$.
    By \Cref{prop:history-domination}, $\Lambda_F \le \mu_F, 1 - p_F \le 32 \abs F d_K$ and one attempt terminates almost surely and has expected cost $\widetilde O\ab(n + \abs F \poly{k, \overline D})$ apart from a one-time precision setup.

    To bound the failure probability, the choice of $K$ provides that
    \begin{align*}
        R^K \ge \frac{128\rho(m + 1)}{\delta}
    \end{align*}
    and therefore $d_K = 2\rho R^{-K} \le \delta/(64(m + 1))$.
    Then $1 - p_F \le \delta/2$, meaning that the successful probability for a stopped attempt is at least $3/4$.
    Therefore, the expected number of attempts is at most $4/3$, implying that $\GoodSampler(F, \delta)$ terminates almost surely.

    Let $X$ be its returned assignment.
    For every event $A$,
    \begin{align*}
        \Pr{X \in A} &= \sum_{j \ge 1} \Pr{\text{the first }j - 1\text{ attempts fail, and the attempt }j\text{ returns in }A} \\
        &= \sum_{j \ge 1} (1 - p_F)^{j - 1} \Lambda_F(A) \\
        &= \frac{\Lambda_F(A)}{p_F}.
    \end{align*}
    Therefore the output law is exactly $\pi_F \defeq \Law(X) = \Lambda_F/p_F$.
    The domination $\Lambda_F \le \mu_F$ immediately gives that $X$ is a satisfying assignment of $F$.

    To control the effect of the normalization $p_F$, when $p_F = 1$, $\Lambda_F = \pi_F = \mu_F$, therefore the total variation distance is zero.
    When $p_F < 1$, the sub-measure $\mu_F - \Lambda_F$ is non-negative and has mass $1 - p_F$.
    Then we have the following decomposition
    \begin{align*}
        \mu_F = p_F \pi_F + (1 - p_F) \eta_F
    \end{align*}
    where $\eta_F = (\mu_F - \Lambda_F)/(1 - p_F)$.
    Then it holds that
    \begin{align*}
        \TV{\pi_F}{\mu_F} &= (1 - p_F) \TV{\pi_F}{\eta_F} \le 1 - p_F \le \delta/2.
    \end{align*}

    To justify the cost, let $C_j$ be the cost of attempt $j$ excluding the shared precision setup.
    For analysis, imagine an infinite sequence of independent attempts, even though the algorithm stops after $N$ attempts.
    The total attempt cost is
    \begin{align*}
        \sum_{j \ge 1} C_j \id{N \ge j}.
    \end{align*}
    The event $\id{N \ge j}$ depends only on the preceding attempts and therefore by Tonelli's theorem,
    \begin{align*}
        \E{\sum_{j \ge 1} C_j \id{N \ge j}} = \sum_{j \ge 1} \E{C_j} \Pr{N \ge j} = \frac{\E{C_1}}{p_F}.
    \end{align*}
    Then the total expected cost of all attempts is therefore
    \begin{align*}
        \widetilde O\ab(n + m \poly{k, \overline D}).
    \end{align*}
    For precision setup, $K$ is the least integer $j\ge1$ such that $\delta R^j\ge128\rho(m+1)$. Doubling the search bound and then using binary search finds $K$ by exact rational comparisons. The numerator and denominator of $R$ have $O(k+\log_2(\overline D+2))$ bits, so repeated squaring and fast integer multiplication give a setup cost of
    \begin{align*}
        \widetilde O\ab((k + \log_2(\overline D + 2)) K + L_\delta + \log_2(m + 1)).
    \end{align*}
    Since $R \ge 4$ and $\rho < 1$,
    \begin{align*}
        K = O\ab(1 + \log_2(m + 1) + \log(1/\delta)).
    \end{align*}
    For $\delta\ge n^{-C}$ this gives $K=O_C(1+\log(m+1)+\log n)$. Since $F\ne\emptyset$ implies $m\ge1$, adding setup to the cost of the attempts gives the stated bound. In particular, for $\delta=2^{-b}$ in this range the numerator and denominator have $L_\delta=O(b+1)=O_C(1+\log n)$ bits.
\end{proof}

\section{Putting All Together: Design of the Complete Sampler} \label{sec:full-sampler-design}

Now we are ready to show \Cref{thm:fast-sampler}.
The main obstacle is that, the analysis in \Cref{sec:core-fast-sampling} for $\gap(\K)$ defines the exact marginal updates while we can only `approximately' simulate the updates of $\K$ by the local marginal coin $\LocalMarginalCoin$ and sample a solution by the executable sampler $\GoodSampler$.
Here we show how to appropriately construct the sampler, provide the accuracy guarantee and the total running cost, and finally conclude a proof of \Cref{thm:fast-sampler}.

Before proceeding, we assume that $\Phi$ is a nice formula satisfying \Cref{cond:nice-formula} and therefore all deterministic properties hold as in \Cref{sec:core-fast-sampling}.
Write $\pi = \mu^{\bad}$ for the target core marginal.

\subsection{The simulated heat-bath update}
Now we design the simulated heat-bath update.
Fix an index $i \in [s]$ and a feasible assignment $\sigma$ on $\V^{\bad}$.
Recall that $\sigma_{-i}$ is the projection of $\sigma$ on $\V^\bad \setminus B_i$.
Let $\vbl(C_a)$ denote the variables occurring in the original clause $C_a$ of $\Phi$. Define
\[
    \J_i \defeq \{a\in\C^\good:\vbl(C_a)\cap B_i\ne\emptyset\},
    \qquad H_{\sigma,i}\defeq D_\sigma\setminus\J_i.
\]
For $a\in\C^\good\setminus\J_i$, the core literals of $C_a$ all lie outside $B_i$, so its membership in $D_\sigma$ depends only on $\sigma_{-i}$. In particular, $H_{\sigma,i}$ contains every purely good clause.
For $z\in\Omega_i$, put $L_{i,\sigma}(z)=D_{(\sigma_{-i},z)}\cap\J_i$. Then
\[
    D_{(\sigma_{-i},z)}=H_{\sigma,i}\,\dot\cup\,L_{i,\sigma}(z).
\]
\begin{definition}[Core heat-bath update] \label{def:core-heat-bath-update}
    The $\CoreHeatBath(\sigma, i; \zeta)$ is the subroutine updating one selected core component $B_i$ as follows.
    \begin{enumerate}
        \item Set the clause family $H = H_{\sigma, i}$ and the cutoff
        \begin{align*}
            K = \max\;\set{1, \left\lceil \log_R \frac{8\rho}{\zeta}\right\rceil}.
        \end{align*}
        \item Draw a uniformly random assignment $z \in \Omega_i$ using the rejection sampler of \Cref{lem:core-feasible-assignment}.
        \item Start from $F = H$.
        Test each clause $a \in L_{i, \sigma}(z)$ by calling $\LocalMarginalCoin_K(F \cup \set{a}, a)$ with fresh randomness and:
        \begin{itemize}
        \item if it returns $1$, discard the trial family $F$ and restart from the second step with the same $H$;
            \item otherwise, add $a$ to $F$ and continue to test the next clause.
        \end{itemize}
        \item When all tests pass, replace $\sigma_i$ by $z$, commit $F$ and return.
    \end{enumerate}
\end{definition}

\begin{proposition} \label{prop:core-heat-bath-update-accuracy-cost}
    For $\CoreHeatBath(\sigma, i; \zeta)$, let $r_i$ be the number of mixed clauses incident to $B_i$, $\underline a_i = (1 - x)^{r_i}$, $U_i$ be the expected cost of generating a uniform feasible candidate and $T_\coin(\zeta)$ be the uniform expected cost of one local coin.
    Assume that the cutoff $K \le r$ and $r_i \zeta \le \underline a_i/2$.
    Then the following guarantees hold.
    \begin{itemize}
        \item \textbf{Accuracy:} the procedure terminates almost surely and returns a feasible assignment $\sigma'$ such that
        \begin{align*}
            \TV{\Law(\sigma'_i \mid \sigma)}{\pi(\cdot \mid \sigma_{-i})} \le \frac{r_i \zeta}{\underline{a}_i}.
        \end{align*}

        \item \textbf{Expected cost:} after shared preprocessing, including every rejected trial,
        \begin{align*}
            \E{\Cost(\CoreHeatBath(\sigma, i; \zeta))} \le \widetilde O\ab(\frac{U_i + (r_i + 1) T_\coin(\zeta)}{(1 - x)^{r_i}}).
        \end{align*}
    \end{itemize}
\end{proposition}
\begin{proof}
    We list $L_{i, \sigma}(z)$ as
    \begin{align*}
        L_{i, \sigma}(z) = \set{a_1, \ldots, a_\ell} \quad \text{where } \ell \le r_i,
    \end{align*}
    and set
    \begin{align*}
        F_0 = H, \qquad F_j = F_{j - 1} \cup \set{a_j}, \quad \forall j = 1, \ldots, \ell.
    \end{align*}
    Set $\nu_j=\mu_{F_{j-1}}(B_{a_j})$. By \Cref{lem:core-assignment-good-variable-extension}, $\nu_j\le x$.
    Each test uses the current family $F_j$ with the newly added clause $a_j$. Since $F_j\setminus\{a_j\}=F_{j-1}$, its ideal violation probability is $\nu_j$. Thus the ideal acceptance probability of $z$ is
    \begin{align*}
        A_z = \prod_{j=1}^\ell (1-\nu_j) \ge (1-x)^{r_i}=\underline a_i.
    \end{align*}
    Each factor $1-\nu_j$ equals $Z_{F_j}/Z_{F_{j-1}}$, and $F_\ell=H\cup L_{i,\sigma}(z)=D_{(\sigma_{-i},z)}$. Hence $A_z=Z_{D_{(\sigma_{-i},z)}}/Z_H$. The common denominator $Z_H$ cancels when normalizing the core extension weights, so
    \begin{align*}
        p(z) \defeq \pi(z \mid \sigma_{-i}) = \frac{A_z}{\sum_{u \in \Omega_i} A_u}.
    \end{align*}
    Now let $\widehat \nu_j$ be the probability that the implemented local coin returns $1$.
    By \Cref{prop:proposal-accuracy-bound,prop:executable-local-marginal-coin} and the cutoff choice $4d_K\le\zeta$,
    \begin{align*}
        \abs{\widehat \nu_j - \nu_j} \le \zeta.
    \end{align*}
    Then the implemented accepting probability of $z$ is
    \begin{align*}
        \widehat A_z = \prod_{j = 1}^\ell \ab(1 - \widehat \nu_j).
    \end{align*}
    We then calculate
    \begin{align*}
        \abs{\widehat A_z - A_z} \le \sum_{j = 1}^\ell \abs{\widehat \nu_j - \nu_j} \le r_i \zeta.
    \end{align*}
    Define the average acceptance probability as
    \begin{align*}
        \mathscr A_i \defeq \frac{1}{\abs{\Omega_i}} \sum_{z \in \Omega_i} A_z, \qquad \widehat{\mathscr A}_i \defeq \frac{1}{\abs{\Omega_i}} \sum_{z \in \Omega_i} \widehat A_z.
    \end{align*}
    It holds that $\widehat{\mathscr A}_i \ge \mathscr A_i - r_i \zeta \ge \underline a_i/2 > 0$.
    Thus the number of trials follows a geometric distribution of mean $1/\widehat{\mathscr A}_i \le 2/\underline a_i$.
    Observe that its first acceptance candidate has law
    \begin{align*}
        \widehat p(z) = \sum_{t \ge 1} \ab(1 - \widehat{\mathscr A}_i)^{t - 1} \frac{\widehat A_z}{\abs{\Omega_i}} = \frac{\widehat A_z}{\abs{\Omega_i} \widehat{\mathscr A}_i}.
    \end{align*}
    Therefore the procedure terminates almost surely and the output is a feasible core assignment. For the accuracy bound, add and subtract $\widehat A_z/(\abs{\Omega_i}\mathscr A_i)$ inside each absolute value to obtain
    \begin{align*}
        \TV{p}{\widehat p} &= \frac{1}{2} \sum_{z \in \Omega_i} \abs{\frac{A_z}{\abs{\Omega_i}\mathscr A_i} - \frac{\widehat A_z}{\abs{\Omega_i} \widehat{\mathscr A}_i}} \\
        &\le \frac{1}{2\abs{\Omega_i}\mathscr A_i} \sum_{z \in \Omega_i} \abs{A_z - \widehat A_z}
        + \frac{1}{2\abs{\Omega_i}} \sum_{z \in \Omega_i} \widehat A_z
          \abs{\frac{1}{\mathscr A_i} - \frac{1}{\widehat{\mathscr A}_i}} \\
        &= \frac{1}{2\abs{\Omega_i}\mathscr A_i} \sum_{z \in \Omega_i} \abs{A_z - \widehat A_z}
        + \frac{\abs{\widehat{\mathscr A}_i-\mathscr A_i}}{2\mathscr A_i} \\
        &\le \frac{r_i\zeta}{\mathscr A_i}
        \le \frac{r_i\zeta}{\underline a_i}.
    \end{align*}
    Here $\abs{\widehat{\mathscr A}_i-\mathscr A_i}\le r_i\zeta$, and both normalizing constants are positive by the preceding acceptance bound.

    For the expected cost, assume that the shared preprocessing has been completed.
    One trial generates one feasible candidate, makes at most $r_i$ local-marginal-coin calls, and performs local list and rollback operations.
    Therefore a uniform expected cost bound is
    \begin{align*}
        C_i = \widetilde O\ab(U_i + (r_i + 1) T_{\coin}(\zeta)).
    \end{align*}
    Let $J$ be the number of trials and $Y_t$ be the cost at the $t$-th trial.
    Fresh randomness guarantees that
    \begin{align*}
        \E{Y_t \mid J \ge t} \le C_i
    \end{align*}
    and the event $\set{J \ge t}$ only depends on its preceding history.
    Tonelli's theorem immediately provides that
    \begin{align*}
        \E{\Cost(\CoreHeatBath(\sigma, i; \zeta))} &= \sum_{t \ge 1} \Pr{J \ge t} \E{Y_t \mid J \ge t} \\
        &\le C_i \sum_{t \ge 1} \ab(1 - \widehat{\mathscr A}_i)^{t - 1} \le \frac{2C_i}{\underline a_i}.
    \end{align*}
    Hence we conclude the cost bound.
\end{proof}

\subsection{Design and analysis of the full sampler}\label{subsec:full-sampler-input}
We take the requested tolerance $\eps\in(0,1/2]$ in binary fractional notation, using the terminating expansion for dyadic inputs. Read its digits through the first $1$, at position $b$ after the binary point, and pass $\eps'=2^{-b}$ to the sampler below. This reads $b$ digits and gives $\eps/2<\eps'\le\eps$. If instead the input is a binary numerator and denominator, the same $b$ can be found from their bit lengths and one shifted integer comparison, with an additional $\widetilde O(L_\eps)$ conversion cost for an encoding of length $L_\eps$.
\begin{definition}[The main sampler] \label{def:main-sampler}
    Given a $k$-SAT formula $\Phi$ with $n$ variables and $m$ clauses, and a dyadic error tolerance $\eps=2^{-b}$ for an integer $b\ge1$, define $\FastSampler(\Phi, \eps)$ as follows. For implementation, evaluate the parameters in \eqref{eq:parameter-choice} at the rational input density $m/n$; in particular, $\lambda=1+m/n$.
    \begin{enumerate}
        \item If $m = 0$, output $n$ independent fair bits.
        Otherwise, compute $\lambda$ and $\Delta$ as \eqref{eq:parameter-choice} and generate $\V^{\high}$, $\V^{\bad}$, $\V^{\good}$, bad components $B_1, \ldots, B_s$ and the good formula $\Psi$.
        If $s = 0$, call $\GoodSampler(\C^{\good}, \eps/2)$ and output the result.

        \item Check $\abs{B_i} \le \log_2 n$ and whether $\Omega_i \neq \emptyset$ for every $i = 1, \ldots, s$.
        If the check fails, output $\FAIL$.

        \item For every $i = 1, \ldots, s$, find $\sigma_i$ by rejection sampling on $\set{0, 1}^{B_i}$ and initialize $D_\sigma$.
        
        \item Set $r_{\max} = \max_{i \in [s]}\; r_i$ and
        \begin{align*}
            K^\good = \max\;\set{1, \left\lceil \log_R \frac{256\rho(m+1)}{\eps}\right\rceil}.
        \end{align*}
        If $r_{\max} = 0$, skip the chain and go directly to the final good-variable call.
        Otherwise, compute the integer mixing factor and chain length
        \begin{align*}
            \gamma_0 &= p_* = 2^{-k+12}, \qquad
            M_n = 2^{1+\ceil{\gamma_0\ceil{\log_2 n}}}, \\
            T^\bad &= M_n s \ab(n + 1 + \left\lceil \log_2\frac{4}{\eps}\right\rceil).
        \end{align*}
        The exponent $\gamma_{\mathrm{mix}}$ in \Cref{thm:core-spectral-gap} satisfies $\gamma_{\mathrm{mix}}\le\gamma_0$ in our parameter range. In particular,
        \begin{align}\label{eq:full-sampler-mixing-factor}
            2n^{\gamma_{\mathrm{mix}}}\le 2n^{\gamma_0}\le M_n\le 8n^{\gamma_0}.
        \end{align}
        The nested ceilings make $M_n$ computable by integer arithmetic, without evaluating an irrational power of $n$.
        Moreover, compute
        \begin{align*}
            \overline x = 3p_*, \qquad \underline b = (1-\overline x)^{r_{\max}}, \qquad
            \zeta_0 = \frac{\eps \underline b}{4 T^\bad r_{\max}}.
        \end{align*}
        These quantities are rational, and $x=\e p_*<\overline x$ implies $\underline b\le\underline a\defeq(1-x)^{r_{\max}}$.
        Choose $\zeta$ to be the largest power of two not exceeding $\zeta_0$, so that $\zeta_0/2<\zeta\le\zeta_0$, and pick the cutoff
        \begin{align*}
            K_\coin \defeq \max\;\set{1, \left\lceil \log_R \frac{8\rho}{\zeta}\right\rceil}.
        \end{align*}

        \item Repeat the following runs for exactly $T^\bad$ times: choose $i$ uniformly from $\set{1, \ldots, s}$ and execute $\CoreHeatBath(\sigma, i; \zeta)$ using cutoff $K_\coin$.
        
        \item Draw $\tau \gets \GoodSampler(D_\sigma, \eps/2)$ with fresh randomness and cutoff $K^{\good}$.
        Return $(\sigma, \tau)$ as the final output.
    \end{enumerate}
\end{definition}

On the structural event and the polymer-root estimate, we provide the following accuracy guarantee for $\FastSampler(\Phi, \eps)$.
\begin{proposition} \label{prop:fast-sampler-accuracy}
    For every nice formula $\Phi$ and every dyadic tolerance $\eps \in (0, 1/2]$ with all cutoffs used by the algorithm at most $r$, $\FastSampler(\Phi, \eps)$ terminates almost surely and outputs a satisfying assignment of $\Phi$ from $\mu_\Phi$ within total variation error at most $\eps$.
\end{proposition}
\begin{proof}
    If $m=0$, the output is exactly uniform. If $s=0$, the assertion follows from \Cref{lem:good-variable-sampler}.
    If $s\ge1$ but $r_{\max}=0$, no clause contains both a core variable and a good variable. The formula factors into the core components and the good formula, so the independent uniform initializations of the $\Omega_i$ give the exact core marginal; the final good-variable call contributes at most $\eps/2$ error.
    We henceforth consider $s\ge1$ and $r_{\max}>0$.
    The precision choice gives $r_i\zeta\le\eps\underline b/(4T^\bad)\le\underline a_i/2$ for every $i$.
    All feasible proposals and local coins terminate almost surely.
    By \Cref{prop:core-heat-bath-update-accuracy-cost}, each selected block has trial acceptance probability at least $\underline a/2$.
    Therefore every completed update terminates almost surely.
    There are finitely many core heat-bath updates, and the final good-variable sampler terminates almost surely by \Cref{lem:good-variable-sampler}.
    Hence $\FastSampler(\Phi, \eps)$ terminates almost surely.

    To bound the accuracy in the remaining branch, define the actual executable kernel as
    \begin{align*}
        \widehat \K = \frac{1}{s} \sum_{i = 1}^s \widehat P_i
    \end{align*}
    where $\widehat P_i(\sigma, \cdot)$ is implemented using $\CoreHeatBath(\sigma, i; \zeta)$ with cutoff $K_\coin$.
    By \Cref{prop:core-heat-bath-update-accuracy-cost},
    \begin{align*}
        \TV{\widehat P_i(\sigma, \cdot)}{P_i(\sigma, \cdot)} \le \frac{r_i \zeta}{\underline a}.
    \end{align*}
    Starting from a common initial assignment $\sigma_0$, let $\sigma_t$ be the assignment after $t$ actual executable updates.
    Then it holds that
    \begin{align*}
        \TV{\K^t(\sigma_0, \cdot)}{\Law(\sigma_t \mid \sigma_0)} \le t \sup_{\sigma}\; \frac{1}{s} \sum_{i = 1}^s \TV{\widehat P_i(\sigma, \cdot)}{P_i(\sigma, \cdot)} \le \frac{t r_{\max} \zeta}{\underline a}.
    \end{align*}
    At $t=T^\bad$, the accumulated simulation error is at most $\eps\underline b/(4\underline a)\le\eps/4$. By \Cref{thm:core-spectral-gap} and \eqref{eq:full-sampler-mixing-factor}, the exact chain satisfies
    \begin{align*}
        \TV{\K^{T^\bad}(\sigma_0,\cdot)}{\pi}
        &\le \frac12\exp\ab(\frac{n\ln2}{2}-\frac{T^\bad}{2s n^{\gamma_{\mathrm{mix}}}}) \\
        &\le \frac12\exp\ab(\frac{n\ln2}{2}-n-1-\ceil{\log_2(4/\eps)})
        \le \frac{\eps}{4}.
    \end{align*}
    Combining the two errors gives
    \begin{align*}
        \TV{\pi}{\Law(\sigma_{T^\bad} \mid \sigma_0)} \le \TV{\K^{T^\bad}(\sigma_0, \cdot)}{\Law(\sigma_{T^\bad} \mid \sigma_0)} + \TV{\K^{T^\bad}(\sigma_0, \cdot)}{\pi} \le \frac{\eps}{2}.
    \end{align*}
    For the final good-variable sampler, by \Cref{lem:good-variable-sampler},
    \begin{align*}
        \TV{\Law(\tau \mid \sigma)}{\mu_\Phi(\cdot \mid \sigma)} \le \frac{\eps}{2}.
    \end{align*}
    Putting everything together, we conclude the accuracy bound.
\end{proof}

The expected cost of $\FastSampler$ is given by the following proposition.
\begin{proposition} \label{prop:fast-sampler-cost}
    In the branch of $\FastSampler$ with $s\ge1$ and $r_{\max}>0$, assume that
    \begin{itemize}
        \item $\Cost_\setup$ bounds all one-time preprocessing and precision setup.
        
        \item $U_i$ bounds the expected cost of one uniform feasible proposal for block $B_i$.
        
        \item $\Cost_i(\zeta)$ bounds the expected work of one trial, including local rollback and initial update bookkeeping; hence by \Cref{prop:core-heat-bath-update-accuracy-cost},
        \begin{align*}
            \Cost_i(\zeta) = \widetilde O\ab(U_i + (r_i + 1) T_\coin(\zeta)).
        \end{align*}

        \item $\Cost_\scan$ bounds the expected block-selection and outer-loop overhead per completed update.
        
        \item $\Cost^\good(\delta)$ bounds the expected cost of one good-formula attempt including cleanup after an abort.
    \end{itemize}

    Then it holds that
    \begin{align*}
        \E{\Cost(\FastSampler(\Phi, \eps))} \le \Cost_{\setup} + \sum_{i = 1}^s U_i &+ T^\bad \cdot \Cost_\scan \\
        &+ \frac{2T^\bad}{s} \sum_{i = 1}^s \frac{\Cost_i(\zeta)}{\underline a_i} + \frac{\Cost^\good(\eps/2)}{1 - \eps/4}.
    \end{align*}
\end{proposition}
\begin{proof}
    Fix a formula $\Phi$.
    In the first step, we run a one-time preprocessing and precision setup.
    It costs at most $U_i$ to generate an initial feasible assignment on $B_i$ for every $i = 1, \ldots, s$.
    Let $\mathscr F_{t - 1}$ contain the complete history before the $t$-th core update and let $I_t$ be its uniformly picked block.
    The precision choice ensures that $r_i \zeta \le \underline a_i/2$ and by \Cref{prop:core-heat-bath-update-accuracy-cost},
    \begin{align*}
        \E{\text{the }t\text{-th cost} \mid \mathscr F_{t - 1}, I_t = i} \le \frac{2\Cost_i(\zeta)}{\underline a_i}
    \end{align*}
    and it takes at most $\Cost_\scan$ to pick $I_t$.
    Summing over $t = 1, \ldots, T^\bad$ and taking expectation, we obtain that
    \begin{align*}
        \E{\text{core-update cost}} \le T^\bad \cdot \Cost_\scan + \frac{2T^\bad}{s} \sum_{i = 1}^s \frac{\Cost_i(\zeta)}{\underline a_i}.
    \end{align*}
    Finally, by \Cref{prop:history-domination} and the cutoff in \Cref{lem:good-variable-sampler}, a final good-variable attempt aborts with probability at most $\eps/4$. With independent restarts, its total expected cost is at most
    \begin{align*}
        \Cost^{\good}(\eps/2) \sum_{j \ge 1} (\eps/4)^{j - 1} = \frac{\Cost^{\good}(\eps/2)}{1 - \eps/4}.
    \end{align*}
    Putting everything together, we conclude the total expected cost bound.
\end{proof}

\begin{lemma}[Implementation cost and a uniform exponent] \label{lem:full-sampler-implementation-cost}
    Fix $B>0$ and a nice formula with $k\ge2^{20}$, and let $n^{-B}\le\eps\le1/2$ have the form $\eps=2^{-b}$ for an integer $b\ge1$. Set
    \begin{align*}
        c_k &= 1-\frac{25}{3k^3}, \qquad
        \xi = \frac{4\overline x}{c_k(1-\overline x)}, \qquad \overline x=3p_*, \\
        \kappa_0 &= \frac{2^{-k+4}}{1-2^{-k+4}}
            \frac{2k^2}{6k^7c_k(k-4)}, \\
        \gamma_0 &= 2^{-k+12}, \qquad \eta=\gamma_0+\kappa_0+\xi.
    \end{align*}
    Then $\eta$ depends only on $k$ and satisfies $0<\eta\le2^{16-k}$.
    For sufficiently large $n$, all cutoffs used by $\FastSampler$ are at most $r$, and there is a polynomial $P(k,\lambda)$ of universal degree such that
    \begin{align*}
        \E{\Cost(\FastSampler(\Phi,\eps))}
        \le \widetilde O_B\ab(P(k,\lambda)n^{2+\eta}).
    \end{align*}
\end{lemma}
\begin{proof}
    We first bound the rejection factors. Every clause containing bad variables meets exactly one core component, since those variables form a clique in the original primal graph. Consequently,
    \begin{align} \label{eq:full-sampler-block-sums}
        s\le n, \qquad \sum_{i=1}^s r_i\le m, \qquad
        \sum_{i=1}^s c_i\le m, \qquad \sum_{i=1}^s\abs{B_i}\le n.
    \end{align}
    By \Cref{lem:core-feasible-assignment},
    \begin{align*}
        r_i\le I(B_i)\le\frac{4\ln n}{c_k}.
    \end{align*}
    Thus, in the branch with $s\ge1$ and $r_{\max}>0$,
    \begin{align} \label{eq:full-sampler-rejection-factor}
        \underline a_i^{-1}\le\underline a^{-1}\le\underline b^{-1}
        =\exp\ab(r_{\max}[-\ln(1-\overline x)])
        \le n^{4\overline x/[c_k(1-\overline x)]}=n^\xi.
    \end{align}
    Here we used $-\ln(1-u)\le u/(1-u)$ for $0\le u<1$.
    The exponent $\kappa$ in \Cref{lem:core-feasible-assignment} satisfies $\kappa\le\kappa_0$, since $\lambda\ge1$.
    For $k\ge2^{20}$, we have $c_k\ge99/100$, $\overline x\le1/100$, and $2k^2/[6k^7c_k(k-4)]\le1$. Hence
    \begin{align*}
        \xi\le\frac{12p_*}{(99/100)^2}<13p_*, \qquad
        \kappa_0\le2^{5-k}, \qquad
        0<\eta<(14\cdot2^{12}+2^5)2^{-k}<2^{16-k}.
    \end{align*}

    Next consider precision setup in the branch $s\ge1$, $r_{\max}>0$.
    Since $\eps\ge n^{-B}$, \eqref{eq:full-sampler-mixing-factor} gives $T^\bad=O_B(sn^{1+\gamma_0})\le O_B(n^{2+\gamma_0})$, and $r_{\max}=O(\ln n)$. By \eqref{eq:full-sampler-rejection-factor} and the dyadic rounding rule,
    \begin{align*}
        \zeta\ge\frac{\eps\underline b}{8T^\bad r_{\max}}
        \ge\frac{n^{-(B+2+\gamma_0+\xi)}}{C_B\ln n}
    \end{align*}
    for a constant $C_B>0$. Since $\gamma_0+\xi<1$, this gives an inverse-polynomial lower bound on $\zeta$ with an exponent bounded independently of $k$ and $\lambda$.
    Since $R\ge4$ and $\rho<1$, the cutoffs satisfy
    \begin{align*}
        K_\coin=O_B(\ln n), \qquad K^\good=O_B(\ln n+\ln(m+1)).
    \end{align*}
    For fixed width and density, these are $O_{B,k,\lambda}(\ln n)$, whereas $r=\floor{q_0n}$ grows linearly in $n$; hence both cutoffs are admissible for all sufficiently large $n$.
    The same conclusion for the good-variable cutoff applies when the core chain is skipped. The final residual tolerance $\delta=\eps/2=2^{-(b+1)}$ has $O_B(1+\log n)$ bits and satisfies $\delta\ge n^{-(B+1)}$ for $n\ge2$. Thus the $L_\delta$ term of \Cref{lem:good-variable-sampler} is absorbed in the setup bound.
    The rational number $\underline b$ has $O(kr_{\max})=O(k\ln n)$ bits. To compute $M_n$, first obtain $\ceil{\log_2 n}$ from the binary representation of $n$, then compute the ceiling of its product with $\gamma_0=2^{-(k-12)}$ by integer division. The resulting $M_n$ has $O(1+\gamma_0\ln n)$ bits. All other precision parameters have bit length polynomial in $k$, $B$, $\ln n$, and $\ln(m+1)$.
    The dyadic rounding and cutoff comparisons can therefore be performed by rational arithmetic in $\widetilde O_B(\poly{k,\lambda})$ bit operations. This accounts for precision setup only; the formula preprocessing is counted next.

    Read the $km$ literal occurrences and build variable incidence lists. A queue of clauses whose number of distinct bad variables exceeds eight constructs the closure $\V^{\bad}$, processing each incidence only a bounded number of times. Joining bad variables that share a clause identifies the components $B_i$ and the lists of their core and mixed clauses.
    Constructing the residual dependency lists costs $\widetilde O(n+km+m\overline D)$: every good variable has degree below $\Delta$, so enumerating the pairs of clauses sharing good variables costs at most $km\Delta\le m\overline D$.
    To implement the nonemptiness test in \Cref{def:main-sampler}, first check the component sizes, then enumerate assignments within each component until a feasible one is found or all have been checked. Its total cost is at most
    \begin{align*}
        \widetilde O\ab(\sum_{i=1}^s2^{\abs{B_i}}(\abs{B_i}+kc_i+1))
        \le\widetilde O\ab(n(2n+km)).
    \end{align*}
    This is a one-time test; every subsequent uniform feasible proposal uses the rejection sampler of \Cref{lem:core-feasible-assignment}.
    Initializing $D_\sigma$ costs $\widetilde O(km)$.
    Together with precision setup and $m\le\lambda n$, $\overline D\le6k^8\lambda+2$, these bounds give
    \begin{align} \label{eq:full-sampler-setup-cost}
        \Cost_\setup\le\widetilde O_B\ab(P(k,\lambda)n^2)
    \end{align}
    for a polynomial $P$ of universal degree, enlarged below when necessary.

    For the local operations, store clause families as persistent membership dictionaries over the shared dependency graph. Membership tests and single-clause insertions or deletions take polylogarithmic time, and recursive calls share unchanged data.
    In particular, $H_{\sigma,i}=D_\sigma\setminus\J_i$ can be formed by $r_i$ deletions. A rejected trial restores the saved version, and an accepted trial commits its final version. Constructing $L_{i,\sigma}(z)$ scans only the $r_i$ incident mixed clauses, in $\widetilde O(kr_i)$ time.
    Recursive proposals likewise modify only the local neighborhoods they expose; their dictionary and cleanup work is absorbed into the local cost bounds of \Cref{lem:connected-subset-proposal,prop:executable-local-marginal-coin}.
    Thus no update requires copying or scanning the entire residual formula.
    Uniform block selection by rejection from $\ceil{\log_2 s}$ fair bits and the outer-loop counter cost $\widetilde O_B(1)$ per update.

    The explicit candidate-cost bound in \Cref{lem:core-feasible-assignment}, summed using \eqref{eq:full-sampler-block-sums}, and the local-coin and one-attempt bounds in \Cref{prop:executable-local-marginal-coin,prop:history-domination}, now give
    \begin{align*}
        \sum_{i=1}^s U_i &\le\widetilde O\ab((2n+km)n^{\kappa_0})
            \le\widetilde O\ab(P(k,\lambda)n^{1+\kappa_0}), \\
        T_\coin(\zeta)&\le\widetilde O_B\ab(P(k,\lambda)), \qquad
        \Cost_\scan\le\widetilde O_B(1), \\
        \Cost^\good(\eps/2)&\le\widetilde O_B\ab(n+m\poly{k,\overline D})
            \le\widetilde O_B\ab(P(k,\lambda)n).
    \end{align*}
    The last bound is for one attempt, including cleanup, as required in \Cref{prop:fast-sampler-cost}; the restart factor is accounted for separately there.
    Also, $\sum_i(r_i+1)\le m+n$, so
    \begin{align*}
        \sum_{i=1}^s\Cost_i(\zeta)
        &\le\widetilde O_B\ab(\sum_i U_i+(m+n)T_\coin(\zeta))
        \le\widetilde O_B\ab(P(k,\lambda)n^{1+\kappa_0}).
    \end{align*}
    Since $T^\bad/s=O_B(n^{1+\gamma_0})$, \eqref{eq:full-sampler-rejection-factor} yields
    \begin{align*}
        \frac{2T^\bad}{s}\sum_{i=1}^s\frac{\Cost_i(\zeta)}{\underline a_i}
        \le\widetilde O_B\ab(P(k,\lambda)n^{2+\gamma_0+\kappa_0+\xi}).
    \end{align*}
    Substituting into \Cref{prop:fast-sampler-cost}, the other terms cost at most
    $\widetilde O_B(P(k,\lambda)(n^2+n^{1+\kappa_0}+n^{2+\gamma_0}+n))$,
    because $(1-\eps/4)^{-1}\le8/7$. This proves the claimed bound with $\eta=\gamma_0+\kappa_0+\xi$.
    If $m=0$, the cost is $O(n)$. If $s=0$ or $r_{\max}=0$, the chain is omitted, and the same preprocessing, initialization, and good-variable bounds suffice.
\end{proof}

\begin{proof}[Proof of \Cref{thm:fast-sampler}]
    Fix $k$, $\alpha$, and $B$ as in the theorem, and put $m=\floor{\alpha n}$.
    Apply the structural estimates with the input density $m/n\le\alpha$ used by \Cref{def:main-sampler}.
    Their explicit failure bounds, combined as in \Cref{lem:random-k-SAT-nice}, still tend to zero: $n_0$ is bounded above by its value at $\alpha$, $q_0$ is bounded below by its value at $\alpha$, and the root-event cutoff is $O_{k,\alpha}(\ln n)$.
    Thus one event of probability $1-o_n(1)$ makes $\Phi$ nice and supports all subsequent calls.

    For any requested accuracy $n^{-B}\le\eps\le1/2$, the binary input convention above gives $b=\ceil{\log_2(1/\eps)}=O_B(1+\log n)$ and $\eps'=2^{-b}$ with $\eps/2<\eps'\le\eps$ after $O(b)$ bit operations. Run $\FastSampler(\Phi,\eps')$.
    This tolerance satisfies $\eps'\ge n^{-(B+1)}$ for $n\ge2$.
    By \Cref{lem:full-sampler-implementation-cost} applied with $B+1$, all cutoffs are admissible for sufficiently large $n$, uniformly over these accuracies.
    Proposition~\ref{prop:fast-sampler-accuracy} then guarantees almost-sure termination and an output satisfying $\Phi$ with total variation error at most $\eps'\le\eps$.
    The implementation-cost lemma, together with $\lambda=1+m/n\le1+\alpha$, gives
    \begin{align*}
        \E{\Cost(\FastSampler(\Phi,\eps'))}
        \le\widetilde O_B\ab(\poly{k,1+\alpha}n^{2+\eta}),
        \qquad 0<\eta\le2^{16-k}.
    \end{align*}
    The polynomial degree is universal, as asserted.
\end{proof}

\newpage
\bibliographystyle{alpha}
\bibliography{refs}

\newpage
\appendix
\section{Supplementary Proofs}
\label{sec:appendix}

\subsection{
  \texorpdfstring
    {Missing proofs in \Cref{sec:core-fast-sampling}}
    {Missing proofs in the bad-core mixing section}
}

\begin{proof}[Proof of \Cref{lem-pisigma-expansion}]
 For every $\sigma\in\Omega$, we have
\begin{equation}
 e^{a_eY_e\left(\sigma_{V_e}\right)}
 =c_e\bigl(1+w_eh_e\left(\sigma_{V_e}\right)\bigr),
 \qquad
 c_e=\begin{cases}1,&a_e>0,\\ e^{a_e},&a_e<0.
       \end{cases}
 \label{eq:core-positive-factorization}
\end{equation}
To verify \eqref{eq:core-positive-factorization}, fix $e\in\mathcal J$
and $\sigma\in\Omega$. If $a_e>0$, then $c_e=1$,
$w_e=e^{a_e}-1$, and $h_e\left(\sigma_{V_e}\right)=Y_e\left(\sigma_{V_e}\right)$, so
\[
 c_e\bigl(1+w_eh_e\left(\sigma_{V_e}\right)\bigr)
 =1+\left(e^{a_e}-1\right)Y_e\left(\sigma_{V_e}\right)
 =\begin{cases}
   1,&Y_e\left(\sigma_{V_e}\right)=0,\\
   e^{a_e},&Y_e\left(\sigma_{V_e}\right)=1.
  \end{cases}
\]
If $a_e<0$, then $c_e=e^{a_e}$, $w_e=e^{-a_e}-1$, and
$h_e\left(\sigma_{V_e}\right)=1-Y_e\left(\sigma_{V_e}\right)$, so
\[
\begin{aligned}
 c_e\bigl(1+w_eh_e\left(\sigma_{V_e}\right)\bigr)
 &=e^{a_e}\bigl(1+\left(e^{-a_e}-1\right)\left(1-Y_e\left(\sigma_{V_e}\right)\right)\bigr)\\
 &=1+\left(e^{a_e}-1\right)Y_e\left(\sigma_{V_e}\right)
 =\begin{cases}
   1,&Y_e\left(\sigma_{V_e}\right)=0,\\
   e^{a_e},&Y_e\left(\sigma_{V_e}\right)=1.
  \end{cases}
\end{aligned}
\]
Since $Y_e\left(\sigma_{V_e}\right)\in\{0,1\}$, the displayed values equal
$e^{a_eY_e\left(\sigma_{V_e}\right)}$ in every case, proving
\eqref{eq:core-positive-factorization}.

By \eqref{eq:core-positive-factorization},
\begin{equation}
\begin{aligned}
\forall \sigma\in\Omega,\qquad \pi(\sigma)
 &\propto\left(\prod_{i=1}^s\ell_i\left(\sigma_i\right)\right)
          \left(\prod_{e\in\mathcal J}c_e\right)
          \prod_{e\in\mathcal J}\bigl(1+w_eh_e\left(\sigma_{V_e}\right)\bigr)\\
 &\propto\left(\prod_{i=1}^s\ell_i\left(\sigma_i\right)\right)
          \prod_{e\in\mathcal J}\bigl(1+w_eh_e\left(\sigma_{V_e}\right)\bigr).
\end{aligned}
\end{equation}
The second proportionality uses that $\prod_{e\in\mathcal J}c_e$ is
positive and independent of $\sigma$.
The lemma is proved.
\end{proof}

\begin{proof}[Proof of \Cref{lem-pisigmatau-expansion}]
Summing the full joint weight in \eqref{eq:core-positive-lift} over
$\eta$ gives, for every $\sigma\in\Omega$,
\begin{align}
 \sum_{\eta\in\{0,1\}^{\mathcal J}}\widehat\pi(\sigma,\eta)
 &\propto
 \left(\prod_{i=1}^s\ell_i\left(\sigma_i\right)\right)
 \sum_{\eta\in\{0,1\}^{\mathcal J}}
 \prod_{e:\eta_e=1}w_eh_e\left(\sigma_{V_e}\right)\notag\\
 &=\left(\prod_{i=1}^s\ell_i\left(\sigma_i\right)\right)
 \prod_{e\in\mathcal J}
 \left(1+w_eh_e\left(\sigma_{V_e}\right)\right)\propto\pi(\sigma),
 \label{eq:core-positive-lift-marginal}
\end{align}
where the last proportionality follows from
\Cref{lem-pisigma-expansion}. The proportionality constants are
independent of $\sigma$; normalizing over $\sigma$ therefore shows that
the $\sigma$-marginal of $\widehat\pi$ is exactly $\pi$, namely,
\begin{equation}
 \sum_{\eta\in\{0,1\}^{\mathcal J}}
 \widehat\pi(\sigma,\eta)=\pi(\sigma),
 \qquad \sigma\in\Omega.
\end{equation}
Moreover, for fixed $\sigma$,
\begin{equation}
\begin{aligned}
 \widehat\pi\left(\eta\mid X=\sigma\right)
 &=\frac{
       \displaystyle
       \left(\prod_{i=1}^s\ell_i\left(\sigma_i\right)\right)
       \prod_{e:\eta_e=1}w_eh_e\left(\sigma_{V_e}\right)}{
       \displaystyle\sum_{\eta'\in\{0,1\}^{\mathcal J}}
       \left(\prod_{i=1}^s\ell_i\left(\sigma_i\right)\right)
       \prod_{e:\eta'_e=1}w_eh_e\left(\sigma_{V_e}\right)}=\frac{
       \displaystyle\prod_{e:\eta_e=1}
       w_eh_e\left(\sigma_{V_e}\right)}{
       \displaystyle\sum_{\eta'\in\{0,1\}^{\mathcal J}}
       \prod_{e:\eta'_e=1}w_eh_e\left(\sigma_{V_e}\right)}\\
 &=\prod_{e:\eta_e=1}
   \frac{w_eh_e\left(\sigma_{V_e}\right)}
        {1+w_eh_e\left(\sigma_{V_e}\right)}
  \prod_{e:\eta_e=0}
   \frac{1}{1+w_eh_e\left(\sigma_{V_e}\right)}.
\end{aligned}
\end{equation}
The lemma is proved.
\end{proof}

\begin{proof}[Proof of \Cref{lem:core-random-component-flip}]
For any fixed set $D\subseteq[s]$, define the involution $T_D$ on
$\left(\Omega\times\{0,1\}^{\mathcal J}\right)^2$ by
\begin{equation}
 T_D\bigl((\sigma,\eta),(\omega,\xi)\bigr)
 =\bigl(\left(\widetilde\sigma,\widetilde\eta\right),
        \left(\widetilde\omega,\widetilde\xi\right)\bigr),
 \label{eq:core-swap-map}
\end{equation}
where, for every $i\in[s]$ and $e\in\mathcal J$,
\[
 \left(\widetilde\sigma_i,\widetilde\omega_i\right)=
 \begin{cases}
  \left(\omega_i,\sigma_i\right),&i\in D,\\
  \left(\sigma_i,\omega_i\right),&i\notin D,
 \end{cases}
 \qquad
 \left(\widetilde\eta_e,\widetilde\xi_e\right)=
 \begin{cases}
  \left(\xi_e,\eta_e\right),&V_e\subseteq D,\\
  \left(\eta_e,\xi_e\right),&V_e\not\subseteq D.
 \end{cases}
\]
Whenever we apply $T_D$ below to an input pair
$\bigl((X,\eta),(X',\xi)\bigr)$, we assume that $D$ is a union of
connected components of its associated hypergraph $g$, where $e$ is
present exactly when $\eta_e\vee\xi_e=1$. 
Formally,
\begin{equation}
 D=\bigcup_{C\in\mathcal D}C\subseteq[s] \qquad \text{ for some }
\mathcal D\subseteq\operatorname{Comp}(g).
 \label{eq:core-swap-region}
\end{equation}
The subcollection $\mathcal D$ may be empty or equal to
$\operatorname{Comp}(g)$.

Fix an input pair $\bigl((X,\eta),(X',\xi)\bigr)$, and let $g$ be its associated hypergraph and $D$ be a union of components of $g$.
Fix $e\in\mathcal J$ with $\eta_e\vee\xi_e=1$. The hyperedge labelled
$e$ belongs to $g$, so any two vertices of $V_e$ are connected by this
hyperedge. Hence $V_e\subseteq C_e$ for a unique
$C_e\in\operatorname{Comp}(g)$. If $C_e\in\mathcal D$, then
\eqref{eq:core-swap-region} gives $V_e\subseteq C_e\subseteq D$.
If $C_e\notin\mathcal D$, then $C_e\cap C=\emptyset$ for every
$C\in\mathcal D$, because distinct connected components have disjoint
vertex sets. Thus \eqref{eq:core-swap-region} gives
\[
 C_e\cap D
 =C_e\cap\bigcup_{C\in\mathcal D}C
 =\bigcup_{C\in\mathcal D}\left(C_e\cap C\right)
 =\emptyset,
\]
and therefore $V_e\cap D=\emptyset$. We have proved
\begin{equation}
 \forall e\in\mathcal J,\qquad
 \eta_e\vee\xi_e=1
 \ \Longrightarrow\
 V_e\subseteq D\ \text{or}\ V_e\cap D=\emptyset.
 \label{eq:core-open-edge-region}
\end{equation}
Write
\[
 \left(\left(\widetilde X,\widetilde\eta\right),
       \left(\widetilde X',\widetilde\xi\right)\right)
 =T_D\bigl((X,\eta),\left(X',\xi\right)\bigr).
\]
For each $e\in\mathcal J$, define
\[
 F_e(\sigma,\eta,\omega,\xi)
 =\bigl(w_eh_e\left(\sigma_{V_e}\right)\bigr)^{\eta_e}
  \bigl(w_eh_e\left(\omega_{V_e}\right)\bigr)^{\xi_e},
 \qquad
 \sigma,\omega\in\Omega,\quad\eta,\xi\in\{0,1\}^{\mathcal J},
\]
where a factor with exponent zero equals one, including $0^0=1$.
We verify the invariance of these factors in three cases.
\begin{enumerate}
 \item If $V_e\subseteq D$, the definition in
 \eqref{eq:core-swap-map} gives
 \[
 \left(\widetilde X_{V_e},\widetilde X'_{V_e},
  \widetilde\eta_e,\widetilde\xi_e\right)
 =\left(X'_{V_e},X_{V_e},\xi_e,\eta_e\right).
 \]
 Consequently,
 \[
 \begin{aligned}
 F_e\left(\widetilde X,\widetilde\eta,\widetilde X',\widetilde\xi\right)
 &=\left(w_eh_e\left(\widetilde X_{V_e}\right)\right)^{\widetilde\eta_e}
   \left(w_eh_e\left(\widetilde X'_{V_e}\right)\right)^{\widetilde\xi_e}\\
 &=\bigl(w_eh_e\left(X'_{V_e}\right)\bigr)^{\xi_e}
   \bigl(w_eh_e\left(X_{V_e}\right)\bigr)^{\eta_e}\\
 &=\bigl(w_eh_e\left(X_{V_e}\right)\bigr)^{\eta_e}
   \bigl(w_eh_e\left(X'_{V_e}\right)\bigr)^{\xi_e}
 =F_e\left(X,\eta,X',\xi\right).
 \end{aligned}
 \]
 \item If $V_e\cap D=\emptyset$, then $V_e\not\subseteq D$ because
 $V_e\ne\emptyset$. Hence \eqref{eq:core-swap-map} gives
 \[
 \left(\widetilde X_{V_e},\widetilde X'_{V_e},
  \widetilde\eta_e,\widetilde\xi_e\right)
 =\left(X_{V_e},X'_{V_e},\eta_e,\xi_e\right).
 \]
 Then $F_e\left(\widetilde X,\widetilde\eta,\widetilde X',\widetilde\xi\right)
 =F_e\left(X,\eta,X',\xi\right)$.
 \item If $V_e\cap D\ne\emptyset$ and $V_e\not\subseteq D$, then
 \eqref{eq:core-open-edge-region} implies $\eta_e\vee\xi_e=0$.
 Since $\eta_e,\xi_e\in\{0,1\}$, we have $\eta_e = \xi_e = 0$.
 Then \eqref{eq:core-swap-map} yields
 \[
 \widetilde\eta_e=\eta_e=0,\qquad
 \widetilde\xi_e=\xi_e=0.
 \]
 Therefore the zero-exponent convention gives
 \[
 F_e\left(\widetilde X,\widetilde\eta,\widetilde X',\widetilde\xi\right)
 =1=F_e\left(X,\eta,X',\xi\right).
 \]
\end{enumerate}
These cases exhaust all possibilities, so
\begin{equation}
 \forall e\in\mathcal J,\qquad
 F_e\left(\widetilde X,\widetilde\eta,\widetilde X',\widetilde\xi\right)
 =F_e\left(X,\eta,X',\xi\right).
 \label{eq:core-swap-edge-factor}
\end{equation}
For every $i\in[s]$, \eqref{eq:core-swap-map} also gives
\[
 \begin{aligned}
 \ell_i\left(\widetilde X_i\right)\ell_i\left(\widetilde X'_i\right)
 &=\begin{cases}
   \ell_i\left(X'_i\right)\ell_i\left(X_i\right),&i\in D,\\
   \ell_i\left(X_i\right)\ell_i\left(X'_i\right),&i\notin D,
  \end{cases}\\
 &=\ell_i\left(X_i\right)\ell_i\left(X'_i\right).
 \end{aligned}
\]
Combined with \eqref{eq:core-swap-edge-factor} and \eqref{eq:core-positive-lift}, we obtain
\begin{equation}
 \widehat\pi\left(\widetilde X,\widetilde\eta\right)
 \widehat\pi\left(\widetilde X',\widetilde\xi\right)
 =\widehat\pi(X,\eta)\widehat\pi\left(X',\xi\right).
 \label{eq:core-lifted-pair-swap-invariance}
\end{equation}
Moreover, for every $e\in\mathcal J$,
\[
 \begin{aligned}
 \widetilde\eta_e\vee\widetilde\xi_e
 &=\begin{cases}
   \xi_e\vee\eta_e,&V_e\subseteq D,\\
   \eta_e\vee\xi_e,&V_e\not\subseteq D,
  \end{cases}\\
 &=\eta_e\vee\xi_e.
 \end{aligned}
\]
Thus,
\[
 \{e\in\mathcal J:\widetilde\eta_e\vee\widetilde\xi_e=1\}
 =\{e\in\mathcal J:\eta_e\vee\xi_e=1\}.
\]
Finally, for every $\sigma,\omega\in\Omega$ and
$\eta,\xi\in\{0,1\}^{\mathcal J}$, applying
\eqref{eq:core-swap-map} twice gives
\[
 T_D\Bigl(T_D\bigl((\sigma,\eta),(\omega,\xi)\bigr)\Bigr)
 =\bigl((\sigma,\eta),(\omega,\xi)\bigr).
\]
Indeed, each selected coordinate and auxiliary bit is swapped twice,
returning to its original value, while all others remain unchanged.
Moreover, the edge indicators $\eta_e\vee\xi_e$ are unchanged by
$T_D$ for every input pair and every fixed $D$. Thus, for each fixed
$g$ and each $D$ that is a union of components of $g$, the restriction
of $T_D$ to the fiber of pairs whose hypergraph is $g$ is a bijection
from that fiber to itself. By
\eqref{eq:core-lifted-pair-swap-invariance}, this restricted bijection
preserves the pair probabilities.

For each $C\in\operatorname{Comp}(g)$, independently draw a fresh bit
$\varepsilon_C\in\{0,1\}$ with
\[
 \mathbb P_\varepsilon\left(\varepsilon_C=0\right)
 =\mathbb P_\varepsilon\left(\varepsilon_C=1\right)=\frac12.
\]
Collect these bits into the vector
$\varepsilon=\left(\varepsilon_C\right)_{C\in\operatorname{Comp}(g)}
\in\{0,1\}^{\operatorname{Comp}(g)}$ and define its complementary bit
vector by
\[
 \bar\varepsilon=\left(1-\varepsilon_C\right)_{C\in\operatorname{Comp}(g)}
 \in\{0,1\}^{\operatorname{Comp}(g)}.
\]
Define
\[
 D(\varepsilon)=
 \bigcup_{\substack{C\in\operatorname{Comp}(g)\\\varepsilon_C=1}}C
\]
and apply $T_{D(\varepsilon)}$ to this sample pair.
By definition, $D(\varepsilon)$ is a union of connected components of the
associated hypergraph $g$, so it satisfies the condition in \eqref{eq:core-swap-region}.
Thus a component $C$ is swapped in its entirety if $\varepsilon_C=1$
and is left unchanged if $\varepsilon_C=0$.
Let $X^\varepsilon,\left(X'\right)^\varepsilon$ be the two
resulting configurations.
Since every component is swapped in exactly one of $\varepsilon$ and
$\bar\varepsilon$, we have
$X^{\bar\varepsilon}=\left(X'\right)^\varepsilon$.
The notation $\mathbb E_\varepsilon$ and
$\operatorname{Var}_\varepsilon$ denotes expectation and variance over
$\varepsilon$ alone, with the sample pair
$\bigl((X,\eta),\left(X',\xi\right)\bigr)$ held fixed.
The map $\varepsilon\mapsto\bar\varepsilon$ is a bijection of
$\{0,1\}^{\operatorname{Comp}(g)}$, on which $\varepsilon$ is uniform.
Thus $\varepsilon$ and $\bar\varepsilon$ have the same distribution, as do
$f\left(X^\varepsilon\right)$ and $f\left(X^{\bar\varepsilon}\right)$; in particular, they
have the same mean and variance.
Since $\mathbb E_\varepsilon[f\left(X^\varepsilon\right)]
=\mathbb E_\varepsilon[f\left(X^{\bar\varepsilon}\right)]$, we have
\[
 \begin{aligned}
 f\left(X^\varepsilon\right)-f\left(X^{\bar\varepsilon}\right)
 &=\bigl(f\left(X^\varepsilon\right)
          -\mathbb E_\varepsilon[f\left(X^\varepsilon\right)]\bigr) -\bigl(f\left(X^{\bar\varepsilon}\right)
          -\mathbb E_\varepsilon[f\left(X^{\bar\varepsilon}\right)]\bigr).
 \end{aligned}
\]
Applying $(u-v)^2\le2u^2+2v^2$ to these two mean-subtracted terms
and taking expectation over $\varepsilon$ gives
\begin{equation}
 \begin{aligned}
 \mathbb E_\varepsilon
   \left(f\left(X^\varepsilon\right)-f\left(X^{\bar\varepsilon}\right)\right)^2
 &\le2\mathbb E_\varepsilon\!\left[
   \bigl(f\left(X^\varepsilon\right)-\mathbb E_\varepsilon[f\left(X^\varepsilon\right)]\bigr)^2
   \right] +2\mathbb E_\varepsilon\!\left[
   \bigl(f\left(X^{\bar\varepsilon}\right)
          -\mathbb E_\varepsilon[f\left(X^{\bar\varepsilon}\right)]\bigr)^2
   \right] \\
 &=2\operatorname{Var}_\varepsilon\left(f\left(X^\varepsilon\right)\right)
   +2\operatorname{Var}_\varepsilon\left(f\left(X^{\bar\varepsilon}\right)\right) =4\operatorname{Var}_\varepsilon\left(f\left(X^\varepsilon\right)\right).
 \end{aligned}
 \label{eq:core-swap-antithetic}
\end{equation}
For each $C\in\operatorname{Comp}(g)$, define the bit vector
$\varepsilon^C\in\{0,1\}^{\operatorname{Comp}(g)}$ by
\[
\forall C'\in\operatorname{Comp}(g), \qquad
 \left(\varepsilon^C\right)_{C'}=
 \begin{cases}
  1-\varepsilon_C,&C'=C,\\
  \varepsilon_{C'},&C'\ne C.
 \end{cases}
\]
Thus $\varepsilon^C$ differs from $\varepsilon$ only at the bit indexed
by $C$. This changes whether the entire component $C$ is swapped,
while leaving the choices for all other components unchanged.
Write $\varepsilon_{-C}=\left(\varepsilon_{C'}\right)_{C'\in
\operatorname{Comp}(g)\setminus\{C\}}$ for all the other bits.
Since the bits are independent, applying
Lemma~\ref{lem:efron-stein} in Section~\ref{subsec:efron-stein} to the
function $\varepsilon\mapsto f\left(X^\varepsilon\right)$ gives
\begin{equation}
 \operatorname{Var}_\varepsilon\left(f\left(X^\varepsilon\right)\right)
 \le\sum_{C\in\operatorname{Comp}(g)}
     \mathbb E_\varepsilon\!\left[
       \operatorname{Var}_\varepsilon\left(f\left(X^\varepsilon\right)
                                     \mid\varepsilon_{-C}\right)\right].
 \label{eq:core-swap-conditional-variance}
\end{equation}
For each fixed value of $\varepsilon_{-C}$, let $a$ and $b$ denote the
values of $f\left(X^\varepsilon\right)$ when $\varepsilon_C=0$ and
$\varepsilon_C=1$, respectively. Since each choice of $\varepsilon_C$
has probability $1/2$, the conditional mean of $f\left(X^\varepsilon\right)$ is
$(a+b)/2$, and its conditional variance is
\[
 \begin{aligned}
 \operatorname{Var}_\varepsilon\left(f\left(X^\varepsilon\right)\mid\varepsilon_{-C}\right)
 &=\frac12\left(a-\frac{a+b}{2}\right)^2
   +\frac12\left(b-\frac{a+b}{2}\right)^2 =\frac{(a-b)^2}{4}.
 \end{aligned}
\]
Flipping $\varepsilon_C$ changes the function value from $a$ to $b$
if $\varepsilon_C=0$, and from $b$ to $a$ if $\varepsilon_C=1$.
Thus $\left(f\left(X^\varepsilon\right)-f\left(X^{\varepsilon^C}\right)\right)^2=(a-b)^2$ for either
choice of $\varepsilon_C$. It follows that
\[
 \operatorname{Var}_\varepsilon\left(f\left(X^\varepsilon\right)\mid\varepsilon_{-C}\right)
 =\frac14\mathbb E_\varepsilon\!\left[
   \left(f\left(X^\varepsilon\right)-f\left(X^{\varepsilon^C}\right)\right)^2
   \mid\varepsilon_{-C}\right].
\]
Taking expectations gives
\begin{equation}
 \begin{aligned}
 &\mathbb E_\varepsilon\!\left[
   \operatorname{Var}_\varepsilon\left(
     f\left(X^\varepsilon\right)\mid\varepsilon_{-C}\right)\right] \\
 &\quad=\mathbb E_{\varepsilon_{-C}}\!\left[
   \operatorname{Var}_\varepsilon\left(
     f\left(X^\varepsilon\right)\mid\varepsilon_{-C}\right)\right] \\
 &\quad=\frac14\mathbb E_{\varepsilon_{-C}}\!\left[
   \mathbb E_\varepsilon\!\left[
     \left(f\left(X^\varepsilon\right)
           -f\left(X^{\varepsilon^C}\right)\right)^2
     \mid\varepsilon_{-C}\right]\right] \\
 &\quad=\frac14\mathbb E_\varepsilon\!\left[
     \left(f\left(X^\varepsilon\right)
           -f\left(X^{\varepsilon^C}\right)\right)^2\right].
 \end{aligned}
 \label{eq:core-swap-averaged-bit-variance}
\end{equation}

By \eqref{eq:core-swap-averaged-bit-variance} and \eqref{eq:core-swap-conditional-variance},
\begin{equation}
 \begin{aligned}
 4\operatorname{Var}_\varepsilon\left(f\left(X^\varepsilon\right)\right)
 &\le4\sum_{C\in\operatorname{Comp}(g)}
     \mathbb E_\varepsilon\!\left[
       \operatorname{Var}_\varepsilon\left(
         f\left(X^\varepsilon\right)\mid\varepsilon_{-C}\right)\right]\\
 &=\sum_{C\in\operatorname{Comp}(g)}
     \mathbb E_\varepsilon\!\left[
       \left(f\left(X^\varepsilon\right)
             -f\left(X^{\varepsilon^C}\right)\right)^2\right].
 \end{aligned}
 \label{eq:core-swap-efron-stein}
\end{equation} 

Let $Y=\left((X,\eta),(X',\xi)\right)$, and let $G$ be its associated
labelled hypergraph. For every $g$ with $\mathbb P(G=g)>0$ and every fixed
$\varepsilon\in\{0,1\}^{\operatorname{Comp}(g)}$, the set $D(\varepsilon)$
is a union of components of $g$. Thus $T_{D(\varepsilon)}$ is a bijection
on the fiber of sample pairs whose hypergraph is $g$, and
\eqref{eq:core-lifted-pair-swap-invariance} shows that this map leaves
the product probabilities
$\widehat\pi(X,\eta)\widehat\pi\left(X',\xi\right)$ unchanged. Since
conditioning on $G=g$ amounts to normalizing these probabilities on that
fiber, for every sample-pair value $y$ in the fiber we have
\[
 \mathbb P\left(T_{D(\varepsilon)}(Y)=y\mid G=g\right)
 =\mathbb P\left(Y=y\mid G=g\right).
\]
Conditional on $G=g$, the bit vector $\varepsilon$ is sampled
uniformly from $\{0,1\}^{\operatorname{Comp}(g)}$, independently of
$\left(\left(X,\eta\right),\left(X',\xi\right)\right)$.
Consequently, for every
$\sigma,\omega\in\Omega$,
\begin{equation}
 \begin{aligned}
 &\mathbb P\left(X^\varepsilon=\sigma,
                 \left(X'\right)^\varepsilon=\omega
                 \mid G=g,\varepsilon\right) =\mathbb P\left(X=\sigma,X'=\omega\mid G=g\right).
 \end{aligned}
 \label{eq:core-conditional-swapped-pair-law}
\end{equation}
By the law of total probability and the independence of $X$ and
$X'$, for every $\sigma,\omega\in\Omega$,
\begin{equation}
 \mathbb P\left(X^\varepsilon=\sigma,
                \left(X'\right)^\varepsilon=\omega\right)
 =\mathbb P\left(X=\sigma,X'=\omega\right)
 =\pi(\sigma)\pi(\omega).
 \label{eq:core-swapped-pair-law}
\end{equation}
Using $X^{\bar\varepsilon}=\left(X'\right)^\varepsilon$,
\eqref{eq:core-swapped-pair-law}, and the independence of $X$ and
$X'$ with common law $\pi$, we obtain
\begin{equation}
 \begin{aligned}
 &\mathbb E\!\left[\mathbb E_\varepsilon\!\left[
   \left(f\left(X^\varepsilon\right)
         -f\left(X^{\bar\varepsilon}\right)\right)^2\right]\right] 
 =\mathbb E\!\left[\mathbb E_\varepsilon\!\left[
   \left(f\left(X^\varepsilon\right)
         -f\left(\left(X'\right)^\varepsilon\right)\right)^2\right]\right] \\
 =&\mathbb E\!\left[
   \left(f(X)-f\left(X'\right)\right)^2\right]
   =2\operatorname{Var}_\pi(f).
 \end{aligned}
 \label{eq:core-independent-pair-variance}
\end{equation}
We now average over the independent lifted samples
$(X,\eta),\left(X',\xi\right)\sim\widehat\pi$.
In the following inequality, the outer $\mathbb E$ averages over
$X,\eta,X',\xi$, whereas the inner $\mathbb E_\varepsilon$ averages
only over the swap bits, holding the entire sample pair fixed.
Combining \eqref{eq:core-swap-antithetic} and
\eqref{eq:core-swap-efron-stein}, taking these expectations, and using
\eqref{eq:core-independent-pair-variance} gives
\begin{equation*}
 \operatorname{Var}_\pi(f)
 \le\frac12\mathbb E\!\left[
   \sum_{C\in\operatorname{Comp}(g)}\mathbb E_\varepsilon\!\left[
     \left(f\left(X^\varepsilon\right)
           -f\left(X^{\varepsilon^C}\right)\right)^2\right]\right].
\end{equation*}
\end{proof}

\begin{proof}[Proof of \Cref{lem:core-oscillation-incidence}]
If $J_\Gamma\cap A\ne\emptyset$, then some clause $a\in U_\Gamma$
contains a variable of $W_A$, by \eqref{eq:core-geometric-scope}.
Every label in $\Lambda$ has $|J_\Gamma|\ge2$, so
\eqref{eq:core-multiblock-root} gives
\[
 R(A)\le
 \sum_{\substack{a\in\C^{\good}\\\vbl(a)\cap W_A\ne\emptyset}}
 \sum_{\substack{\Gamma\in\Lambda\\a\in U_\Gamma}}|a_\Gamma|
 \le\rho_0
 \left|\{a\in\C^{\good}:\vbl(a)\cap W_A\ne\emptyset\}\right|
 \le\rho_0 I(W_A).
\]
The first sum may count a cluster once for each incident clause, which
only increases it. This proves the claimed bound.
\end{proof}

\begin{proof}[Proof of \Cref{lem-rqequalaj}]
  
Recall that $\zeta$ is a hypergraph on vertex set $[s]$, and that the
hyperedge indexed by a retained label $\Gamma$ has vertex set $J_\Gamma$.
Suppose that $\Gamma\in\mathcal D_t$ and $\zeta_\Gamma=1$. By the choice
of $a_t$, there is some $i\in\mathcal R_t$ such that
$B_i\cap\vbl(a_t)\ne\emptyset$. Since $a_t\in U_\Gamma$, the definition
of $J_\Gamma$ in \eqref{eq:core-geometric-scope} gives
$i\in J_\Gamma$. Thus the retained hyperedge $J_\Gamma$ meets
$\mathcal R_t$, and adding its vertices preserves connection to $j$.
It follows inductively that every vertex in $\mathcal R_Q$ is connected
to $j$ in the hypergraph $\zeta$, and hence $\mathcal R_Q\subseteq\mathcal A_j$.

Conversely, consider any hyperedge $J_\Gamma$ of $\zeta$ meeting
$\mathcal R_Q$. By the definition of $J_\Gamma$, some clause
$a\in U_\Gamma$ contains a variable from a block $B_i$ with
$i\in\mathcal R_Q$. Before the procedure stops, every such good clause
is used. When $a$ is used, the label $\Gamma$ is checked, unless it was
checked earlier; since $\zeta_\Gamma=1$, all vertices of $J_\Gamma$ are
therefore added to $\mathcal R_Q$. Thus every retained hyperedge that
meets $\mathcal R_Q$ is entirely contained in $\mathcal R_Q$.
Consequently, no path in $\zeta$ can connect a vertex of $\mathcal R_Q$
to a vertex outside $\mathcal R_Q$. Since $j\in\mathcal R_Q$, this proves
$\mathcal A_j\subseteq\mathcal R_Q$, and therefore
$\mathcal R_Q=\mathcal A_j$. 
\end{proof}

\begin{proof}[Proof of \Cref{lem:core-exploration-mgf}]
Let $\mathscr F_t$ be the information revealed before step $t$.
For $t<Q$, the set $\mathcal D_t$ is $\mathscr F_t$-measurable and
\begin{equation}
 \delta T_t\triangleq T_{t+1}-T_t
   =\sum_{\Gamma\in\mathcal D_t}K_\Gamma\zeta_\Gamma.
 \label{eq:core-exploration-increment}
\end{equation}
The indicators with labels in $\mathcal D_t$ have not been exposed.
Conditionally on $\mathscr F_t$, they remain independent with the
Bernoulli probabilities $p_\Gamma=1-e^{-2|a_\Gamma|}$ of
\eqref{eq:core-zeta-law}. Since $p_\Gamma\le2|a_\Gamma|$ and
$1+u\le e^u$, we have
\begin{align*}
 \mathbb E[e^{\delta T_t}\mid\mathscr F_t]
 &=\prod_{\Gamma\in\mathcal D_t}
    \left(1+p_\Gamma(e^{K_\Gamma}-1)\right)\\
 &\le\exp\left(2\sum_{\Gamma\in\mathcal D_t}
          |a_\Gamma|(e^{K_\Gamma}-1)\right).
\end{align*}
Every checked label has $K_\Gamma\ge2$, by
\eqref{eq:core-scope-size}. As $\tau\ge1$, for each integer $K\ge2$
we have
\[
 e^K-1\le e^K
 =e^{\tau K}e^{-(\tau-1)K}
 \le e^2e^{-2\tau}e^{\tau K}.
\]
Moreover, $\mathcal D_t\subseteq\{\Gamma:a_t\in U_\Gamma\}$.
The cluster-root bound \eqref{eq:core-cluster-root} therefore yields
\[
 \sum_{\Gamma\in\mathcal D_t}|a_\Gamma|(e^{K_\Gamma}-1)
 \le e^2e^{-2\tau}\rho=e^2\rho_0.
\]
Since $2e^2<15$, it follows that
\begin{equation}
 \mathbb E[e^{\delta T_t-15\rho_0}\mid\mathscr F_t]\le1,
 \qquad t<Q.
 \label{eq:core-one-step-supermartingale}
\end{equation}
Define the stopped process
\begin{equation}
 M_t=\exp\left(T_{t\wedge Q}-15\rho_0(t\wedge Q)\right),
 \qquad t\ge0.
 \label{eq:core-stopped-supermartingale}
\end{equation}
On $\{t<Q\}$, equation \eqref{eq:core-one-step-supermartingale}
gives $\mathbb E[M_{t+1}\mid\mathscr F_t]\le M_t$; on
$\{t\ge Q\}$, the process is constant. Thus $M_t$ is a nonnegative
supermartingale. No clause is queried twice, so $Q\le m$, and hence
\[
 \mathbb E e^{T-15\rho_0Q}=\mathbb E M_m\le M_0=1.
\]
This proves the lemma without an unbounded stopping-time argument.
\end{proof}

\begin{proof}[Proof of \Cref{lem-wleksquareovertwodeltaiw}]
 By
\eqref{eq:core-component-growth-input}, every $v\in W\cap\V^{\high}$ satisfies
$\deg(v)\ge\Delta$. Hence,
\[
 \Delta|W\cap\V^{\high}|
 \le\sum_{v\in W\cap\V^{\high}}\deg(v).
\]
By the definition of $I(W)$ in \eqref{eq:incident-clause-count}, exactly
$I(W)$ clauses contain at least one variable of $W$. Each such clause
contains at most $k$ variables and therefore contributes at most $k$ to
the degree sum. Hence
\[
 \sum_{v\in W\cap\V^{\high}}\deg(v)\le kI(W).
\]
Therefore, 
\[
  |W\cap\V^{\high}|\le\frac{k}{\Delta}I(W).
\]
Combining this with 
\eqref{eq:core-component-growth-input} gives
\[
 |W|\le\frac{k}{2}|W\cap\V^{\high}|
 \le\frac{k^2}{2\Delta}I(W).
\]
The lemma is proved.
\end{proof}

\begin{proof}[Proof of \Cref{lem-winlq-wplust}]
Recall that $\H_\Psi$ is the dependency graph of the residual formula
$\Psi$ on the clause labels $\C^{\good}$, where two clause labels are
adjacent exactly when their residual clauses share a good variable.
Recall also that $\G_\Phi$ is the primal graph of the original formula
$\Phi$: its vertices are the variables of $\Phi$, and two variables are
adjacent when they occur together in an original clause.

Fix
$\Gamma\in\mathcal E_Q$ with $\zeta_\Gamma=1$. Recall that
$\Gamma=(\gamma_1,\ldots,\gamma_\ell)$ is a tuple of polymers, each
$\gamma_i$ is a connected set of clause labels in $\H_\Psi$, and
$U_\Gamma=\bigcup_{i=1}^{\ell}\gamma_i$ is itself a set of clause labels. 
Since $\zeta_\Gamma=1$,
\eqref{eq:core-zeta-law} forces $a_\Gamma\ne0$: indeed, if
$a_\Gamma=0$, then
$\mathbb P_\zeta(\zeta_\Gamma=1)=1-e^{-2|a_\Gamma|}=0$.
By \Cref{lem-core-cluster-connectivity}, the induced clause graph
$\H_\Psi[U_\Gamma]$ is connected.

Choose a spanning tree $T_\Gamma$ of $\H_\Psi[U_\Gamma]$. For each edge
$e=\{c,d\}$ of $T_\Gamma$, choose a good variable $x_e$ appearing in
both original clauses $C_c$ and $C_d$; such a variable exists by the
definition of $\H_\Psi$. Define
\begin{equation}
 L_\Gamma=
 \left(\bigcup_{c\in U_\Gamma}
       \bigl(\vbl(c)\cap\V^{\bad}\bigr)\right)
 \cup\{x_e:e\in E(T_\Gamma)\}
 \cup W_{J_\Gamma}.
 \label{eq:core-label-variable-set}
\end{equation}
For each $c\in U_\Gamma$, the connector variables $x_e$ on tree edges
incident to $c$, together with the bad variables of $C_c$, lie in the
same clause $C_c$ and hence are mutually adjacent in $\G_\Phi$.
Here $|U_\Gamma|\ge|J_\Gamma|\ge2$ by
\eqref{eq:core-scope-size}, so every vertex $c$ of the spanning tree
is incident to a tree edge. Thus even a purely good clause, which has
no bad variables, contains at least one chosen connector variable.
The nonempty variable sets attached to adjacent clauses share their
edge's connector. Since $T_\Gamma$ is connected, the union of the
first two sets in \eqref{eq:core-label-variable-set} is connected.

For every $i\in J_\Gamma$, \eqref{eq:core-geometric-scope} gives
$B_i\cap\vbl(U_\Gamma)\ne\emptyset$. Thus $B_i$ intersects the first
set in \eqref{eq:core-label-variable-set}. Since $B_i$ is itself
connected in $\G_\Phi$, adjoining the blocks
$\{B_i:i\in J_\Gamma\}$ to the union of the first two sets in
\eqref{eq:core-label-variable-set} preserves connectivity. By
\eqref{eq:core-block-union}, the resulting set
is exactly $L_\Gamma$.
Hence, $L_\Gamma$ is connected. 
Moreover, the label $\Gamma$
belongs to a unique $\mathcal D_t$. By
\eqref{eq:core-exploration-new-labels}
and the choice of $a_t$, we have $a_t\in U_\Gamma$. 
Furthermore, by the
exploration rule, $a_t$ intersects some currently discovered block
$B_i$ with $i\in\mathcal R_t$, and hence
$\vbl(a_t)\cap W_{\mathcal R_t}\ne\emptyset$ by
\eqref{eq:core-block-union}.
Since $a_t\in U_\Gamma$ and
$W_{\mathcal R_t}\subseteq\V^{\bad}$, we have
\[
 \emptyset\ne\vbl(a_t)\cap W_{\mathcal R_t}
 \subseteq W_{\mathcal R_t}\cap
 \bigcup_{c\in U_\Gamma}\bigl(\vbl(c)\cap\V^{\bad}\bigr).
\]
Combined with
\eqref{eq:core-label-variable-set}, we have
\begin{equation}
 L_\Gamma\cap W_{\mathcal R_t}\ne\emptyset.
 \label{eq:core-label-set-meets-current-blocks}
\end{equation}

Set $L_0=B_j$ and, for
$0\le t<Q$, define
\begin{align}\label{eq-def-ltplusone}
L_{t+1}=L_t\cup
 \bigcup_{\substack{\Gamma\in\mathcal D_t\\
                     \zeta_\Gamma=1}}L_\Gamma.  
\end{align}
We prove by induction that $L_t$ is connected and contains
$W_{\mathcal R_t}$. At $t=0$, we have
$L_0=B_j=W_{\mathcal R_0}$ because $\mathcal R_0=\{j\}$. Now suppose
that $L_t$ is connected and $W_{\mathcal R_t}\subseteq L_t$. Then every
set adjoined in the definition of $L_{t+1}$ intersects
$W_{\mathcal R_t}$ by
\eqref{eq:core-label-set-meets-current-blocks}, so $L_{t+1}$ is
connected because $W_{\mathcal R_t}\subseteq L_t$. 
By \eqref{eq:core-label-variable-set}, we have 
\begin{align}\label{eq-wjgammasubseteqlgamma}
\forall\  \Gamma,\qquad    W_{J_\Gamma}
 \subseteq L_\Gamma. 
\end{align}
Moreover, by the
update rule for $\mathcal R_{t+1}$,
\[
 W_{\mathcal R_{t+1}}
 =W_{\mathcal R_t}\cup
   \bigcup_{\substack{\Gamma\in\mathcal D_t\\
                       \zeta_\Gamma=1}}W_{J_\Gamma}.
\]
Combined with \eqref{eq-def-ltplusone}, \eqref{eq-wjgammasubseteqlgamma}, and $W_{\mathcal R_t}\subseteq L_t$,
we obtain
\[
 W_{\mathcal R_{t+1}}
 =W_{\mathcal R_t}\cup
   \bigcup_{\substack{\Gamma\in\mathcal D_t\\
                       \zeta_\Gamma=1}}W_{J_\Gamma}
 \subseteq L_t\cup
 \bigcup_{\substack{\Gamma\in\mathcal D_t\\
                     \zeta_\Gamma=1}}L_\Gamma = L_{t+1}.
\]
Thus $L_Q$ is connected and contains
$W_{\mathcal R_Q}=W_{\mathcal A_j}=W$.

Since $T_\Gamma$ has $|U_\Gamma|-1$ edges, it requires at most
$|U_\Gamma|-1$ connector variables. By
\eqref{eq:core-cluster-data},
$U_\Gamma=\bigcup_{i=1}^{\ell}\gamma_i$ and
$K_\Gamma=\sum_{i=1}^{\ell}|\gamma_i|$, and therefore
$|U_\Gamma|\le K_\Gamma$. Hence the total number of variables added to
$W$ is at most
\[
 \sum_{\Gamma\in\mathcal E_Q}K_\Gamma\zeta_\Gamma=T,
\]
which proves the claimed bound $|W|+T$.
\end{proof}

\subsection{Proof for Bernoulli factory and reciprocal test} \label{sec:reciprocal-test}
Here we state the proof of the reciprocal test which follows \cite[Lemma 6.5]{LWYZZ26}.
\begin{proof}[Proof of \Cref{lem:reciprocal-test}]
  At beginning, the auxiliary procedure $\Coin(\O)$ has success probability $u = 2(1 - b)/(2 - b) \le 2/11$ and uses $2/(2 - b) \le 2$ oracle calls in expectation.
  Just for analysis, suppress the main procedure's return-one tests and leave the auxiliary procedure unchanged.
  The walk decreases $i$ by $1$ with probability $1 - v$ and increases $(i, j)$ by $(2, 1)$ with probability $v$ where
  \begin{align*}
    v = \frac{u^2}{1 - u + u^2} \le \frac{4}{103}, \quad 1 - 3v \ge \frac{91}{103} > 0.
  \end{align*}
  Let $H$ be the number of counter changes until the embedded walk hits zero from $i$. Its drift before $H$ is $-(1-3v)$. Stopping after $H\wedge t$ changes gives $\E{i_{H\wedge t}}=i-(1-3v) \E{H\wedge t}\ge0$, where $i_s$ denotes the counter after $s$ changes. Letting $t\to\infty$ yields $\E{H}\le i/(1-3v)$ and almost-sure termination.
  Let $G(z)$ be the probability generating function of the number of increases when the process starts at $i = 1$.
  A similar argument to the proof of \cite[Lemma 6.5]{LWYZZ26} gives that
  \begin{align*}
    G(z) = 1 - v + v z G(z)^3.
  \end{align*}
  Let $N$ be the final value of $j$ in the suppressed process.
  Then $\E{z^{N - 1}} = (1 - u)G(z) + uz G(z)^3$.
  Solving the equation provides that
  \begin{align*}
    \int_0^1 G(z) \d{z} = 1 - \frac{v}{2(1 - v)}, \qquad \int_0^1 z G(z)^3 \d{z} = \frac{1 - 2v}{2(1 - v)}.
  \end{align*}
  Conditional on $N$, the probability that all return-one tests fail is $\frac{1}{2} \prod_{j = 2}^N (1 - 1/j) = 1/(2N)$.
  Therefore, the probability of $\Reciprocal(\O)$ returning zero is
  \begin{align*}
    \Pr{\Reciprocal(\O) = 0} = \frac{1}{2} \E{\frac{1}{N}} = \frac{1}{2} \E{\int_0^1 z^{N - 1} \d{z}}.
  \end{align*}
  By Tonelli's theorem,
  \begin{align*}
    \Pr{\Reciprocal(\O) = 0} &= \frac{1}{2} \E{\int_0^1 z^{N - 1} \d{z}} \\
    &= \frac{1}{2} \int_0^1 \E{z^{N - 1}} \d{z} \\
    &= \frac{1}{2} \int_0^1 \ab((1 - u)G(z) + uz G(z)^3) \d{z} \\
    &= 1 - \frac{1}{2b}.
  \end{align*}

  To bound the oracle cost, let $A$ be the expected number of auxiliary calls needed to obtain one counter change.
  An iteration uses one auxiliary call plus a second call with probability $u$.
  It restarts without a counter change with probability $u(1 - u)$.
  Thus
  \begin{align*}
    A = (1 + u) + u(1 - u)A, \quad \text{and thus} \quad A = \frac{1 + u}{1 - u + u^2} \le \frac{13}{9}.
  \end{align*}
  After initialization, the starting value $I_0$ of $i$ has mean $\E{I_0} = 1 + 2u \le 15/11$.
  Ignoring every early return after the initial fair test, the expected number of auxiliary calls is bounded by
  \begin{align*}
    \frac{1}{2} \ab(1 + \E{I_0} \cdot A \cdot \frac{1}{1 - 3v}) \le \frac{1}{2} \ab(1 + \frac{15}{11} \frac{13}{9} \frac{103}{91}).
  \end{align*}
  Since each auxiliary call uses at most two oracle calls in expectation,
  \begin{align*}
    \E{\text{oracle calls}} \le \frac{1}{2}\ab(1 + \frac{15}{11} \frac{13}{9} \frac{103}{91}) \cdot 2 \le 4.
  \end{align*}

  It remains to justify the $O(1)$ expected additional work.
  The additional work includes counter arithmetic and the tests of probability $1/j$.
  Let $M$ be the total number of embedded counter changes in the suppressed path.
  For one branching-process tree truncated at depth $r$, let $F_r \defeq \E{(3/2)^{T_r}}$ where $T_r$ is its total size.
  Then $F_0 = 3/2$ and
  \begin{align*}
    F_{r + 1} = \frac{3}{2}\ab(1 - v + v F_r^3).
  \end{align*}
  A simple induction on $r$ gives that $F_r \le 2$.
  Conditional on the initial counter $I_0\le3$, the suppressed walk explores a forest of $I_0$ independent trees: each change processes one vertex and produces either zero or three children. Its number $M$ of changes is their total size. Independence at depth $r$ and monotone convergence therefore give
  \begin{align*}
    \E{(3/2)^M\mid I_0}=\lim_{r\to\infty}F_r^{I_0}\le8.
  \end{align*}
  The number of raw iterations between counter changes is geometric with parameter $1 - u + u^2 \ge 9/11$.
  Conditional on $M$, the expected number of raw iterations is at most $M/(1 - u + u^2)$ and throughout the path,
  \begin{align*}
    i \le 3 + 2M, \qquad j \le 2 + M.
  \end{align*}
  Therefore, counter operations cost at most $O(\log_2(M + 4))$ bit operations each.
  A test of probability $1/j$ can be implemented by drawing a uniform integer from $\set{0, \ldots, j - 1}$ using binary rejection sampling, which costs at most $O(\log_2(j + 1))$ expected bit operations.
  Hence the additional work costs in expectation at most
  \begin{align*}
    \E{(M + 1) \log_2(M + 4)} = O(1)
  \end{align*}
  bit operations.
  Putting everything together concludes the Bernoulli factory.
\end{proof}

\end{document}